\documentclass[envcountsame]{llncs}

\newif\ifignore 
\ignorefalse

\newcommand{\auxproof}[1]{
\ifignore\mbox{}\newline
\textbf{PROOF:} \dotfill\newline
{\it #1}\mbox{}\newline
\textbf{ENDPROOF}\dotfill
\fi}

\newcommand{\ignore}[1]{}

\usepackage{hyperref}
\usepackage{times}
\usepackage{amssymb}
\usepackage{mathtools}
\usepackage{shadow}
\usepackage{pdfrender}
\usepackage{enumitem}
\usepackage{nicefrac}
\usepackage{xspace}

\usepackage{xypic}
\usepackage[all]{xy}
\xyoption{tips}
\SelectTips{cm}{}  

\newcommand*{\fatten}[1][.4pt]{%
  \textpdfrender{
    TextRenderingMode=FillStroke,
    LineWidth={\dimexpr(#1)\relax},
  }%
}

\ifdefined\mathsl\else
  \DeclareMathAlphabet{\mathsl}{\encodingdefault}{\rmdefault}{\mddefault}{\sldefault}
  \SetMathAlphabet{\mathsl}{bold}{\encodingdefault}{\rmdefault}{\bfdefault}{\sldefault}
\fi

\makeatletter
\newcommand{\mathoverlap}[2]{\mathpalette\mathoverlap@{{#1}{#2}}}
\newcommand{\mathoverlap@}[2]{\mathoverlap@@{#1}#2}
\newcommand{\mathoverlap@@}[3]{\ooalign{$\m@th#1#2$\crcr\hidewidth$\m@th#1#3$\hidewidth}}
\makeatother

\def\QEDbox{\def\sboxsep{.3em}\def\sdim{.1em}\raisebox{.3em}{\shabox{}}}
\def\QED{\hfill\QEDbox}

\newcommand{\idmap}[1][]{\ensuremath{\mathsl{id}_{#1}}}
\newcommand{\after}{\mathrel{\circ}}

\newcommand{\NNO}{\mathbb{N}}

\newcommand{\R}{\mathbb{R}}
\newcommand{\pR}{\mathbb{R}_{>0}}
\newcommand{\nnR}{\mathbb{R}_{\geq 0}}

\newcommand{\Mlt}{\ensuremath{\mathcal{M}}}
\newcommand{\natMlt}{\Mlt}
\newcommand{\Dst}{\ensuremath{\mathcal{D}}}
\newcommand{\Giry}{\ensuremath{\mathcal{G}}}
\newcommand{\Pow}{\ensuremath{\mathcal{P}}}
\newcommand{\supp}{\ensuremath{\mathsl{supp}}}

\newcommand{\Prob}{\ensuremath{\mathsl{P}}}
\newcommand{\no}[1]{#1^{\scriptscriptstyle \bot}} 
\newcommand{\DKL}{\ensuremath{\mathsl{D}_{\mathsl{KL}}}}

\newcommand{\Fact}{\ensuremath{\mathsl{Fact}}}

\newcommand{\Pred}{\ensuremath{\mathsl{Pred}}}

\newcommand{\Obs}{\ensuremath{\mathsl{Obs}}}

\newcommand{\concat}{\ensuremath{\mathbin{+{\kern-.5ex}+}}}
\newcommand{\one}{\ensuremath{\mathbf{1}}}
\newcommand{\zero}{\ensuremath{\mathbf{0}}}
\newcommand{\andthen}{\mathrel{\&}}
\newcommand{\bigandthen}{\mathop{\textnormal{\Large\&}}}
\newcommand{\expec}{\mathop{\mathbb{E}}}
\newcommand{\evidand}{\ensuremath{\mathop{\andthen}}}
\newcommand{\evidtensor}{\ensuremath{\mathop{\otimes}}}
\newcommand{\flrn}{\ensuremath{\mathsl{Flrn}}}
\newcommand{\indic}[1]{\mathbf{1}_{#1}}
\newcommand{\Pmodels}{\ensuremath{\mathbin{\smash{\models_{\hspace*{-1.2ex}{\raisebox{-0.4ex}{\tiny P}}}}\,}}}
\newcommand{\Jmodels}{\ensuremath{\mathbin{\smash{\models_{\hspace*{-1.2ex}{\raisebox{-0.4ex}{\tiny J}}}}\,}}}

\newcommand{\update}[1]{\ensuremath{\raisebox{-0.2ex}{$|$}\raisebox{-0.7ex}{$\scriptstyle#1$}}\xspace}
\newcommand{\JupdateSign}{\raisebox{-1.2ex}{$\scalebox{0.9}{\begin{tikzpicture}
\draw (1,0) -- (1,0.75);
\draw (1,0) arc (180:0:-0.08cm);
\end{tikzpicture}}$}}
\newcommand{\Jupdate}[1]{\ensuremath{{\kern-0.2ex}\JupdateSign{\kern0.2ex}\raisebox{-0.7ex}{$\scriptstyle#1$}}\xspace}
\newcommand{\PupdateSign}{\raisebox{-1.4ex}{$\scalebox{1.0}{\begin{tikzpicture}
\draw (1,0) -- (1,0.75);
\draw (1,0.67) arc (320:40:-0.07cm);
\end{tikzpicture}}$}}
\newcommand{\Pupdate}[1]{\ensuremath{{\kern0.3ex}\PupdateSign{\kern-0.2ex}\raisebox{-0.7ex}{$\scriptstyle#1$}}\xspace}
\newcommand{\FupdateSign}{\raisebox{-1.4ex}{$\scalebox{1.0}{\begin{tikzpicture}
\draw (1,0) -- (1,0.75);
\draw (1,0.70) -- (1.2,0.70);
\draw (1,0.6) -- (1.15,0.6);
\end{tikzpicture}}$}}
\newcommand{\Fupdate}[1]{\ensuremath{{\kern0.3ex}\FupdateSign{\kern-0.25ex}\raisebox{-0.7ex}{$\scriptstyle#1$}}\xspace}

\DeclareMathOperator*{\argmax}{argmax} 
\DeclareMathOperator*{\argmin}{argmin} 
\DeclareMathOperator*{\NormOp}{Norm} 
\newcommand{\acc}{\ensuremath{\mathsl{acc}}}

\newcommand{\bibinom}[2]{\left({\kern-.5ex}\binom{#1}{#2}{\kern-.5ex}\right)}
\newcommand{\mulnom}{\ensuremath{\mathsl{mn}}}
\newcommand{\multinomial}[1][]{\ensuremath{\mulnom[#1]}}

\newcommand{\tuple}[1]{\langle#1\rangle}

\newcommand{\setin}[3]{\{#1\in#2\;|\;#3\}}

\newcommand{\ketstrut}{\vrule height 8.5pt depth 4.5pt width 0pt}
\newcommand{\ket}[1]{\ensuremath{|{\kern.1em}#1{\kern.1em}\rangle}}
\newcommand{\bigket}[1]{\ensuremath{\big|{\kern.1em}#1{\kern.1em}\big\rangle}}
\newcommand{\Bigket}[1]{\ensuremath{\left|\ketstrut{\kern.1em}#1{\kern.005em}\right>}}
\newcommand{\coefm}[1]{\ensuremath{\fatten[0.6pt]{(}{\kern1pt}#1{\kern1pt}\fatten[0.6pt]{)}}}
\newcommand{\facto}[1]{\ensuremath{#1{\kern-2.5pt}\raisebox{-2.5pt}{\includegraphics[width=0.9em]{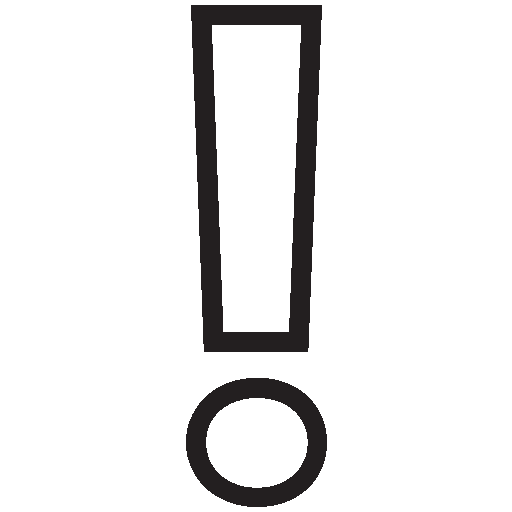}}{\kern-.2em}}}

\newcommand{\pull}{\mathrel{\mathchoice%
   {\scalebox{-0.5}[1]{$\gg=$}}
   {\scalebox{-0.5}[1]{$\gg{\kern-1.5ex}=$}}
   {\scalebox{-0.5}[1]{${\kern.5ex}\scriptstyle\gg{\kern-0.2ex}={\kern.5ex}$}}
   {\scalebox{-0.5}[1]{$\scriptscriptstyle\gg=$}}}}
\newcommand{\triplepull}{\mathrel{\mathchoice%
   {\scalebox{-0.5}[1]{$\ggg{\kern-1.5ex}=$}}
   {\scalebox{-0.5}[1]{$\ggg{\kern-1.5ex}=$}}
   {\scalebox{-0.5}[1]{${\kern.5ex}\scriptstyle\ggg{\kern-0.2ex}={\kern.5ex}$}}
   {\scalebox{-0.5}[1]{$\scriptscriptstyle\ggg=$}}}}
\newcommand{\push}{\mathrel{\mathchoice%
   {\scalebox{-0.5}[1]{$=\ll$}}
   {\scalebox{-0.5}[1]{$={\kern-1.5ex}\ll$}}
   {\scalebox{-0.5}[1]{${\kern.5ex}\scriptstyle={\kern-0.2ex}\ll{\kern.5ex}$}}
   {\scalebox{-0.5}[1]{$\scriptscriptstyle=\ll$}}}}

\DeclareSymbolFont{T1op}{T1}{cmr}{m}{n}
\SetSymbolFont{T1op}{bold}{T1}{cmr}{bx}{n}
\DeclareMathSymbol{\mathguilsinglleft}{\mathopen}{T1op}{'016}
\DeclareMathSymbol{\mathguilsinglright}{\mathclose}{T1op}{'017}

\usepackage{tikzit}

\newsavebox\sbpto
\savebox\sbpto{\begin{tikzpicture}[baseline=-2.5pt]
            \filldraw[fill=white,draw=white] circle (1.4pt);
            \filldraw[fill=white,draw=black,line width=0.2pt]circle(2pt);
                \end{tikzpicture}}
\newcommand\chanto{\mathrel{\ooalign{$\to$\cr
      \hfil\!$\usebox\sbpto$\hfil\cr}}}            

\newcommand{\postest}{\mathsl{pt}}
\newcommand{\negtest}{\mathsl{nt}}
\newcommand{\threetest}{\mathsl{ppnt}}

\begin{document}

\begin{frontmatter}

\title{Getting Wiser from Multiple Data: \\ Probabilistic Updating
  according to Jeffrey and Pearl}

\author{Bart Jacobs}

\institute{iHub, Radboud University, Nijmegen, The Netherlands
\\ 
\email{bart@cs.ru.nl}
\\
\today
}

\maketitle

\begin{abstract} 
In probabilistic updating one transforms a prior distribution in the
light of given evidence into a posterior distribution, via what is
called conditioning, updating, belief revision or inference. This is
the essence of learning, as Bayesian updating. It will be illustrated
via a physical model involving (adapted) water flows through pipes
with different diameters.

Bayesian updating makes us wiser, in the sense that the posterior
distribution makes the evidence more likely than the prior, since it
incorporates the evidence.  Things are less clear when one wishes to
learn from multiple pieces of evidence / data. It turns out that there
are (at least) two forms of updating for this, associated with Jeffrey
and Pearl. The difference is not always clearly recognised.

This paper provides an introduction and an overview in the setting of
discrete probability theory. It starts from an elementary question,
involving multiple pieces of evidence, that has been sent to a small
group academic specialists. Their answers show considerable
differences. This is used as motivation and starting point to
introduce the two forms of updating, of Jeffrey and Pearl, for
multiple inputs and to elaborate their properties. In the end the
account is related to so-called variational free energy (VFE) update
in the cognitive theory of predictive processing. It is shown that
both Jeffrey and Pearl outperform VFE updating and that VFE updating
need not decrease divergence --- that is correct errors --- as it is
supposed to do.
\end{abstract}


\end{frontmatter}

\renewcommand{\thepage}{\arabic{page}}

\section{Introduction}\label{IntroSec}

Before introducing the topic of this paper, readers are kindly invited
to consider the following situation and answer the subsequent three
question themselves. This self-test will be useful for understanding
what this paper is about. Impatient readers are invited to answer at
least the first question.

\noindent\makebox[\linewidth]{\rule{\textwidth}{1pt}} 

Consider a disease, like Covid, say with a \textbf{prevalence} of
$5\%$. This means that the chance that an arbitrary person in the
population has the disease is $\frac{1}{20} = 0.05$. This is the
\textbf{prior} disease probability.

\quad There is a test for the disease that is not perfect, as usual.
\begin{itemize}
\item The \textbf{sensitivity} of the test is $90\%$; this means that
  if a person has the disease, the probability that the test is
  positive (for this person) is $\frac{9}{10} = 0.9$.

\item The \textbf{specificity} is $60\%$; this means that if a person
  does not have the disease, then the probability that the test is
  negative is $\frac{3}{5} = 0.6$.
\end{itemize}

\noindent In this situation the predicted positive test probability
is $\frac{17}{40} = 0.425$.


\subsection*{The questions}

In the situation described above \textbf{three tests} are performed,
out of which two turn out to be positive and one negative.

\quad Consider the following three questions.
\begin{enumerate}
\item What is the likelihood of this three-test outcome, with two
  positive, one negative?



\item Using this three-test outcome as evidence, what is the posterior
  (updated) disease probability?



\item This posterior disease probability is now taken as new prior.
  Again three tests are performed, and again two of them appear to be
  positive and one negative. How likely is that outcome --- with this
  new prior, and with the same test? This is the same calculation as
  in point~1, but now with a different prior.



\end{enumerate}

\noindent You may ask yourself the question if the probability outcome
of point~3 should be lower or higher than the outcome of point~1? (In
this example the differences are small, but significant.)

\noindent\makebox[\linewidth]{\rule{\textwidth}{1pt}} 

\smallskip

\noindent Please take your time.

This example with questions has been distributed among about a hundred
(local, academic) colleagues of the author, working in AI and
(medical) statistics. The questions came with the following
explanations\footnote{More explanations were provided, including a
clarification that the answers would be used for a publication (this
one) in an anonymous manner, without identifying individuals, but with
an indication of their area. Participants could also indicate that
they did not allow the usage of their answers, but no-one did so. The
full text of the question is available via Mastodon at:
\href{https://social.edu.nl/@bjacobs/112051242549499170}{https://social.edu.nl/@bjacobs/112051242549499170}.}.
\begin{quote}
If the above description appears unclear or incomplete to you --- for
instance whether or not all tests are applied to the same individual
--- please do not ask for clarification, but make explicit how you
choose to interprete or complete it.  This is relevant information in
itself. 
\end{quote}

\noindent This small survey was not meant as a systematic study, but
was intended to get an impression of how specialists address such a
situation, with multiple test results. Seventeen replies came back,
before a given deadline, which will be discussed later in
Subsection~\ref{ReplySubsec}. The short summary is that the first
question is anwered in a reasonably systematic manner, yielding two
groups of answers that are labelled as `Jeffrey' and `Pearl'
below. Things are less clear with the second question about
updating. Several people use the `Pearl' update, but no-one uses the
`Jeffrey' update. Some people use `Jeffrey' for the first question and
`Pearl' for the second.

The fact that there is limited consensus among academic specialists,
even in a small survey, serves as motivation to try and clarify the
situation.  At the same time this lack of consensus is worrying, since
the above question describes a simple scenario of learning from
multiple data that is common in the current AI revolution. Whether you
or I get a mortgage or a medical treatment may be determined by such
update mechanisms. What is the rationale to use which update
mechanism, and what are the consequences? These questions have a
certain scientific and societal urgency. How are we ever going to
explain AI-based decisions if we do not understand the rules that are
used and the guarantees that they provide?

Updating (or conditioning, belief revision or inference) is one of the
most fascinating topics of probability theory. It involves going from
prior belief, in the form of a probability distribution, to an updated
posterior belief, by incorporating evidence. This is the essence of
learning. This paper describes updating as an operation on probability
distributions: given a prior distribution $\omega$ and evidence $p$,
one can form a posterior distribtution written as $\omega\update{p}$
that incorporates $p$ into $\omega$. This updating with a single piece
of evidence $p$ is the standard form of Bayesian updating. Its general
mathematical formulation appears in Section~\ref{UpdateSec}.

Since this Bayesian updating $\omega \mapsto \omega\update{p}$ is such
an elementary operation we illustrate it via a simple physical model,
see Subsection~\ref{PhysicsSubsec}. A probability distribution is
represented by a pump that pushes water through parallel pipes with
different diameters. These pipes may then be (partially) blocked,
leading to adapted flows, corresponding to a Bayesian update of the
original (prior) distribution.

A key property of learning that runs through this paper is that
learning makes us wiser. We express this via validity $\models$, in
several ways.  In elementary form, validity is the expected value
$\omega\models p$ of evidence $p$ in a distribution $\omega$. A key
property of Bayesian updating is that the evidence $p$ is `more true'
in the posterior than in the prior. This takes the form of a
fundamental inequality $\big(\omega\update{p} \models p\big) \geq
\big(\omega\models p\big)$, see
Theorem~\ref{UpdateGainThm}~\eqref{UpdateGainThmSingle} below.

Bayesian updating is a well-established technique that is used here as
a building block. This paper is about how to learn from multiple
pieces of evidence $p_{i}$, like in the above example where we have
three test outcomes as evidence. We consider three forms of learning
from multiple evidence --- initially two, later one more. We briefly
illustrate them below, for the simple situation with only two pieces
of evidence $p_1$ and $p_2$.  The three forms of updating combine
Bayesian updating in different ways.
\begin{equation}
\label{UpdateShort}
\begin{array}{rcl}
\mbox{Jeffrey update:}
& \quad &
\frac{1}{2}\cdot\omega\update{p_1} + \frac{1}{2}\cdot\omega\update{p_2}
\\[+0.2em]
\mbox{Pearl update:}
& &
\omega\update{p_{1} \andthen p_{2}}
\;=\;
\omega\update{p_{1}}\update{p_{2}}
\\[+0.2em]
\mbox{VFE update:}
& &
\omega\update{p_{1}^{\nicefrac{1}{2}} \andthen p_{2}^{\nicefrac{1}{2}}}
\;=\;
\omega\update{p_{1}^{\nicefrac{1}{2}}}\update{p_{2}^{\nicefrac{1}{2}}}.
\end{array}
\end{equation}

\noindent This first form of Jeffrey updating uses a convex
combination of (independent) Bayesian updates with the separate pieces
of evidence. The second form of updating, associated with Pearl, does
a single (dependent) Bayesian update with a conjunction $p_{1}\andthen
p_{2}$ of the two pieces of evidence. This is the same as doing the
two updates consecutively. The variational free energy (VFE) update is
similar, but uses a conjunction of the square roots of the two
pieces. This third form of updating arises in computational cognition
theory and will be discussed only at the very end, in
Section~\ref{PredictiveCodingSec}. The emphasis lies on Jeffrey and
Pearl. We include the VFE case here for comparison.

Jeffrey's update rule goes back to~\cite{Jeffrey83} (see
also~\cite{Halpern03,Shafer81}) and is reformulated
in~\cite{Jacobs19c} into a form that is close to how it is used
here. Pearl's update rule, see~\cite{Pearl88,Pearl90}, captures what
happens in Bayesian networks when evidence arrives at a node via two
different branches and is combined via conjunction $\andthen$, see
also~\cite{JacobsZ21}.

The above short descriptions~\eqref{UpdateShort} are meant to give a
first impression. Details will appear below. Associated with the first
two update mechanism there are forms of validity, written as
$\Jmodels$ and $\Pmodels$, for Jeffrey and Pearl. Each form of update
increases the associated form of validity --- making us appropriately
wiser in each case. This getting wiser may fail if we mix up the
various forms of updating and validity, as will be demonstrated.
Prevention of such mix-ups is an important reason for distinguishing
these two forms of updating and validity.

The focus of this paper is on handling multiple pieces of evidence.
Interestingly, when we have only a single piece of evidence, all three
update rules (Jeffrey, Pearl and also VFE) coincide and give the same
outcome as Bayesian updating. Differences arise when there is more
than one piece of evidence. This multiplicity of evidence will be
formalised below in terms of multisets of predicates, as pieces of
evidence. We recall that a multiset is like a subset but can contain
multiple occurrences of elements. Indeed, in learning from data /
evidence, there may be multiple occurences of the same data / evidence
item.

This paper is organised as follows. It continues the introduction
below, as already announced, with a physical explanation of (Bayesian)
updating. Section~\ref{MedicalSolutionSec} answers the above medical
question by describing the approaches of Jeffrey and Pearl in a
concrete situation. It ends in Subsection~\ref{ReplySubsec} with an
account of the answers given by colleagues to the
question. Sections~\ref{MltDstSec} and~\ref{ValiditySec} start the
more thorough mathematical analysis of the situation, by providing
background information on multisets and distributions, and on
predicates and validity. This includes validity for evidence, both for
Jeffrey and Pearl. Next, Section~\ref{UpdateSec} looks at updating of
distributions, first in basic Bayesian form, for a single piece of
evidence, then for multiple pieces of evidence, using separate update
rules of Jeffrey and Pearl. For each of them various properties are
made explicit.  Section~\ref{ChannelSec} introduces the concept of
channel, as mathemtical formalisation of conditional probabilities in
generative models. These channels make it possible to handle more
general situations, with distributions and evidence given for
different underlying sets (hidden and observable).  It is shown that
in the running medical example there is a channel involved, capturing
the sensitivity and specificity of the test. Finally,
Section~\ref{PredictiveCodingSec} sketches the context of the
cognitive theory of predictive coding, where the VFE update rule
appears. Within this context the human mind is understood as an update
engine that constantly corrects errors, in the form of reduction of
so-called KL-divergence. Interestingly, Jeffrey's rule decreases
KL-divergence, but the VFE rule does not, in general. The VFE update
rule does increase Pearl validity, but not as much as Pearl updating.
The question remains what the right form of updating is in
probabilistic cognition theory. This paper provides input for further
clarification.

This paper explains updating within the setting of finite, discrete
probability theory. This setting is rich enough to explain and
illustrate the relevant phenonema. Definitions and results extend
easily to a more general continuous setting --- technically by
changing from the distribution monad $\Dst$ to the Giry monad $\Giry$,
see \textit{e.g.}~\cite{Jacobs18c,Panangaden09}. Indeed, underneath
the current work there is deeper level of analysis of probabilistic
systems in terms of category theory and string
diagrams~\cite{Fritz20,Jacobs21g}. We do not use those mathematical
formalisms here in order not to restrict the audience (size). These
new formalisms try to improve the language that is common in
probability theory. Indeed, we see the fact that there is such a wide
variety of answers to a basic question as a sign that the current
dominant language --- that is based on assigning probabilities to
events and does not treat distributions, predicates, updating and
channels as first class citizens --- is deficient.  This area could
use a new, more formal `logic'. This paper tries to contribute to such
a more precise approach.

\subsection{A physical model of Bayesian updating}\label{PhysicsSubsec}


Updating of probability distributions plays a central role in this
paper. That's why it is important to have a good understanding of what
is going on. Here, in this introduction, we introduce a physical model
in terms of (water) flows to provide an intuition. Mathematical
formalisations will appear in the course of the paper.

Consider the picture below of a pump with one input pipe at the bottom
and three output pipes at the top. The outgoing pipes have relative
diameters, as indicated at the top. Thus, half a liter per second
emerges from the left pipe, one third from the middle pipe, and one
sixth from the pipe on the right. This represents a probability
distribution $\omega$, on the right.
\[ \vcenter{\hbox{\begin{picture}(100,100)
\thicklines
\put(0, 35){\line(1, 0){38}}
\put(62, 35){\line(1, 0){38}}
\put(38, 15){\line(0, 1){20}}
\put(62, 15){\line(0, 1){20}}
\put(0, 35){\line(0, 1){25}}
\put(100, 35){\line(0, 1){25}}
\put(0,60){\line(1,0){12}}
\put(24,60){\line(1,0){25}}
\put(58,60){\line(1,0){25}}
\put(88,60){\line(1,0){12}}
\put(12,60){\line(0, 1){20}}
\put(24,60){\line(0, 1){20}}
\put(49,60){\line(0, 1){20}}
\put(58,60){\line(0, 1){20}}
\put(83,60){\line(0, 1){20}}
\put(88,60){\line(0, 1){20}}
\put(15,5){\mbox{1 liter per second}}
\put(48,25){$\uparrow$}
\put(40,45){\mbox{pump}}
\put(16,77){$\uparrow$}
\put(51,77){$\uparrow$}
\put(83,77){$\uparrow$}
\put(15.5,90){$\frac{1}{2}$}
\put(50.5,90){$\frac{1}{3}$}
\put(82.5,90){$\frac{1}{6}$}
\end{picture}}} 
\hspace*{3em}
\begin{array}{c}
\mbox{representing}
\\
\mbox{the distribution}
\end{array}
\hspace*{3em}
\begin{array}{rcl}
\omega
& = &
\frac{1}{2}\ket{L} + \frac{1}{3}\ket{M} + \frac{1}{6}\ket{R}.
\end{array} \]

\noindent The letters $L$, $M$, $R$ refer to the left, middle and
right pipe. We write these letters inside `kets' $\ket{-}$ in order to
separate them from the corresponding fragments $\frac{1}{2}$,
$\frac{1}{3}$ and $\frac{1}{6}$ of the throughput. These kets are just
syntactic sugar.

We now turn to conditioning / updating. New information arrives and is
added to the picture, namely that the middle pipe is blocked. We
assume that the pump keeps on operating and still realises the
throughput of one liter per second (with increased pressure). The
question becomes what the left and right flows become, indicated as
$f_L$ and $f_R$ below.
\[ \vcenter{\hbox{\begin{picture}(100,100)
\thicklines
\put(0, 35){\line(1, 0){38}}
\put(62, 35){\line(1, 0){38}}
\put(38, 15){\line(0, 1){20}}
\put(62, 15){\line(0, 1){20}}
\put(0, 35){\line(0, 1){25}}
\put(100, 35){\line(0, 1){25}}
\put(0,60){\line(1,0){12}}
\put(24,60){\line(1,0){25}}
\put(58,60){\line(1,0){25}}
\put(88,60){\line(1,0){12}}
\put(12,60){\line(0, 1){20}}
\put(24,60){\line(0, 1){20}}
\put(49,60){\line(0, 1){20}}
\put(58,60){\line(0, 1){20}}
\put(83,60){\line(0, 1){20}}
\put(88,60){\line(0, 1){20}}
\put(15,5){\mbox{1 liter per second}}
\put(48,25){$\uparrow$}
\put(40,45){\mbox{pump}}
\put(16,77){$\uparrow$}
\put(83,77){$\uparrow$}
\put(43,80){\line(1, 0){20}}
\put(13,90){$f_L$}
\put(80,90){$f_R$}
\end{picture}}} \]

\noindent After a moment's thought we see that the left pipe (with
diameter $\frac{1}{2}$) is three times wider than the right pipe (with
diameter $\frac{1}{6}$). Hence the new outflows will be in
relationship $3:1$. Since the total throughput is (still) one liter
per second, the new left flow $f_{L}$ will be $\frac{3}{4}$ and the
right flow $f_{R}$ will be $\frac{1}{4}$.

In a more systematic approach one computes the new flows $f_{L}$ and
$f_{R}$ via (re)norma\-lisation. The combined diameters are
$\frac{1}{2} + \frac{1}{6} = \frac{2}{3}$. This gives:
\[ \begin{array}{lcl}
\mbox{updated left flow:}
& \quad &
f_{L}
\hspace*{\arraycolsep}=\hspace*{\arraycolsep}
\frac{\nicefrac{1}{2}}{\nicefrac{2}{3}}
\hspace*{\arraycolsep}=\hspace*{\arraycolsep}
\frac{3}{4}
\\[+0.2em]
\mbox{updated right flow:}
& &
f_{R}
\hspace*{\arraycolsep}=\hspace*{\arraycolsep}
\frac{\nicefrac{1}{6}}{\nicefrac{2}{3}}
\hspace*{\arraycolsep}=\hspace*{\arraycolsep}
\frac{1}{4}.
\end{array} \]

\noindent In ket notation, the new updated distribution $\omega'$ can
be written as $\omega' = \frac{3}{4}\ket{L} + \frac{1}{4}\ket{R} =
\frac{3}{4}\ket{L} + 0\ket{M} + \frac{1}{4}\ket{R}$.

We make the situation a bit more interesting, by not fully closing one
pipe, but by closing all of them partially via three taps, as sketched
in:
\[ \vcenter{\hbox{\begin{picture}(100,100)
\thicklines
\put(0, 35){\line(1, 0){38}}
\put(62, 35){\line(1, 0){38}}
\put(38, 15){\line(0, 1){20}}
\put(62, 15){\line(0, 1){20}}
\put(0, 35){\line(0, 1){25}}
\put(100, 35){\line(0, 1){25}}
\put(0,60){\line(1,0){12}}
\put(24,60){\line(1,0){25}}
\put(58,60){\line(1,0){25}}
\put(88,60){\line(1,0){12}}
\put(12,60){\line(0, 1){20}}
\put(24,60){\line(0, 1){20}}
\put(49,60){\line(0, 1){20}}
\put(58,60){\line(0, 1){20}}
\put(83,60){\line(0, 1){20}}
\put(88,60){\line(0, 1){20}}
\put(15,5){\mbox{1 liter per second}}
\put(48,25){$\uparrow$}
\put(40,45){\mbox{pump}}
\put(16,77){$\uparrow$}
\put(51,77){$\uparrow$}
\put(83,77){$\uparrow$}
\put(13,90){$f_L$}
\put(48,90){$f_M$}
\put(80,90){$f_R$}
\put(6,70){\line(1, 0){10}}
\put(4, 70){\oval(4, 6)}
\put(-6,66){$\frac{2}{3}$}
\put(43,70){\line(1, 0){12}}
\put(41, 70){\oval(4, 6)}
\put(32,66){$\frac{1}{3}$}
\put(79,70){\line(1, 0){7}}
\put(77, 70){\oval(4, 6)}
\put(68,66){$\frac{1}{2}$}
\end{picture}}} \]

\noindent The fractions written on the left of these taps indicate
the fraction of openness. We still assume that the pump realises the
same throughput of one liter per second\footnote{When the taps leave
only small openings the pressure must rise considerably, turbulences
may arise, and the pump may break. We ignore such physical effects and
limitations.} and ask what the new output flows $f_{L}, f_{M}$ and
$f_R$ are in this new case.

The normalisation now has to take both the original diameters and the
openness of the taps into account, resulting in a normalisation
factor:
\[ \begin{array}{rcl}
\frac{1}{2}\cdot\frac{2}{3} + \frac{1}{3}\cdot\frac{1}{3} + 
   \frac{1}{6}\cdot\frac{1}{2}
& = &
\frac{19}{36}.
\end{array} \]

\noindent The second conditioning $\omega''$ of the original distribution
$\omega$ thus takes the form:
\[ \begin{array}{rcccl}
\omega''
& = &
\displaystyle
\frac{\nicefrac{1}{2}\cdot\nicefrac{2}{3}}{\nicefrac{19}{36}}\ket{L} + 
   \frac{\nicefrac{1}{3}\cdot\nicefrac{1}{3}}{\nicefrac{19}{36}}\ket{M} + 
   \frac{\nicefrac{1}{6}\cdot\nicefrac{1}{2}}{\nicefrac{19}{36}}\ket{R}
& = &
\frac{12}{19}\ket{L} + \frac{4}{19}\ket{M} + \frac{3}{19}\ket{R}.
\end{array} \]

\noindent Thus, in the above picture we get $f_{L} = \frac{12}{19}$,
$f_{M} = \frac{4}{19}$ and $f_{R} = \frac{3}{19}$.

Conditioning will be described more abstractly below.  The first form,
with the entire closure of a pipe, leading to the `posterior'
distribution $\omega'$, is an instance of updating with a \emph{sharp}
predicate.  The second form, leading to posterior $\omega''$, is
updating with a \emph{fuzzy} or \emph{soft} predicate, see
Example~\ref{PhysicsEx} for details.

\section{Two solutions for the medical question}\label{MedicalSolutionSec}

This section introduces basic terminology, notation, and results for
distributions, predicates, validity, and updating (conditioning) in
order to answer the medical test challenge from the introduction in an
informal but systematic manner. A more general treatment follows in
subsequent sections.

Consider an urn containing 10~balls, with 5~red ($R$), 2~blue ($B$)
and 3~green ($G$). There is an associated probability distribution
that gives the chance of drawing at random a ball of a particular
colour. We write this distribution as $\frac{1}{2}\ket{R} +
\frac{1}{5}\ket{B} + \frac{3}{10}\ket{G}$. It uses ket notation
$\ket{-}$, borrowed from quantum physics, that we also used for the
pipes in Subsection~\ref{PhysicsSubsec}. The number $\frac{1}{2} =
\frac{5}{10}$ written before $\ket{R}$ is the probability of drawing a
red ball, corresponding to the fraction of red balls in the urn.
Similarly for $B$ and $G$.

In general, a (discrete finite probability) distribution over a set
$X$ is given by an expression of the form $r_{1}\ket{x_1} + \cdots +
r_{n}\ket{x_n}$ where $x_{1}, \ldots, x_{n}$ are elements from the set
$X$, and $r_{1}, \ldots, r_{n}$ are probabilities from the unit
interval $[0,1]$ satisfying $r_{1} + \cdots + r_{n} = 1$.

A predicate on a set $X$ is a function $p\colon X \rightarrow [0,1]$.
For a distribution $\omega$ over $X$ we write $\omega\models p$ for
the validity (or expected value) of the predicate $p$ in $\omega$.
Explicitly, we define this validity as the following number in
$[0,1]$.
\begin{equation}
\label{ValidityKetEqn}
\begin{array}{rcl}
r_{1}\ket{x_1} + \cdots + r_{n}\ket{x_n} \models p
& \;\coloneqq\; &
r_{1}\cdot p(x_{1}) + \cdots + r_{n}\cdot p(x_{n}).
\end{array}
\end{equation}

\noindent This validity $\omega\models p$ can also be written as
$\displaystyle\expec_{x\sim\omega} p(x)$.

If we have two predicates $p,q\colon X \rightarrow [0,1]$ we write
$p\andthen q \colon X \rightarrow [0,1]$ for their conjunction,
defined as pointwise multiplication: $\big(p \andthen q\big)(x)
\coloneqq p(x)\cdot q(x)$, for $x\in X$.

\subsection{Question 1}\label{QuestionOneSubsec}

At this stage we can already start our analysis of the medical test
scenario from the introduction. Recall the disease prevalence of
$5\%$. This is captured in a (prior) disease distribution $\omega$. In
ket notation we write it as:
\[ \begin{array}{rcl}
\omega
& \coloneqq &
\frac{1}{20}\ket{d} + \frac{19}{20}\ket{\no{d}},
\qquad\mbox{a distribution over the set }D = \{d,\no{d}\}.
\end{array} \]

\noindent The symbols $d$ and $\no{d}$ are used for `disease' and
`no-disease'. In this expression the presence of the disease $d$ gets
probability $\frac{1}{20}$ and the absence of the disease $\no{d}$
gets probability $\frac{19}{20}$.

Consider next the two predicates $\postest, \negtest \colon D
\rightarrow [0,1]$, for `positive test' and for `negative
test'. Following the descriptions of the sensitivity (of
$\frac{9}{10}$) and specificity (of $\frac{3}{5}$) of the test they
satisfy:
\[ \begin{array}{rclclcrclcl}
\postest(d)
& = &
\frac{9}{10} 
& & & \qquad\qquad &
\negtest(d)
& = &
1 - \frac{9}{10} 
& = & 
\frac{1}{10} 
\\
\postest(\no{d})
& = &
1 - \frac{3}{5}
& = &
\frac{2}{5}
& &
\negtest(\no{d})
& = &
\frac{3}{5}.
& &
\end{array} \]

\noindent We can compute the probability of a positive test in the
prior distribution $\omega = \frac{1}{20}\ket{d} +
\frac{19}{20}\ket{\no{d}}$ as:
\[ \begin{array}{rcccccccl}
\omega \models \postest
& \,=\, &
\frac{1}{20}\cdot \postest(d) + \frac{19}{20}\cdot \postest(\no{d})
& = &
\frac{1}{20}\cdot \frac{9}{10} + \frac{19}{20}\cdot \frac{2}{5}
& = &
\frac{85}{200}
& = &
\frac{17}{40}.
\end{array} \]

\noindent Similarly, the probability of a negative test is:
\[ \begin{array}{rcccccccl}
\omega \models \negtest
& \,=\, &
\frac{1}{20}\cdot \negtest(d) + \frac{19}{20}\cdot \negtest(\no{d})
& = &
\frac{1}{20}\cdot \frac{1}{10} + \frac{19}{20}\cdot \frac{3}{5}
& = &
\frac{115}{200}
& = &
\frac{23}{40}.
\end{array} \]

We can now address the first question: what is the likelihood of this
three-test outcome, with two positive, one negative? There are two
reasonable approaches, associated with the researchers Jeffrey and
Pearl.
\begin{description}
\item[Jeffrey] We have just calculated the probabilities of a single
  positive test and of a single negative test as two validities
  $\frac{17}{40}$ and $\frac{23}{40}$. The Jeffrey-style likelihood of
  the three tests, two positive one negative, is obtained by
  multiplying these validities, as in:
\begin{equation}
\label{JeffreyMedPriorValidity}
\begin{array}{rcl}
3\cdot \big(\omega \models \postest\big) \cdot
  \big(\omega \models \postest\big) \cdot
  \big(\omega \models \negtest\big)
& = &
3\cdot\frac{17}{40} \cdot \frac{17}{40} \cdot \frac{23}{40}
\\[+0.2em]
& = &
19941/64000  
\hspace*{\arraycolsep}\approx\hspace*{\arraycolsep}
0.3116.
\end{array}
\end{equation}

\noindent The multiplication factor $3$ is needed to accomodate the
three possible orders of the test outcomes, pos-pos-neg, pos-neg-pos,
neg-pos-pos, each with the same probability. One may recognise in this
expression the binomial probability for two-out-of-three, with
parameter $\big(\omega\models\postest\big) = \frac{17}{40}$.


\item[Pearl] Instead of multiplying these validities, one can multiply
  the predicates by taking their conjunction, as in $\threetest
  \coloneqq \postest \andthen \postest \andthen \negtest$. This works
  pointwise, so that:
\[ \begin{array}{rcl}
\threetest(d)
& = &
\postest(d) \cdot \postest(d) \cdot \negtest(d)
\hspace*{\arraycolsep}=\hspace*{\arraycolsep}
\frac{9}{10} \cdot \frac{9}{10} \cdot \frac{1}{10}
\hspace*{\arraycolsep}=\hspace*{\arraycolsep}
\frac{81}{1000}
\\[+0.2em]
\threetest(\no{d})
& = &
\postest(\no{d}) \cdot \postest(\no{d}) \cdot \negtest(\no{d})
\hspace*{\arraycolsep}=\hspace*{\arraycolsep}
\frac{2}{5} \cdot \frac{2}{5} \cdot \frac{3}{5}
\hspace*{\arraycolsep}=\hspace*{\arraycolsep}
\frac{12}{125}.
\end{array} \]

\noindent The Pearl-style answer to the first question is the validity
of this conjunction of predicates:
\begin{equation}
\label{PearlMedPriorValidity}
\begin{array}{rcl}
3\cdot\big(\omega \models \postest \andthen \postest \andthen \negtest\big)
\hspace*{\arraycolsep}=\hspace*{\arraycolsep}
3\cdot\big(\omega \models \threetest\big)
& = &
3\cdot\Big(\frac{1}{20}\cdot\frac{81}{1000} + 
   \frac{19}{20} \cdot \frac{12}{125}\Big)
\\
& = &
\frac{1143}{4000}
\hspace*{\arraycolsep}\approx\hspace*{\arraycolsep}
0.2858
\end{array}
\end{equation}


\noindent Again, the factor $3$ is used to deal with the three
possible orders in the conjunction. The Pearl outcome differs slightly
from Jeffrey's outcome~\eqref{JeffreyMedPriorValidity}. In general,
the differences may be substantial, see the illustrations before
Remark~\ref{VarianceRem} below.
\end{description}

Using more traditional expectation notation $\expec$, with sampling
$x\sim\omega$, one can describe the Jeffrey likelihood as on the left
below, and the Pearl likelihood as on the right.
\[ 3\cdot\Big(\expec_{x\sim\omega} \postest(x)\Big)^{2} \cdot
   \Big(\expec_{x\sim\omega} \negtest(x)\Big)
\hspace*{6em}
\displaystyle3\cdot\expec_{x\sim\omega} 
   \Big(\postest(x)^{2}\cdot\negtest(x)\Big). \]

\noindent In the Jeffrey case one evaluates each of the three tests
separately / independently, as if for a new person from the
population. In the Pearl cases one first combines the three tests via
conjunction, and then evaluates the result. This corresponds to
applying the tests to the same individual. In light of this
interpretation we associate (in this medical setting) Jeffrey's method
with an epidemiological approach, and Pearl's method with a clinical
approach.

\subsection{Question 2}\label{QuestionTwoSubsec}

The second question involves updating of distributions. We describe it
here as incorporating evidence into a distribution. The evidence has
the form of a predicate. It corresponds to the taps on the outgoing
pipes in Subsection~\ref{PhysicsSubsec}.

Thus, for a general distribution $\omega = r_{1}\ket{x_1} + \cdots +
r_{n}\ket{x_n}$ over a set $X$ and for a predicate $p\colon X
\rightarrow [0,1]$, we introduce a new, updated distribution
$\omega\update{p}$.
\begin{equation}
\label{UpdateKetEqn}
\begin{array}{rcl}
\omega\update{p}
& \;\coloneqq\; &
\displaystyle\frac{r_{1}\cdot p(x_{1})}{\omega\models p}\,\bigket{x_1} 
   + \cdots + \frac{r_{n}\cdot p(x_{n})}{\omega\models p}\,\bigket{x_n}.
\end{array}
\end{equation}

\noindent This updating only works if the validity $\omega\models p$
is non-zero.

We can now update the prior $\omega$ with the positive-test predicate
$\postest$. This gives, according to the above formula:
\begin{equation}
\label{PostestUpdateEqn}
\begin{array}{rcl}
\omega\update{\postest}
& = &
\displaystyle
   \frac{\nicefrac{1}{20}\cdot \postest(d)}{\omega\models\postest}\,\bigket{d}
   +
   \frac{\nicefrac{19}{20}\cdot \postest(\no{d})}{\omega\models\postest}
   \,\bigket{\no{d}}
\\[+1em]
& = &
\displaystyle
   \frac{\nicefrac{1}{20}\cdot\nicefrac{9}{10}}{\nicefrac{17}{40}}\,\bigket{d}
   +
   \frac{\nicefrac{19}{20}\cdot\nicefrac{2}{5}}{\nicefrac{17}{40}}
   \,\bigket{\no{d}}
\hspace*{\arraycolsep}=\hspace*{\arraycolsep}\textstyle
\frac{9}{85}\,\bigket{d} + \frac{76}{85}\,\bigket{\no{d}}.
\end{array}
\end{equation}

\noindent In a similar way one obtains as update of the prior with a
negative test:
\begin{equation}
\label{NegtestUpdateEqn}
\begin{array}{rcl}
\omega\update{\negtest}
& = &
\frac{1}{115}\,\bigket{d} + \frac{114}{115}\,\bigket{\no{d}}.
\end{array}
\end{equation}

\noindent We see, as expected, that updating with a positive test
yields a higher disease probability in~\eqref{PostestUpdateEqn}, than
the $5\%$ of the prior, and that updating with a negative test yields
a lower disease probability, in~\eqref{NegtestUpdateEqn}.

\auxproof{
\[ \begin{array}{rcl}
\omega\update{\negtest}
& = &
\displaystyle
   \frac{\nicefrac{1}{20}\cdot \negtest(d)}{\omega\models\negtest}\,\bigket{d}
   +
   \frac{\nicefrac{19}{20}\cdot \negtest(\no{d})}{\omega\models\negtest}
   \,\bigket{\no{d}}
\\[+1em]
& = &
\displaystyle
   \frac{\nicefrac{1}{20}\cdot\nicefrac{1}{10}}{\nicefrac{23}{40}}\,\bigket{d}
   +
   \frac{\nicefrac{19}{20}\cdot\nicefrac{3}{5}}{\nicefrac{23}{40}}
   \,\bigket{\no{d}}
\hspace*{\arraycolsep}=\hspace*{\arraycolsep}\textstyle
\frac{1}{115}\,\bigket{d} + \frac{114}{115}\,\bigket{\no{d}}.
\end{array} \]
}

We have seen that there are two ways to answer the first question
about likelihoods of tests. There are also two corresponding ways to
perform updating.
\begin{description}
\item[Jeffrey] The posterior $\omega_J$, given the evidence of two
  positive and one negative test, in Jeffrey-style is a weighted
  mixture, reflecting two-out-three positive plus one-out-of-three
  negative:
\begin{equation}
\label{JeffreyMedUpdate}
\begin{array}{rcl}
\omega_{J}
\hspace*{\arraycolsep}\coloneqq\hspace*{\arraycolsep}
\frac{2}{3} \cdot \omega\update{\postest} +  \frac{1}{3} \cdot \omega\update{\negtest}
& = &
\frac{431}{5865}\bigket{d} + \frac{5434}{5865}\bigket{\no{d}}
\\[0.2em]
& \approx &
0.073\bigket{d} + 0.927\bigket{\no{d}}.
\end{array}
\end{equation}

\noindent This Jeffrey-posterior $\omega_{J}$ has a slightly higher
disease probability of $7.3\%$ than the prior ($5\%$).

\item[Pearl] Recall the conjunction $\threetest = \postest \andthen
  \postest \andthen \negtest$. In Pearl-style updating one uses this
  predicate as evidence for the posterior:
\begin{equation}
\label{PearlMedUpdate}
\begin{array}{rcl}
\omega_{P}
\hspace*{\arraycolsep}\coloneqq\hspace*{\arraycolsep}
\omega\update{\postest \andthen \postest \andthen \negtest}
\hspace*{\arraycolsep}=\hspace*{\arraycolsep}
\omega\update{\threetest}
& = &
\displaystyle
   \frac{\nicefrac{1}{20}\cdot \nicefrac{81}{1000}}{\nicefrac{381}{4000}}
   \,\bigket{d}
   +
   \frac{\nicefrac{19}{20}\cdot\nicefrac{12}{125}}{\nicefrac{381}{4000}}
   \,\bigket{\no{d}}
\\[0.8em]
& = &
\frac{27}{635}\bigket{d} + \frac{608}{635}\bigket{\no{d}}
\\[0.2em]
& \approx &
0.043\bigket{d} + 0.957\bigket{\no{d}}.
\end{array}
\end{equation}

\noindent The Pearl-posterior $\omega_{P}$ has a lower disease
probability than the prior. As we shall see later, this Pearl update
$\omega\update{\threetest}$ can equivalently be computed as three
successive updates
$\omega\update{\postest}\update{\postest}\update{\negtest}$. The order
of the separate updates is irrelevant.
\end{description}

\begin{figure}
\[ \hspace*{-3em}\vcenter{\hbox{\begin{picture}(150,100)
\thicklines
\put(0, 35){\line(1, 0){57}}
\put(93, 35){\line(1, 0){57}}
\put(57, 15){\line(0, 1){20}}
\put(93, 15){\line(0, 1){20}}
\put(0, 35){\line(0, 1){25}}
\put(150, 35){\line(0, 1){25}}
\put(0,60){\line(1,0){25}}
\put(33,60){\line(1,0){60}}
\put(125,60){\line(1,0){25}}
\put(25,60){\line(0, 1){20}}
\put(33,60){\line(0, 1){20}}
\put(93,60){\line(0, 1){20}}
\put(125,60){\line(0, 1){20}}
\put(40,5){\mbox{1 liter per second}}
\put(72,25){$\uparrow$}
\put(25,45){\mbox{prior $\omega = \frac{1}{20}\ket{d} + \frac{19}{20}\ket{\no{d}}$}}
\put(26.5,77){$\uparrow$}
\put(106,77){$\uparrow$}
\put(24,90){$\frac{1}{20}$}
\put(103,90){$\frac{19}{20}$}
\end{picture}}}
\hspace*{4em}
\vcenter{\hbox{\begin{picture}(100,130)
\thicklines
\put(0, 35){\line(1, 0){57}}
\put(93, 35){\line(1, 0){57}}
\put(57, 15){\line(0, 1){20}}
\put(93, 15){\line(0, 1){20}}
\put(0, 35){\line(0, 1){25}}
\put(150, 35){\line(0, 1){25}}
\put(0,60){\line(1,0){25}}
\put(33,60){\line(1,0){60}}
\put(125,60){\line(1,0){25}}
\put(25,60){\line(0, 1){10}}
\put(27,70){\line(0, 1){15}}
\put(29,85){\line(0, 1){15}}
\put(31,100){\line(0, 1){10}}
\put(33,60){\line(0, 1){50}}
\put(93,60){\line(0, 1){10}}
\put(111,70){\line(0, 1){15}}
\put(120,85){\line(0, 1){15}}
\put(122,100){\line(0, 1){10}}
\put(125,60){\line(0, 1){50}}
\put(40,5){\mbox{1 liter per second}}
\put(72,25){$\uparrow$}
\put(19,70){\line(1, 0){8}}
\put(17, 70){\oval(4, 6)}
\put(-3,66){$\postest$}
\put(87,70){\line(1, 0){24}}
\put(85, 70){\oval(4, 6)}
\put(21,85){\line(1, 0){8}}
\put(19, 85){\oval(4, 6)}
\put(-3, 81){$\postest$}
\put(105,85){\line(1, 0){15}}
\put(103, 85){\oval(4, 6)}
\put(23,100){\line(1, 0){8}}
\put(21, 100){\oval(4, 6)}
\put(-3, 96){$\negtest$}
\put(114, 100){\line(1, 0){8}}
\put(112, 100){\oval(4, 6)}
\put(29.5,110){$\uparrow$}
\put(121,110){$\uparrow$}
\put(25,123){$\frac{27}{635}$}
\put(117,123){$\frac{608}{635}$}
\end{picture}}} \]
\\
\[ \hspace*{-3em}\vcenter{\hbox{\begin{picture}(300,165)
\thicklines
\put(0, 35){\line(1, 0){57}}
\put(93, 35){\line(1, 0){57}}
\put(57, 15){\line(0, 1){20}}
\put(93, 15){\line(0, 1){20}}
\put(0, 35){\line(0, 1){25}}
\put(150, 35){\line(0, 1){25}}
\put(0,60){\line(1,0){25}}
\put(33,60){\line(1,0){60}}
\put(125,60){\line(1,0){25}}
\put(25,60){\line(0, 1){10}}
\put(27,70){\line(0, 1){21}}
\put(33,60){\line(0, 1){25}}
\put(93,60){\line(0, 1){10}}
\put(111,70){\line(0, 1){20}}
\put(125,60){\line(0, 1){30}}
\put(111,92){\line(0, 1){24}}
\put(125,92){\line(0, 1){10}}
\put(19,70){\line(1, 0){8}}
\put(17, 70){\oval(4, 6)}
\put(-3,66){$\postest$}
\put(87,70){\line(1, 0){24}}
\put(85, 70){\oval(4, 6)}
\put(40,5){\mbox{$\frac{2}{3}$ liter per second}}
\put(72,25){$\uparrow$}
\put(60.5,142){$\uparrow$}
\put(240,142){$\uparrow$}
\put(54,155){$\frac{431}{5865}$}
\put(233,155){$\frac{5434}{5865}$}
\put(180, 35){\line(1, 0){57}}
\put(273, 35){\line(1, 0){57}}
\put(237, 15){\line(0, 1){20}}
\put(273, 15){\line(0, 1){20}}
\put(180, 35){\line(0, 1){25}}
\put(330, 35){\line(0, 1){25}}
\put(180,60){\line(1,0){25}}
\put(213,60){\line(1,0){60}}
\put(305,60){\line(1,0){25}}
\put(205,60){\line(0, 1){10}}
\put(211,70){\line(0, 1){15}}
\put(213,60){\line(0, 1){27}}
\put(273,60){\line(0, 1){10}}
\put(287,70){\line(0, 1){32}}
\put(305,60){\line(0, 1){60}}
\put(199,70){\line(1, 0){12}}
\put(197, 70){\oval(4, 6)}
\put(177,66){$\negtest$}
\put(267,70){\line(1, 0){20}}
\put(265, 70){\oval(4, 6)}
\put(220,5){\mbox{$\frac{1}{3}$ liter per second}}
\put(252,25){$\uparrow$}
\put(38, 85){\oval(10, 10)[tl]}
\put(32, 91){\oval(10, 10)[tl]}
\put(32,96){\line(1,0){25}}
\put(36,90){\line(1,0){25}}
\put(206, 85){\oval(10, 10)[tr]}
\put(208, 87){\oval(10, 10)[tr]}
\put(206,90){\line(-1,0){136}}
\put(208,92){\line(-1,0){136}}
\put(60, 95){\oval(10, 10)[br]}
\put(54, 101){\oval(10, 10)[br]}
\put(70, 95){\oval(10, 10)[bl]}
\put(72, 97){\oval(10, 10)[bl]}
\put(59,100){\line(0,1){45}}
\put(67,95){\line(0,1){50}}
\put(116, 116){\oval(10, 10)[tl]}
\put(130, 102){\oval(10, 10)[tl]}
\put(130, 107){\line(1,0){105}}
\put(116, 121){\line(1,0){105}}
\put(282, 102){\oval(10, 10)[tr]}
\put(300, 120){\oval(10, 10)[tr]}
\put(282, 107){\line(-1,0){37}}
\put(300, 125){\line(-1,0){37}}
\put(235, 112){\oval(10, 10)[br]}
\put(221, 126){\oval(10, 10)[br]}
\put(245, 112){\oval(10, 10)[bl]}
\put(263, 130){\oval(10, 10)[bl]}
\put(226, 124){\line(0,1){21}}
\put(258, 128){\line(0,1){17}}
\end{picture}}} \]
\caption{Posterior, updated disease distributions as water flows
  regulated by taps, in the style of
  Subsection~\ref{PhysicsSubsec}. The prior is on the top left, and
  the Pearl update is on the top right. It involves successive
  conditionings, via successive taps, corresponding to the conjunction
  of predicates $\postest \andthen \postest \andthen \negtest$ used
  for updating~\eqref{PearlMedUpdate}. The Jeffrey update, as a convex
  combination of conditionings, is at the bottom. The incoming flow of
  one liter per second is divided, in a convex sum of $\frac{2}{3}$
  and $\frac{1}{3}$ liter over two pumps. The positive test tap
  setting is applied to the left pump and the negative test tap
  setting is used on the right. The resulting outgoing flows are
  combined in merged pipes, corresponding to the convex
  sum~\eqref{JeffreyMedUpdate}. The diameters of the various pipes are
  not precise and only give an indication.}
\label{FlowFig}
\end{figure}
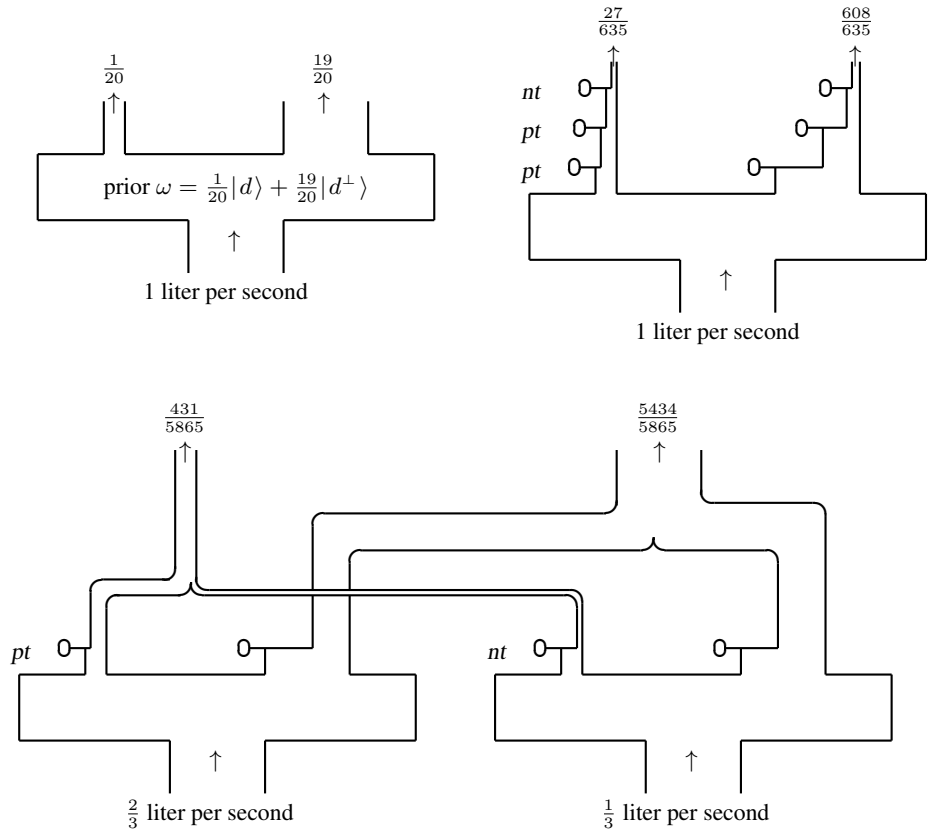

The Jeffrey and Pearl updates are represented in Figure~\ref{FlowFig}
as adapted flows, in the style of Subsection~\ref{PhysicsSubsec}. In
the remainder of this paper we shall work with the mathematical
formulations of conditioning and no longer uses such water flow
illustrations. In principle, these flows can be used to describe more
general probabilistic constructions, like product distributions and
marginalisation (see Section~\ref{ValiditySec} below). One could even
try to build such physical models for Bayesian networks, in which
evidence at some point in the network, for instance in the form of
blocking outgoing flows, propagates through the network and influences
other flows, according to the laws of forward and backward
reasoning~\cite{JacobsZ21}. One then has to reason with both water
pressure and flows to form a physical model of string
diagrams~\cite{Fritz20,ChoJ19,Selinger11}. These string diagrams form
an intuitive formalism that is increasingly used in quantum
foundations~\cite{CoeckeK16} and probability theory.

Which update outcome makes most sense? Let's reason informally. We
have seen in the formulation of the question that the predicted
positive test probability is $\frac{17}{40} = 0.425$. Thus one expects
roughly~$2$ out of~$5$ tests to be positive. The evidence that we have
shows $2$ positive out of~$3$, with a higher positive test ratio. In
light of this a higher posterior disease probability makes sense ---
as with Jeffrey and not with Pearl.

\begin{remark}
\label{PearlMedicalUpdateRem}
Now that we have seen how to update a distribution, one can reconsider
the first question and argue that one should compute the three-test
likelihood by iteratively evaluating, in a dependent manner, each of
the individual tests in an updated distribution. The outcome order
pos-pos-neg gives:
\[ \begin{array}{rcl}
\big(\omega\models\postest\big)\cdot
   \big(\omega\update{\postest}\models\postest\big)\cdot
   \big(\omega\update{\postest}\update{\postest}\models\negtest\big)
& = &
\frac{381}{4000}.
\end{array} \]

\noindent This can be done in three different orders, each time giving
the same outcome. Hence we have to multiply this outcome by three.  It
reproduces the Pearl validity that we have seen
in~\eqref{PearlMedPriorValidity} as: $3 \cdot \frac{381}{4000} =
\frac{1143}{4000}$. Later in
Lemma~\ref{PredicateUpdateLem}~\eqref{PredicateUpdateLemIter} a
general explanation is given.
\end{remark}

\subsection{Question 3}\label{QuestionThreeSubsec}

The third question is in fact the same one as the first, but now with
the Jeffrey and Pearl updates $\omega_{J}$ and $\omega_{P}$ instead of
the original prior $\omega$.
\begin{description}
\item[Jeffrey] We apply the same
  formula~\eqref{JeffreyMedPriorValidity}, but now with $\omega_{J}$
  instead of $\omega$. For this we first calculate:
\[ \begin{array}{rccccclcrcccccl}
\omega_{J} \models \postest
& \,=\, &
\frac{5123}{11730}
& \approx &
0.437
& \qquad\qquad &
\omega_{J} \models \negtest
& \,=\, &
\frac{6607}{11730}
& \approx &
0.563
\end{array} \]

\noindent The Jeffrey-likelihood of the three tests, two positive one
negative, in the Jeffrey-posterior is then:
\begin{equation}
\label{JeffreyMedPosteriorValidity}
\begin{array}{rcl}
3 \cdot \big(\omega_{J} \models \postest\big)^{2} \cdot
  \big(\omega_{J} \models \negtest\big)
& \approx &
0.322  
\end{array}
\end{equation}

\item[Pearl] We now calculate, as in~\eqref{PearlMedPriorValidity},
the Pearl-style likelihood in the Pearl-posterior:
\begin{equation}
\label{PearlMedPosteriorValidity}
\begin{array}{rcccl}
3 \cdot \big(\omega_{P} \models \threetest\big)
& = &
3 \cdot \big(\omega_{P} \models \postest^{2} \andthen \negtest\big)
& \approx &
0.2861 
\end{array}
\end{equation}
\end{description}

In the end we take a step back and compare the likelihoods of the
three tests that we have found as answers to questions~1 and~3. The
following table gives an overview.
\begin{center}
\begin{tabular}{c||c|c}
& \quad\qquad Jeffrey \mbox{\quad\qquad} & 
  \quad\qquad Pearl \mbox{\quad\qquad} 
\\
\hline\hline
prior & 0.3116 \mbox{ in Eqn.}\,\eqref{JeffreyMedPriorValidity}
   & 0.2858 \mbox{ in Eqn.}\,\eqref{PearlMedPriorValidity}
\\
\hspace*{1em} posterior \mbox{\hspace*{1em}} 
   & 0.322 \mbox{ in Eqn.}\,\eqref{JeffreyMedPosteriorValidity}
   & 0.2861 \mbox{ in Eqn.}\,\eqref{PearlMedPosteriorValidity}
\end{tabular}
\end{center}

\noindent In both cases we see increases of likelihoods, from prior to
posterior, even though the increases are small, certainly in the Pearl
case. The increases make perfect sense, since the posterior is
obtained by incorporating during updating the very evidence of which
we determine the likelihood.  Informally speaking, after learning $p$,
this $p$ becomes more true.  Later on in
Theorems~\ref{JeffreyEvicenceUpdateThm}~\eqref{JeffreyEvicenceUpdateThmInc},
and~\ref{PearlEvicenceUpdateThm}~\eqref{PearlEvicenceUpdateThmInc}.
we shall see general results about such increases, for the Jeffrey,
and Pearl cases.

The numbers in the prior and test described in the introduction have
been chosen in a devious way, namely such that such increases do not
happen if one mixes up the mechanisms of Pearl and Jeffrey. This
happens for both mix-ups.

If one combines Jeffrey likelikhood with Pearl update one computes
for question~3 the validity:
\[ \begin{array}{rcl}
3 \cdot \big(\omega_{P} \models \postest\big)^{2} \cdot
  \big(\omega_{P} \models \negtest\big)
& \approx &
0.3081.
\end{array} \]

\noindent This is less than the original Jeffrey
likelihood~\eqref{JeffreyMedPriorValidity} of 0.3116.

Similarly, if Pearl likelihood is combined with a Jeffrey update one's
answer to question~3 is:
\[ \begin{array}{rcl}
3 \cdot \big(\omega_{J} \models \threetest\big)
& \approx &
0.2847.
\end{array} \]

\noindent This is less than the original Pearl
likelihood~\eqref{PearlMedPriorValidity} of 0.2858

Our lesson is that mixing up likelihood and update mechanisms may
break the likelihood increase that is naturally associated with
learning / updating. Therefore it is important not only to be aware of
the kind of likelihood computation (Jeffrey or Pearl), but also of the
(associated) kind of update mechanism.

\subsection{Replies to medical questions from the 
   introduction}\label{ReplySubsec}

As mentioned in the introduction, seventeen academic colleagues have
been kind enough to send in answers to the questions formulated in the
introduction. There is quite a bit of variation in the replies,
especially for updating. The number of participants is too low to draw
general conclusions, but the replies do give indications. At least
they show that specialists in the field do not all give the same
answer --- a fact that serves as motivation for the clarifications
provided in this paper.

Roughly, people in AI choose Pearl likelihood~\eqref{PearlValidity}
and also do the corresponding Pearl update
rule~\eqref{PearlMedUpdate}. This is a consistent approach. Four of
the respondents used a tool (Octave, Prodigy, EfProb, Genfer) for the
computation, which then does require a particular formalisation of the
solution.

People in medical statistics mostly choose Jeffrey
likelihood~\eqref{JeffreyValidity}, one of them via the tool
Julia. Several of them mention that they choose an independent
interpretation of the three tests. They do not use the corresponding
Jeffrey update rule~\eqref{JeffreyMedUpdate} that we apply above. This
update rule seems to be unfamiliar, even among specialists. Three of
these respondents combine Jeffrey's likelihood with Pearl's update ---
which is more familiar --- but the decrease of likelihood resulting
from this mix-up does not seem to raise concerns among them.

Several of the respondents shared their calculations. It was hard to
recognise systematics in the notation or in the approach --- apart
from two people referring to the binomial distribution. Several people
fill in numbers in a concrete computation and calculate an outcome,
without explaining the particular arrangement of numbers.

It is fair to say, even after such a small survey, that there is no
common, established `logic' to anwer the medical test questions
involving multiple pieces of evidence.

\section{Multisets and distributions}\label{MltDstSec}

After these preliminary discussions and observations about the medical
test questions we now proceed to introduce and develop the relevant
theory. This section describes for this purpose some basic facts about
multisets and distributions.

A multiset is like a subset, except that elements may occur multiple
times. An urn with five red balls, two blue and three green, is
written as a multiset of the form $5\ket{R} + 2\ket{B} + 3\ket{G}$.
We thus use ket's $\ket{-}$ not only for distributions, but also for
multisets. The three-test that we considered in the previous section
involves a multiset of predicates, namely $2\ket{\postest} +
1\ket{\negtest}$. Data items that are used for (probabilistic or
machine) learning can often be organised appropriately as multisets,
especially when items may occur multiple times and the order of
occurrence is irrelevant.

In general, a (finite) multiset over a set $X$ is written as
$m_{1}\ket{x_1} + \cdots + m_{N}\ket{x_N}$, with $x_{1}, \ldots, x_{N}
\in X$ and $m_{1}, \ldots, m_{N} \in \NNO$. The number $m_i$ expresses
the multiplicity of element $x_i$, that is, the number of its
occurrences in the multiset. Such a multiset over a set $X$ may
equivalently be described as a function $\varphi\colon X \rightarrow
\NNO$, with finite support, where the support of $\varphi$ is the
subset $\supp(\varphi) \coloneqq \setin{x}{X}{\varphi(x) \neq 0}$.  We
can then write $\varphi$ alternatively as formal sum
$\sum_{x}\varphi(x)\ket{x}$. We will freely switch between the ket
notation and function notation and use whichever form suits best.
We write $\natMlt(X)$ for the set of multisets over a set $X$.

The size $\|\varphi\|\in\NNO$ of a multiset $\varphi\in\natMlt(X)$ is
the total number of its elements, including multiplicities. Thus we
define $\|\varphi\| \coloneqq \sum_{x} \varphi(x)$. For a number
$K\in\NNO$ we write $\natMlt[K](X) \subseteq \natMlt(X)$ for the
subset of multisets of size $K$.


A $K$-sized list $(x_{1}, \ldots, x_{K}) \in X^{K}$ of elements from
$X$ can be turned into a $K$-sized multiset, via what we call
accumulation.  It is defined as $\acc(x_{1}, \ldots, x_{K}) =
1\ket{x_1} + \cdots + 1\ket{x_K}$ and gives a function $\acc \colon
X^{K} \rightarrow \natMlt[K](X)$. Concretely, we have
$\acc(c,b,a,a,b,a) = 3\ket{a} + 2\ket{b} + 1\ket{c}$.

Given a multiset $\varphi\in\natMlt[K](X)$ there are $\coefm{\varphi}$
many sequences in $X^{K}$ that accumulate to $\varphi$. This number
$\coefm{\varphi}\in\NNO$ is the multinomial coefficient of the
multiset $\varphi$ of size $K$, defined as:
\begin{equation}
\label{MultinomialCoefficientEqn}
\begin{array}{rclcrcl}
\coefm{\varphi}
& \coloneqq &
\displaystyle\frac{K!}{\facto{\varphi}}
& \qquad\mbox{where}\qquad &
\facto{\varphi}
& \coloneqq &
\displaystyle\prod_{x\in X} \varphi(x)!.
\end{array}
\end{equation}


\noindent Multinomial coefficients are useful, for instance in the
famous Multinomial Theorem (see \textit{e.g.}~\cite{Ross18}) which has
a snappy formulation in terms of multisets, as:
\begin{equation}
\label{MultinomMltTheoremEqn}
\begin{array}{rcll}
\big(r_{1} + \cdots + r_{n}\big)^{K}
& = &
\displaystyle\sum_{\varphi\in\natMlt[K](\{1,\ldots,n\})} 
   \coefm{\varphi}\cdot \prod_{1\leq i \leq n} r_{i}^{\varphi(i)}
& \qquad \mbox{for $r_{i}\in\R$.}
\end{array}
\end{equation}

We turn to distributions. As mentioned in
Section~\ref{MedicalSolutionSec}, a distribution over a set $X$ is a
finite formal convex sum $\sum_{i} r_{i}\ket{x_i}$ with $x_{i}\in X$,
$r_{i}\in [0,1]$ and $\sum_{i}r_{i}=1$. We can equivalently describe
such a distribution as a function $\omega\colon X \rightarrow [0,1]$
with finite support and $\sum_{x} \omega(x) = 1$. We write $\Dst(X)$
for the set of distributions on a set $X$. Each element $x\in X$ gives
rise to a point (or Dirac) distribution $1\ket{x} \in \Dst(X)$.

Each non-empty multiset $\varphi\in\natMlt(X)$ can be turned into a
distribution via normalisation. We call this frequentist learning and
describe it as $\flrn(\varphi)$, where:
\begin{equation}
\label{FlrnEqn}
\begin{array}{rcl}
\flrn(\varphi)
& \coloneqq &
\displaystyle\sum_{x\in X} \, \frac{\varphi(x)}{\|\varphi\|}\,\bigket{x}.
\end{array}
\end{equation}

\noindent Thus, for instance, $\flrn\big(3\ket{a} + 4\ket{b} +
5\ket{c}\big) = \frac{1}{4}\ket{a} + \frac{1}{3}\ket{b} +
\frac{5}{12}\ket{c}$.

For two distributions $\omega\in\Dst(X)$ and $\rho\in\Dst(Y)$ we can
form the parallel `tensor' product $\omega\otimes\rho\in\Dst(X\times
Y)$. It is defined in ket notation on the left below, and as function on
the right.
\[ \begin{array}{rclcrcl}
\omega\otimes\rho
& \coloneqq &
\displaystyle\sum_{x\in X,y\in Y} \omega(x)\cdot\rho(y)\,\bigket{x,y}
& \quad\mbox{\textit{i.e.}}\quad &
\big(\omega\otimes\rho\big)(x,y)
& \coloneqq &
\omega(x)\cdot\rho(y).
\end{array} \]

\noindent We write $\omega^{n}$ for the $n$-fold tensor
$\omega\otimes\cdots\otimes\omega \in \Dst(X^{n})$.

Distributions are closed under convex sums, in the following way.  For
numbers $r_{1}, \ldots, r_{n}\in [0,1]$, with $\sum_{i} r_{i} = 1$,
and distributions $\omega_{1}, \ldots, \omega_{n} \in \Dst(X)$ we can
form a new distribution in $\Dst(X)$, namely:
\[ \begin{array}{rcccl}
{\displaystyle\sum}_{i} \, r_{i} \cdot \omega_{i}
& = &
r_{1}\cdot\omega_{1} + \cdots + r_{n}\cdot\omega_{n}
& = &
\displaystyle\sum_{x\in X} \left(\textstyle{\displaystyle\sum}_{i} \, 
   r_{i} \cdot \omega_{i}(x)\right)\bigket{x}.
\end{array} \]

\noindent Concretely: $\frac{1}{3}\cdot\Big(\frac{1}{2}\ket{a} +
\frac{1}{2}\ket{b}\Big) + \frac{2}{3}\cdot\Big(\frac{1}{4}\ket{a} +
\frac{3}{4}\ket{b}\Big) = \frac{1}{3}\ket{a} + \frac{2}{3}\ket{b}$.

A function $f\colon X \rightarrow Y$ can be turned into a function
$\Dst(f) \colon \Dst(X) \rightarrow \Dst(Y)$. It acts on a
distribution in $\Dst(X)$ as described below.
\begin{equation}
\label{DstFunEqn}
\begin{array}{rcl}
\Dst(f)\left({\displaystyle\sum}_{i} \, r_{i}\bigket{x_i}\right)
& \coloneqq &
{\displaystyle\sum}_{i} \, r_{i}\bigket{f(x_{i})}.
\end{array}
\end{equation}

\noindent The function $\Dst(\pi_{1}) \colon \Dst(X\times Y)
\rightarrow \Dst(X)$ obtained from the first projection $\pi_{1}
\colon X\times Y \rightarrow X$ performs marginalisation. It maps a
joint distribution $\tau$ over $X\times Y$ to a distribution over $X$
by summing out the $Y$-part, as in:
\[ \begin{array}{rclcrcl}
\Dst(\pi_{1})(\tau)
& = &
\displaystyle\sum_{x\in X} \left(\sum_{y\in Y}\tau(x,y)\right)\bigket{x}
& \quad\mbox{so that}\quad &
\Dst(\pi_{1})\big(\omega\otimes\rho\big)
& = &
\omega.
\end{array} \]

\noindent The copy map $\Delta \colon X \rightarrow X\times X$, given
by $\Delta(x) = (x,x)$, gives $\Dst(\Delta) \colon \Dst(X) \rightarrow
\Dst(X\times X)$. It satisfies in general $\Dst(\Delta)(\omega) \neq
\omega\otimes\omega$. For instance, for the fair coin flip distribution
$\sigma = \frac{1}{2}\ket{H} + \frac{1}{2}\ket{T}$ one has:
\[ \begin{array}{rcl}
\Dst(\Delta)(\sigma)
& = &
\frac{1}{2}\ket{H,H} + \frac{1}{2}\ket{T,T}
\\[+0.2em]
& \neq &
\frac{1}{4}\ket{H,H} + \frac{1}{4}\ket{H,T} +
  \frac{1}{4}\ket{T,H} + \frac{1}{4}\ket{T,T}
\hspace*{\arraycolsep}=\hspace*{\arraycolsep}
\sigma\otimes\sigma.
\end{array} \]

\noindent Indeed, flipping a coin once and copying the outcomes is
different from flipping two coins.

\section{Validity for multisets of evidence}\label{ValiditySec}

In Section~\ref{MedicalSolutionSec} we have already introduced
validity $\omega\models p \coloneqq \sum_{x} \omega(x)\cdot p(x)$ for
a distribution $\omega$ and a predicate $p$. Also, we introduced
likelihoods in the style of Jeffrey and Pearl, for multiple
predicates. In this section we take a more systematic look at these
notions.

\begin{definition}
\label{PredDef}
Let $X$ be an arbitrary set.
\begin{enumerate}
\item A \emph{predicate} on $X$ is a function $X \rightarrow
  [0,1]$. More generally, a \emph{factor} on $X$ is a function $X
  \rightarrow \nnR$ and an \emph{observation} is a function $X
  \rightarrow \R$. Thus, each predicate is a factor, and each factor
  is an observation. A predicate is called \emph{sharp} if it
  restricts to $\{0,1\} \subseteq [0,1]$, that is, if either $p(x) =
  0$ or $p(x) = 1$, for each $x\in X$.

\item The truth predicate $\one\colon X \rightarrow [0,1]$ is always
  $1$ and the falsity predicate $\zero \colon X \rightarrow [0,1]$ is
  always $0$. Explicitly, $\one(x) = 1$ and $\zero(x) = 0$ for each
  $x\in X$. Both $\one$ and $\zero$ are thus sharp.

\item Observations (and factors and predicates) are ordered pointwise:
  $p\leq q$ means that $p(x) \leq q(x)$ holds for all $x$. Thus, each
  factor $p$ satisfies $\zero \leq p$ and each predicate $p$
  additionally satisfies $p \leq \one$.

\item For each subset $U\subseteq X$ there is sharp indicator
  predicate $\indic{U} \colon X \rightarrow [0,1]$ with $\indic{U}(x)
  = 1$ if $x\in U$ and $\indic{U}(x) = 0$ if $x\not\in U$. For an
  element $x\in X$ we write $\indic{x}$ for $\indic{\{x\}}$. This
  $\indic{x}$ is called a point predicate.

\item For two observations $p,q\colon X \rightarrow \R$ on the same
  set we define the conjunction $p\andthen q \colon X \rightarrow \R$
  as $\big(p\andthen q\big)(x) \coloneqq p(x)\cdot q(x)$. This
  conjunction $\andthen$ restricts to factors and to (sharp)
  predicates. It is easy to see that $\andthen$ is commutative and
  that it satisfies $p\andthen \one = p$ and $p\andthen \zero =
  \zero$.

\item For two observations $p\colon X \rightarrow \R$ and $q\colon Y
  \rightarrow \R$ on (possibly) different sets we define the parallel
  conjunction $p\otimes q \colon X \times Y \rightarrow \R$ also by
  pointwise multiplication: $\big(p \otimes q\big)(x,y) \coloneqq p(x)
  \cdot q(y)$. This $\otimes$ restricts to factors and predicates too.

\item For two observations $p,q\colon X \rightarrow \R$ we define the
  sum $p+q \colon X \rightarrow \R$ also pointwise, as
  $\big(p+q\big)(x) \coloneqq p(x) + q(x)$. This sum restricts to
  factors but not to predicates.

\item For an observation $p\colon X \rightarrow \R$ and a scalar
  $s\in\R$ one defines the scalar multiplication $s\cdot p \colon X
  \rightarrow \R$ as $\big(s\cdot p\big)(x) \coloneqq s\cdot p(x)$.
  This scalar multiplication restricts to factors if one restricts the
  scalar $s$ to the non-negative reals $\nnR$, and to predicates if
  the scalars are from the unit interval $[0,1]$.

\item For a predicate $p\colon X \rightarrow [0,1]$ there is the
  orthosupplement / negation $p^{\bot} \colon X \rightarrow [0,1]$
  defined as $p^{\bot}(x) \coloneqq 1 - p(x)$. Then $p^{\bot\bot} = p$
  and $p + p^{\bot} = \one$.

\item We write $\Obs(X) = \R^{X}$ for the set of observations /
  functions $X \rightarrow \R$. The subset of factors and predicates
  are written as $\Pred(X) \subseteq \Fact(X) \subseteq \Obs(X)$. One
  may identify sharp predicates with subsets, giving an inclusion
  $\indic{(-)} \colon \Pow(X) \rightarrow \Pred(X)$.
\end{enumerate}
\end{definition}

Traditionally in probability theory one assigns probabilities to
`events', that is, to subsets $U\subseteq X$ of the space of
possiblities $X$. In the present setting such an event corresponds to
a sharp indicator predicate $\indic{U}$. The probability of this event
is often written as $\Prob(U)$, where the distribution at hand is left
implicit. We would write this as $\omega\models\indic{U} = \sum_{x\in
  U}\omega(x) = \Prob(U)$. Leaving the distribution implicit is
inconvenient, or even confusing, when there are different
distributions around. For instance, the important validity increase of
Theorem~\ref{UpdateGainThm}~\eqref{UpdateGainThmSingle} cannot be
expressed at all if one leaves the distribution implicit.

Predicates with values in $[0,1]$, as used above, are often called
fuzzy or soft. They make perfectly good sense in a probabilistic
setting, as we have seen in Section~\ref{MedicalSolutionSec} with the
positive and negative test predicates $\postest$ and $\negtest$. See
also for instance~\cite[\S3.6-3.7]{Darwiche09} for other usage of
fuzzy predicates in a probabilistic context. Using the notation
introduced in the above definition we can write these two predicates
$\postest$ and $\negtest$ as weighted sums of point predicates:
$\postest = \frac{9}{10}\cdot\indic{d} +
\frac{2}{5}\cdot\indic{\no{d}}$ and $\negtest =
\frac{1}{10}\cdot\indic{d} + \frac{3}{5}\cdot\indic{\no{d}}$. Notice
that $\postest + \negtest = \one$, so that we can also write $\postest
= \negtest^{\bot}$.

The sum $p+q$ is defined above only for observations and factors
$p,q$. It may also be defined for predicates $p,q$ as a partial
operation that exists when $p(x) + q(x) \leq 1$ for each $x$. Such
partial sums are axiomatised in the notion of effect algebra; it
appeared in the context of quantum logic,
see~\cite{FoulisB94,Gudder96,DvurecenskijP00,ChoJWW15b}, but will not
be used here.

It is easy to see that $p\otimes\zero = \zero$. In contrast,
$p\otimes\one$ does not trivialise. In logical terms, $p\otimes\one$
is the weakening of $p$, arising by moving a predicate (or factor,
observation) $p\colon X \rightarrow [0,1]$ in context $X$ to an
extended context $X\times Y$, as $p\otimes\one \colon X \times Y
\rightarrow [0,1]$, with $(p\otimes\one)(x,y) = p(x)$. Thus, where
marginalisation moves a (joint) distribution to a smaller context,
weakening moves a predicate to a larger context, see
Lemma~\ref{ValidityLem}~\eqref{ValidityLemMarg} below.

We collect some basic results about validity $\models$. The proofs
are easy and left to the reader.

\begin{lemma}
\label{ValidityLem}
Recall from~\eqref{ValidityKetEqn} that we write $\omega\models p$ for
the sum $\sum_{x} \omega(x)\cdot p(x)$, for a distribution
$\omega\in\Dst(X)$ and an observation $p\in\Obs(X)$. Then:
\begin{enumerate}
\item \label{ValidityLemOne} $\omega\models\one \,=\, 1$ and
  $\omega\models\zero \,=\, 0$, and $\omega\models\indic{x} \,=\,
  \omega(x)$;

\item \label{ValidityLemLin} $\omega \models (p+q) \,=\,
  \big(\omega\models p\big) + \big(\omega\models q\big)$ and
  $\omega\models (s\cdot p) \,=\, s\cdot \big(\omega\models p\big)$;

\item \label{ValidityLemOrd} if $p\leq q$ then $\big(\omega\models
  p\big) \leq \big(\omega \models q\big)$;

\item \label{ValidityLemTensor} $(\omega\otimes\rho) \models (p\otimes
  q) \,=\, \big(\omega\models p\big) \cdot \big(\rho \models q\big)$;

\item \label{ValidityLemOrtho} $\omega\models p^{\bot} \,=\, 1 -
  \big(\omega\models p)$;

\item \label{ValidityLemMarg} $\Dst(\pi_{1})(\tau) \models p \,=\,
  \tau \models p\otimes \one$ and $\omega \models (p\andthen q) \,=\,
  \Dst(\Delta)(\omega) \models (p\otimes q)$. \QED
\end{enumerate}
\end{lemma}

In the present context we use the word evidence for multiple pieces of
information whose validity can be determined and that can be used for
updating (learning). A crucial point is that evidence may have
multiple items, for instance coming from an experiment that is
repeated several times. More concretely, we define evidence as a
multiset of factors. We use factors, and not more general
observations, since we like to use evidence for updating (in the next
section). Updating works for factors, not for observations, since
negative values are problematic in normalisation.

\begin{definition}
\label{EvidenceDef}
\begin{enumerate}
\item Evidence is given by a multiset of factors over the same
  underlying set, say $X$, so that evidence is a multiset
  $\psi\in\natMlt\big(\Fact(X)\big)$. A common special case is where
  all factors involved are predicates, so that
  $\psi\in\natMlt\big(\Pred(X)\big)$.

\item We use the following two special cases.
\begin{itemize}
\item A single factor $p\colon X \rightarrow \nnR$ will also be called
  evidence, where formally it would be a single piece of evidence as
  given by the singleton multiset $1\ket{p}$.

\item When evidence $\psi\in\natMlt\big(\Pred(X)\big)$ consists only of
  point predicates $\indic{x}$, this $\psi$ may be identified with a
  multiset $\psi \in \natMlt(X)$ over $X$. We may then call $\psi$ point
  evidence.
\end{itemize}

\item We say that evidence $\psi\in\natMlt\big(\Fact(X)\big)$
  \emph{matches} (or is a match) when the sum of all factors in the
  support of $\psi$ is below the truth predicate $\one$, that is, if
  $\sum_{p\in\supp(\psi)} p \leq \one$. This implies that all factors
  in $\psi$ are predicates, and thus that $\psi$ inhabits the set
  $\natMlt\big(\Pred(X)\big)$. We call $\psi$ a \emph{perfect match} when
  $\sum_{p\in\supp(\psi)} p = \one$.
\end{enumerate}
\end{definition}

The evidence used in Section~\ref{MedicalSolutionSec} consisted of two
positive and one negative test. We can write this now as a multiset
$2\bigket{\postest} + 1\bigket{\negtest}$. It is a perfect match since
$\postest + \negtest = \one$.

We define evidence as a \emph{multiset}, and not as a \emph{list}, of
factors because we regard the order of the factors involved as
irrelevant. This order-irrelevance also applies to updating, in
Section~\ref{UpdateSec}. 

We need to fix some notation for iterated conjunctions of predicates
and evidence.

\begin{definition}
\label{EvidenceConjunctionDef}
\begin{enumerate}
\item \label{EvidenceConjunctionDefFact} Let a factor $p\in\Fact(X)$
  be given, with a number $n\in\NNO$.  We iterate the two conjunctions
  $\andthen$ and $\otimes$ for factors in the following way.
\[ \begin{array}{rclcrcl}
p^{n}
& \coloneqq &
\underbrace{p \andthen \cdots \andthen p}_{n\text{ times}} \in \Fact(X)
& \quad\mbox{and}\quad &
p^{\otimes n}
& \coloneqq &
\underbrace{p \otimes \cdots \otimes p}_{n\text{ times}} \in \Fact\big(X^{n}\big).
\end{array} \]

\noindent Explicitly, $p^{n}(x) = p(x)^{n}$ and $p^{\otimes n}(x_{1},
\ldots, x_{n}) = p(x_{1}) \cdot \ldots \cdot p(x_{n})$.  In the border
case, when $n=0$, we have $p^{0} = \one\in\Fact(X)$ and $p^{\otimes 0}
= \one \in \Fact(X^{0}) = \Fact(1)$, where $1$ is a singleton set.

\item \label{EvidenceConjunctionDefEv} For an evidence multiset
  $\psi\in\natMlt[K]\big(\Fact(X)\big)$ of size $K$ we write:
\[ \begin{array}{rcl}
\evidand\psi
& \coloneqq &
\displaystyle\bigandthen\limits_{p\in\supp(\psi)}\, p^{\psi(p)} \,\in\, \Fact(X)
\\[+1em]
\evidtensor\psi
& \coloneqq &
\displaystyle\bigotimes_{p\in\supp(\psi)} p^{\otimes\psi(p)} \,\in\,\Fact\big(X^{K}\big).
\end{array} \]
\end{enumerate}
\end{definition}

For instance, for evidence $\psi = 2\ket{q} + 3\ket{r}$ we have:
\[ \begin{array}{rcl}
\evidand\psi
& = &
q^{2} \andthen r^{3}
\hspace*{\arraycolsep}=\hspace*{\arraycolsep}
q \andthen q \andthen r \andthen r \andthen r
\\[+0.2em]
\evidtensor\psi
& = &
q^{\otimes 2} \andthen r^{\otimes 3}
\hspace*{\arraycolsep}=\hspace*{\arraycolsep}
q \otimes q \otimes r \otimes r \otimes r.
\end{array} \]

\noindent As an aside: the notation $\evidtensor\psi$ is a bit
dangerous because the order of tensoring is relevant, in principle.
We evaluate such expressions $\evidtensor\psi$ only in distributions
of the form $\omega\otimes\cdots\otimes\omega$ or
$\Dst(\Delta)(\omega)$, so that such order issues do not matter.

We now turn to a general formulation of the two likelihoods of Jeffrey
and Pearl that we distinguished in Subsection~\ref{QuestionOneSubsec}.

\begin{definition}
\label{EvidenceValidityDef}
Let $\omega\in\Dst(X)$ be a distribution and
let $\psi\in\natMlt\big(\Fact(X)\big)$ be an evidence multiset of
size $K = \|\psi\|$.
\begin{enumerate}
\item The \emph{Jeffrey validity} $\omega \Jmodels \psi$ of $\psi$ in
  $\omega$ is defined as:
\begin{equation}
\label{JeffreyValidity}
\begin{array}{rcl}
\omega \Jmodels \psi
& \coloneqq &
\displaystyle\coefm{\psi}\cdot
   \prod_{p\in\supp(\psi)} \big(\omega\models p\big)^{\psi(p)}
\\[+0.8em]
& = &
\coefm{\psi}\cdot \displaystyle\left(
   \underbrace{\omega\otimes\cdots\otimes\omega}_{K\text{ times}} \models
  \bigotimes_{p\in\supp(\psi)} p^{\otimes\psi(p)}\right)
\\[+1.4em]
& = &
\coefm{\psi}\cdot\Big(\omega^{K} \models \evidtensor\psi\Big).
\end{array}
\end{equation}

\noindent This last equation holds by
Lemma~\ref{ValidityLem}~\eqref{ValidityLemTensor}. We recall that
$\coefm{\psi}$ is the multinomial coefficient
from~\eqref{MultinomialCoefficientEqn}. It takes care of the different
orders in which the validities of the factors in $\psi$ may be
evaluated.

\item The \emph{Pearl validity} $\omega \Pmodels \psi$ of $\psi$ in
  $\omega$ is:
\begin{equation}
\label{PearlValidity}
\begin{array}{rcl}
\omega \Pmodels \psi
& \coloneqq &
\coefm{\psi}\cdot\Big(\omega \models \evidand\psi\Big)
\\[+0.5em]
& = &
\displaystyle\coefm{\psi}\cdot \left(
   \omega\models \bigandthen\limits_{p\in\supp(\psi)}\, p^{\psi(p)} \right)
\\[+0.8em]
& = &
\coefm{\psi}\cdot \displaystyle\left(
   \Dst(\Delta)(\omega) \models
  \bigotimes_{p\in\supp(\psi)} p^{\otimes\psi(p)}\right)
\\[+1.4em]
& = &
\coefm{\psi}\cdot\Big(\Dst(\Delta)(\omega) \models \evidtensor\psi\Big),
\end{array}
\end{equation}

\noindent where $\Delta \colon X \rightarrow X^{K}$ is the $K$-fold
copier $x \mapsto (x,\ldots,x)$. The last equation holds by
Lemma~\ref{ValidityLem}~\eqref{ValidityLemMarg}.
\end{enumerate}
\end{definition}

The last lines of~\eqref{JeffreyValidity} and~\eqref{PearlValidity}
give a clear difference between the formulations of Jeffrey and Pearl:
they both use the validity of the parallel conjunction factor
$\evidtensor\psi$, in the parallel product distribution $\omega^{K} =
\omega\otimes\cdots\otimes\omega \in \Dst\big(X^{K}\big)$ for Jeffrey,
and in the copied distribution $\Dst(\Delta)(\omega) \in
\Dst\big(X^{K}\big)$ for Pearl. As we have seen in
Section~\ref{MltDstSec}, these two distributions $\omega^{K}$ and
$\Dst(\Delta)(\omega)$ differ, in general. The difference between the
joint distributions $\omega^K$ and $\Dst(\Delta)(\omega)$ corresponds
to the difference between independent and dependent evaluation of the
evidence items, in Jeffrey's and in Pearl's approach.  Later, in
Lemma~\ref{PredicateUpdateLem}~\eqref{PredicateUpdateLemIter}, we
shall see that this Pearl validity can be expressed in yet another
way, as discussed in Remark~\ref{PearlMedicalUpdateRem}.

It is easy to see that the earlier Jeffrey and Pearl likelihood
calculations~\eqref{JeffreyMedPriorValidity}
and~\eqref{PearlMedPriorValidity} are instances of the above general
definitions~\eqref{JeffreyValidity} and~\eqref{PearlValidity}. As
further illustration, we elaborate these definitions for $\psi =
2\ket{p} + 3\ket{q}$, then:
\[ \begin{array}{rcl}
\omega \Jmodels \psi
& \smash{\stackrel{\eqref{JeffreyValidity}}{=}} &
\coefm{\psi} \cdot \big(\omega\models p)^{2} \cdot \big(\omega\models q)^{3}
\\[+0.2em]
& = &
\frac{5!}{2!\cdot 3!} \cdot
\big(\omega\models p) \cdot \big(\omega\models p) \cdot 
   \big(\omega\models q) \cdot \big(\omega\models q) \cdot \big(\omega\models q)
\\[+0.2em]
& = &
10 \cdot \Big(\omega\otimes\omega\otimes\omega\otimes\omega\otimes\omega \models
   p \otimes p\otimes q \otimes q \otimes q\Big)
\\[+0.4em]
& = &
10 \cdot \Big(\omega^{5} \models p^{\otimes2} \otimes q^{\otimes3}\Big)
\\[+0.4em]
& = &
10 \cdot \Big(\omega^{5} \models \evidtensor\psi\Big).
\end{array} \]

\noindent Similarly, for Pearl:
\[ \begin{array}{rcl}
\omega \Pmodels \psi
& \smash{\stackrel{\eqref{PearlValidity}}{=}} &
\coefm{\psi} \cdot \Big(\omega \models \evidand\psi\Big)
\\[+0.4em]
& = &
10 \cdot \Big(\omega \models p^{2} \andthen q^{3}\Big)
\\[+0.4em]
& = &
10 \cdot \Big(\omega \models p \andthen p \andthen q \andthen q \andthen q\Big)
\\[+0.4em]
& = &
10 \cdot \Big(\Dst(\Delta)(\omega) \models 
   p \otimes p\otimes q \otimes q \otimes q\Big)
\\[+0.4em]
& = &
10 \cdot \Big(\Dst(\Delta)(\omega) \models p^{\otimes2} \otimes q^{\otimes3}\Big)
\\[+0.4em]
& = &
10 \cdot \Big(\Dst(\Delta)(\omega) \models \evidtensor\psi\Big).
\end{array} \]

\begin{figure}
\begin{center}
\includegraphics[scale=0.5]{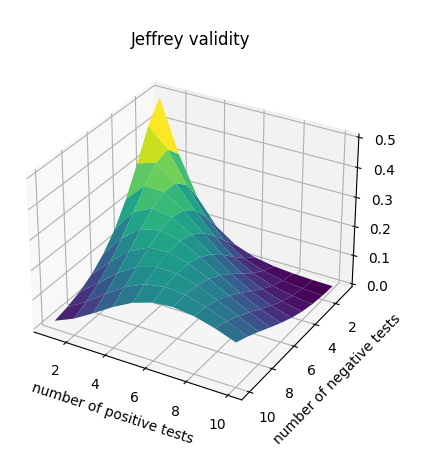}
\hspace*{2em}
\includegraphics[scale=0.5]{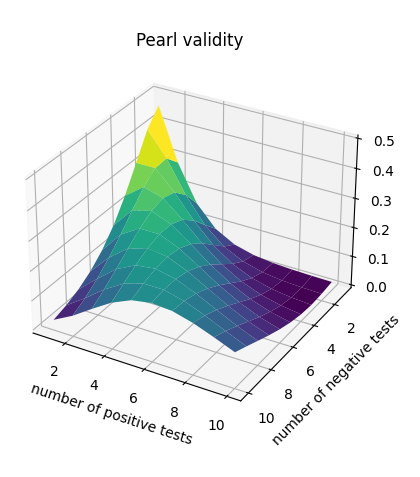}
\end{center}
\caption{Jeffrey and Pearl validities in the medical test
  scenario from the introduction, with 100 different versions of
  evidence $i\bigket{\postest} + j\ket{\negtest}$, where
  $i,j\in\{1,\ldots,10\}$.}
\label{ValidityFig}
\end{figure}

In Figure~\ref{ValidityFig} we plot, in the earlier medical setting,
the 100 Jeffrey and Pearl validities for evidence of the form
$i\bigket{\postest} + j\ket{\negtest}$, representing $i$ positive and
$j$ negative tests, where $i,j\in\{1,\ldots,10\}$. We see that the
shapes look remarkably similar, but there are small differences
between their validities (all less than 0.033). 

Unlike in these cases, there can be substantial differences between
Jeffrey and Pearl valdities. For instance for the set $X =
\{a,b,c\}$ consider the distribution and predicate:
\[ \begin{array}{rclcrcl}
\omega
& = &
\frac{3}{10}\ket{a} + \frac{3}{10}\ket{b} + \frac{2}{5}\ket{c}
& \qquad\qquad &
p
& = &
\frac{1}{100}\cdot\indic{a} + \frac{1}{100}\cdot\indic{b} + 
   \frac{49}{50}\cdot\indic{c},
\end{array} \]

\noindent with (perfectly matching) evidence $\psi = 1\bigket{p} +
1\bigket{p^{\bot}}$. The difference between the Jeffrey and Pearl
validities is more than $0.45$, since:
\[ \begin{array}{rclcrcl}
\omega\Jmodels\psi
& \approx &
0.479
& \qquad\mbox{and}\qquad &
\omega\Pmodels\psi
& \approx &
0.028.
\end{array} \]

\noindent It may also happen that the Pearl validity is bigger. For
instance, for the same distribution $\omega$, but with predicate $q =
\frac{3}{10}\cdot\indic{a} + \frac{1}{5}\cdot\indic{b} +
\frac{9}{10}\cdot\indic{c}$ and evidence $\chi = 1\bigket{q} +
4\bigket{q^{\bot}}$ one gets:
\[ \begin{array}{rclcrcl}
\omega\Jmodels\chi
& \approx &
0.054
& \qquad\mbox{and}\qquad &
\omega\Pmodels\chi
& \approx &
0.154.
\end{array} \]

\noindent Hence it does matter which kind of validity is used in a
particular situation. We see that the Jeffrey valdity can be both
bigger and smaller than the Jeffrey validity. In one special situation
it is always smaller, namely when we have have evidence of the form
$n\ket{p}$ with one factor $p$ occurring multiple (in fact, $n$)
times. Then, for any distribution $\omega$,
\[ \begin{array}{rcccccl}
\omega\Jmodels n\ket{p}
& = &
\big(\omega\models p\big)^{n}
& \leq &
\omega \models p^{n}
& = &
\omega\Pmodels n\ket{p}.
\end{array} \]

\noindent The proof for $n=2$ is given in~\cite[Proof of
  Prop.~12]{Jacobs23c}, but the inequality holds for any $n\in\NNO$.

\ignore{

X = alphaspace(3)
w = DState([0.3, 0.3, 0.4], X)
a = 0.01
b = 0.01
p = DPred([a, b, 1 - a - b], X)
ev = DState([1, 1], Space(p, ~p))
jv = w.Jeffrey_validity(ev)
pv = w.Pearl_validity(ev)
print("\nJeffrey validity: ", jv )
print("Pearl validity:   ", pv )
print("Validity difference: ", abs(jv - pv) )
print("\nFor distribution and evidence:")
print( w )
print( ev )

# Jeffrey validity:  0.479192
# Pearl validity:    0.027560000000000015
# Validity difference:  0.451632

# For distribution and evidence:
# 0.3|a> + 0.3|b> + 0.4|c>
# 1|0.01 . a || 0.01 . b || 0.98 . c> + 1|0.99 . a || 0.99 . b || 0.02 . c>

X = alphaspace(3)
w = DState([0.2, 0.2, 0.6], X)
p = DPred([0.3, 0.2, 0.9], X)
ev = DState([1, 4], Space(p, ~p))
jv = w.Jeffrey_validity(ev)
pv = w.Pearl_validity(ev)
print("\nJeffrey validity: ", jv )
print("Pearl validity:   ", pv )
print("Validity difference: ", abs(jv - pv) )
print("\nFor distribution and evidence:")
print( p )
print( ev )

# Jeffrey validity:  0.053747711999999996
# Pearl validity:    0.15422
# Validity difference:  0.10047228799999999

# For distribution and evidence:
# 0.3 . a || 0.2 . b || 0.9 . c
# 1|0.3 . a || 0.2 . b || 0.9 . c> + 4|0.7 . a || 0.8 . b || 0.1 . c>

}

\begin{remark}
\label{VarianceRem}
When the evidence consists of two factors both occurring once,
  say $\psi = 1\ket{p_{1}} + 1\ket{p_{2}}$, one has as Jeffrey and
  Pearl validities:
\[ \begin{array}{rclcrcl}
\omega\Jmodels\psi
& = &
2 \cdot \big(\omega\models p_{1}\big) \cdot \big(\omega\models p_{2}\big)
& \qquad\mbox{and}\qquad &
\omega \Pmodels \psi
& = &
2 \cdot\big(\omega \models p_{1}\andthen p_{2}\big).
\end{array} \]

\noindent These two expressions occur in the covariance of $p_{1}, p_{2}$,
since:
\[ \begin{array}{rcl}
\mathsl{Cov}\big(\omega,p_{1},p_{2}\big)
& \,\coloneqq\, &
\omega \models \big(p_{1} - (\omega\models p_{1})\cdot\one\big) \andthen
   \big(p_{2} - (\omega\models p_{2})\cdot\one\big)
\\[+0.2em]
& = &
\big(\omega\models p_{1} \andthen p_{2}\big) -
   \big(\omega\models p_{1}\big)\cdot\big(\omega\models p_{2}\big)
\\[+0.2em]
& = &
\frac{1}{2}\cdot\Big(\big(\omega \Pmodels \psi\big) - 
   \big(\omega\Jmodels\psi\big)\Big).
\end{array} \]

\noindent This relationship suggests that one can define covariance
more generally, of the form $\mathsl{Cov}\big(\omega, \psi\big)$, for
a general evidence multiset $\psi$, via the difference between Pearl
and Jeffrey validity of $\psi$.
\end{remark}


\begin{lemma}
\label{MatchingEvidenceBoundLem}
Let $\omega\in\Dst(X)$ be an arbitrary distribution.
\begin{enumerate}
\item Let $\psi\in\natMlt\big(\Pred(X)\big)$ be matching evidence, so
  that $\sum_{p\in\supp(\psi)} p \leq \one$. Both the Jeffrey and
  Pearl validities $\omega \Jmodels \psi$ and $\omega \Pmodels \psi$
  are in the unit interval $[0,1]$. Explicitly:
\[ \begin{array}{rclcrcl}
\omega \Jmodels \psi
& \leq &
1
& \qquad\mbox{and}\qquad &
\omega \Pmodels \psi
& \leq &
1.
\end{array} \]

\item Let $P \subseteq \Pred(X)$ be a finite subset of predicates that
  add up to one: $\sum_{p\in P} p = \one$. For each $K\geq 1$, the
  sums of Jeffrey and Pearl validities of all (perfectly matching)
  evidence over $P$ of size $K$ add up to one and form distributions
  of the form:
\[ \begin{array}{rcl}
\displaystyle\sum_{\varphi\in\natMlt[K](P)} 
   \big(\omega\Jmodels \varphi\big)\,\bigket{\varphi}
& = &
\displaystyle\sum_{\varphi\in\natMlt[K](P)} \coefm{\varphi} \cdot
   \prod_{p\in P} \big(\omega\models p\big)^{\varphi(p)}\,\bigket{\varphi}
\\[+1.2em]
\displaystyle\sum_{\varphi\in\natMlt[K](P)} 
   \big(\omega\Pmodels \varphi\big)\,\bigket{\varphi}
& = &
\displaystyle\sum_{\varphi\in\natMlt[K](P)} \coefm{\varphi} \cdot
    \big(\omega\models \bigandthen_{p\in P}p^{\varphi(p)}\big)\,\bigket{\varphi}.
\end{array} \]
\end{enumerate}
\end{lemma}

\begin{proof}
\begin{enumerate}
\item Let the evidence multiset $\psi$ have size $K$, that is, $\|\psi\| =
K$.  We use the Multinomial Theorem~\ref{MultinomMltTheoremEqn}. In
the Jeffrey case this gives:
\[ \begin{array}{rcl}
\omega \Jmodels \psi
& \smash{\stackrel{\eqref{JeffreyValidity}}{=}} &
\displaystyle \coefm{\psi} \cdot \prod_{p\in\supp(\psi)}
   \big(\omega\models p\big)^{\psi(p)}
\\[+1.0em]
& \leq &
\displaystyle\sum_{\varphi\in\natMlt[K](\supp(\psi))} \coefm{\varphi} \cdot
   \prod_{p\in\supp(\psi)} \big(\omega\models p\big)^{\varphi(p)}
\\[+1.4em]
& \smash{\stackrel{\eqref{MultinomMltTheoremEqn}}{=}} &
\displaystyle\left(\sum_{p\in\supp(\psi)} \omega\models p\right)^{K}
\\[+1.4em]
& = &
\displaystyle\left(\omega\models\sum_{p\in\supp(\psi)} p\right)^{K}
   \hspace*{1em}\mbox{by Lemma~\ref{ValidityLem}~\eqref{ValidityLemLin}}
\\[+1.2em]
& \leq &
\Big(\omega\models\one\Big)^{K} 
   \hspace*{5.3em}\mbox{since $\psi$ is matching}
\\[+0.4em]
& = &
1^{K}
\hspace*{\arraycolsep}=\hspace*{\arraycolsep} 1.
\end{array} \]

For the Pearl case we write $q \coloneqq \bigandthen_{p} p^{\psi(p)}$
for the conjunction. We now use the Multinomial Theorem to show that
$\coefm{\psi}\cdot q$ is pointwise below $\one$, in:
\[ \begin{array}{rcl}
\Big(\coefm{\psi}\cdot q\Big)(x)
\hspace*{\arraycolsep}=\hspace*{\arraycolsep}
\displaystyle \coefm{\psi} \cdot q(x)
& = &
\displaystyle \coefm{\psi} \cdot \prod_{p\in\supp(\psi)} p(x)^{\psi(p)}
\\[+1.2em]
& \leq &
\displaystyle \sum_{\varphi\in\natMlt[K](\supp(\psi))} \coefm{\varphi} \cdot
   \prod_{p\in\supp(\psi)} p(x)^{\varphi(p)}
\\
& \smash{\stackrel{\eqref{MultinomMltTheoremEqn}}{=}} &
\displaystyle\left(\sum_{p\in\supp(\psi)} p(x)\right)^{K}
\\[+1.4em]
& \leq &
1^{K} \hspace*{6em}\mbox{since $\psi$ is matching}
\\
& = &
1.
\end{array} \]
   
\noindent Now we hare done by
Lemma~\ref{ValidityLem}~\eqref{ValidityLemLin}:
\[ \begin{array}{rcccccccl}
\omega \Pmodels \psi
& \smash{\stackrel{\eqref{PearlValidity}}{=}} &
\coefm{\psi}\cdot \Big(\omega \models q\Big)
& = &
\omega \models \coefm{\psi}\cdot q
& \leq &
\omega \models \one
& = &
1.
\end{array} \]

\item Again, by the Multinomial Theorem the probabilities
  $\omega\Jmodels\varphi$ and $\omega\Pmodels\varphi$, for all
  perfectly matching evidence $\varphi\in\natMlt[K](P)$, now add up to
  precisely one. \QED
\end{enumerate}
\end{proof}

The match requirement is necessary in order to get validities below
one. Consider, for instance the uniform distribution $\omega =
\frac{1}{2}\ket{a} + \frac{1}{2}\ket{b}$ with non-matching evidence
$\psi = 2\ket{p} + 3\ket{q}$, for $p = \indic{a} +
\frac{1}{2}\cdot\indic{b}$ and $q = \frac{4}{5}\cdot\indic{a} +
\frac{1}{2}\cdot\indic{b}$. Then $\omega\Jmodels\psi =
\frac{19773}{12800} > 1$, but $< 2$, and $\omega\Pmodels\psi =
\frac{2173}{800} > 2$.

\ignore{

X = Space("a", "b")
w = uniform_distribution(X, frac=True)
p = DPred([ one_frac, Frac(1,2) ], X)
q = DPred([ Frac(4,5), Frac(1,2) ], X)
psi = DState([2,3], Space(p,q))
jv = w.Jeffrey_validity(psi)
print( jv, jv.as_float() )
pv = w.Pearl_validity(psi)
print( pv, pv.as_float() )

# 19773/12800 1.544765625
# 2173/800 2.71625

}

Jeffrey validity is sometimes expressed via log-likelihood. One
motivation to switch to (natural) logarithmic formulations is that the
multiplication in~\eqref{JeffreyValidity} may lead to very small
numbers, so small that rounding errors start to kick in. Another
motivation is that sum formulas may be easier to maximise, for
instance by taking derivatives --- since the derivative of a sum is a
sum of derivatives.

\begin{lemma}
\label{JeffreyValLogLem}
Let $\omega, \omega'\in\Dst(X)$ be two distributions, with an evidence
multiset $\psi \in \natMlt\big(\Fact(X)\big)$. Then:
\[ \begin{array}{rcl}
\omega\Jmodels\psi \,\leq\, \omega'\Jmodels\psi
& \;\Longleftrightarrow\; &
\displaystyle
   \flrn(\psi) \models \ln\left(\frac{\omega\models -}{\omega'\models -}\right)
\,\leq\, 0.
\end{array} \]
\end{lemma}

\begin{proof}
We use that the (natural) logarithm $\ln \colon \pR \rightarrow \R$
preserves and reflects the order: $a \leq b$ iff $\ln(a) \leq
\ln(b)$. We additionally use that the logarithm sends multiplications
to sums.
\[ \begin{array}[b]{rcl}
\lefteqn{\omega\Jmodels\psi
\,\leq\,
\omega'\Jmodels\psi}
\\[+0.2em]
& \Longleftrightarrow &
\ln\big(\omega\Jmodels\psi\big)
\,\leq\,
\ln\big(\omega'\Jmodels\psi\big)
\\[+0.2em]
& \smash{\stackrel{\eqref{JeffreyValidity}}{\Longleftrightarrow}} &
\displaystyle
\ln\left(\coefm{\psi} \cdot
   \prod_{p\in\supp(\psi)} \big(\omega\models p\big)^{\psi(p)}\right)
\,\leq\,
\ln\left(\coefm{\psi} \cdot
   \prod_{p\in\supp(\psi)} \big(\omega'\models p\big)^{\psi(p)}\right)
\\[+1.4em]
& \Longleftrightarrow &
\displaystyle\ln\big(\coefm{\psi}\big) + 
   \sum_{p\in\supp(\psi)} \psi(p)\cdot \ln\big(\omega\models p\big)
\\[+1.0em]
& & \qquad \,\leq\, \displaystyle
\ln\big(\coefm{\psi}\big) + 
\sum_{p\in\supp(\psi)} \psi(p)\cdot \ln\big(\omega'\models p\big)
\\[+0.4em]
& \Longleftrightarrow &
\displaystyle\sum_{p\in\supp(\psi)} 
   \frac{\psi(p)}{\|\psi\|} \cdot \ln\big(\omega\models p\big)
\,\leq\,
\sum_{p\in\supp(\psi)} \frac{\psi(p)}{\|\psi\|} \cdot 
   \ln\big(\omega'\models p\big)
\\[+1.4em]
& \Longleftrightarrow &
\flrn(\psi) \models \ln\big(\omega\models -\big)
\,\leq\,
\flrn(\psi) \models \ln\big(\omega'\models -\big).
\\[+0.4em]
& \Longleftrightarrow &
\flrn(\psi) \models \Big(\ln\big(\omega\models -\big) - 
   \ln\big(\omega'\models -\big)\Big)
\,\leq\,
0
\\[+0.2em]
& \Longleftrightarrow &
\displaystyle
\flrn(\psi) \models \ln\left(\frac{\omega\models -}{\omega'\models -}\right)
\,\leq\, 0.
\end{array} \eqno{\QEDbox} \]
\end{proof}

\subsection{Jeffrey validity of point evidence}\label{PointEvSubsec}

In the remainder of this section we concentrate on point evidence and
its Jeffrey validity. But first, we introduce the multinomial
distribution, in general, multivariate form.  It captures the
probabilities associated with drawing (with replacement) multiple
balls from an urn, filled with balls of various colours.  The urn is
represented by a distribution $\omega\in\Dst(X)$, where $X$ is the set
of colours. Thus, the probability of drawing a ball of colour $x\in X$
is given by $\omega(x)$. We are interested in drawing multiple balls,
of a fixed size $K$. These draws are represented by a multiset
$\varphi\in\natMlt[K](X)$. The multinomial distribution
$\multinomial[K](\omega) \in \Dst\big(\natMlt[K](X)\big)$ assigns
probabilities to such draws. It can be described as follows.
\begin{equation}
\label{MultinomialEqn}
\begin{array}{rcl}
\multinomial[K](\omega)
\hspace*{\arraycolsep}\coloneqq\hspace*{\arraycolsep}
\Dst(\acc)\big(\omega^{K}\big)
& \smash{\stackrel{\eqref{DstFunEqn}}{=}} & 
\displaystyle\sum_{\vec{x}\in X^{K}} \, \omega^{K}(\vec{x})
   \, \bigket{\acc(\vec{x})}
\\[+1em]
& = &
\displaystyle\sum_{\varphi\in\natMlt[K](X)} \, 
   \sum_{\vec{x}\in\acc^{-1}(\varphi)} \, \omega^{K}(\vec{x})\,\bigket{\varphi}.
\\[+1.2em]
& = &
\displaystyle\sum_{\varphi\in\natMlt[K](X)} \, \coefm{\varphi} \cdot
   \prod_{x\in X} \omega(x)^{\varphi(x)}\,\bigket{\varphi}.
\end{array}
\end{equation}

\noindent For instance, for $\omega = \frac{1}{2}\ket{R} +
\frac{1}{3}\ket{G} + \frac{1}{6}\ket{B}$ the multinomial distribution
of draws of size three is:
\[ \begin{array}{l}
\multinomial[3](\omega)
\hspace*{\arraycolsep}=\hspace*{\arraycolsep}
\frac{1}{8}\Bigket{3\ket{R}} + 
\frac{1}{4}\Bigket{2\ket{R} + 1\ket{G}} + 
\frac{1}{6}\Bigket{1\ket{R} + 2\ket{G}} + 
\frac{1}{27}\Bigket{3\ket{G}}
\\[+0.2em]
\qquad + \,
\frac{1}{8}\Bigket{2\ket{R} + 1\ket{B}} + 
\frac{1}{6}\Bigket{1\ket{R} + 1\ket{G} + 1\ket{B}} + 
\frac{1}{18}\Bigket{2\ket{G} + 1\ket{B}} 
\\[+0.2em]
\qquad + \,
\frac{1}{24}\Bigket{1\ket{R} + 2\ket{B}} + 
\frac{1}{36}\Bigket{1\ket{G} + 2\ket{B}} + 
\frac{1}{216}\Bigket{3\ket{B}}.
\end{array} \]

\ignore{

X = Space("R", "G", "B")
w = DState([Frac(1,2), Frac(1,3), Frac(1,6)], X)
print( Multinomial(3)(w) )

1/8|3|R>> + 
1/4|2|R> + 1|G>> + 
1/6|1|R> + 2|G>> + 
1/27|3|G>> + 
1/8|2|R> + 1|B>> + 
1/6|1|R> + 1|G> + 1|B>> + 
1/18|2|G> + 1|B>> + 
1/24|1|R> + 2|B>> + 
1/36|1|G> + 2|B>> + 
1/216|3|B>>

}

\noindent Each draw, as multiset of size~3, is inside a big ket, with
the probability of that draw written before the big ket.

\begin{lemma}
\label{JeffreyPointMultinomialLem}
Let $\psi\in\natMlt[K](X)$ be point evidence, where we identify $\psi =
\sum_{x} \psi(x)\ket{x}$ with the multiset of point predicates
$\sum_{x} \psi(x)\bigket{\indic{x}}$.
\begin{enumerate}
\item For a distribution $\omega\in\Dst(X)$,
\[ \begin{array}{rcl}
\omega \Jmodels \psi
& = &
\multinomial[K](\omega)(\psi).
\end{array} \]

\item The distribution that gives the highest Jeffrey validity to the
point evidence $\psi$ is $\flrn(\psi)$.
\end{enumerate}
\end{lemma}

\begin{proof}
\begin{enumerate}
\item We use the equation $\omega\models\indic{x} = \omega(x)$ in:
\[ \begin{array}[b]{rcl}
\omega \Jmodels \psi
& \smash{\stackrel{\eqref{JeffreyValidity}}{=}} &
\displaystyle \coefm{\psi} \cdot \prod_{x\in X} 
   \big(\omega\models \indic{x}\big)^{\psi(x)}
\\[+1.0em]
& = &
\displaystyle \coefm{\psi} \cdot \prod_{x\in X} \omega(x)^{\psi(x)}
\hspace*{\arraycolsep}\,\smash{\stackrel{\eqref{MultinomialEqn}}{=}}\,\hspace*{\arraycolsep}
\multinomial[K](\omega)(\psi).
\end{array} \]

\item This is a standard result, see
  \textit{e.g.}~\cite[Ex.~17.5]{KollerF09}, or the appendix
  of~\cite{Jacobs23c}, which one proves via the Lagrange multiplier
  method (as in the appendix). \QED
\end{enumerate}
\end{proof}

Pearl validity of the form~\eqref{PearlValidity} does not make much
sense for point evidence because it becomes zero as soon as there is
more than one point, since $\indic{x} \andthen \indic{y} = \zero$ when
$x \neq y$. However, in Section~\ref{ChannelSec}, we shall see that
point evidence in the presence of channels does make sense for Pearl
validity.

For future use we show how Jeffrey validity of point evidence is
related to Kullback-Leibler divergence $\DKL$. This divergence is a
standard method for measuring the difference between two
distributions. Its definition (here) uses the natural logarithm
$\ln$. Explicitly, for $\omega,\rho\in\Dst(X)$,
\begin{equation}
\label{KLdivEqn}
\begin{array}{rcl}
\DKL(\omega,\rho)
& \coloneqq &
\displaystyle\sum_{x\in X}\, \omega(x) \cdot \displaystyle
   \ln\left(\frac{\omega(x)}{\rho(x)}\right).
\end{array}
\end{equation}

\noindent The convention is that $r\cdot \ln(r) = 0$ when $r =
0$. This is not a metric, since KL-divergence $\DKL$ is not symmetric:
$\DKL(\omega,\rho) \neq \DKL(\rho,\omega)$, in general. One does have
$\DKL(\omega,\rho) \geq 0$, which can be shown via Jensen's
inequality, and $\DKL(\omega,\rho) = 0$ if and only if $\omega=\rho$.

\begin{lemma}[From~\cite{JacobsS23b}]
\label{KLdivMulnomOrderLem}
Let $\varphi\in\natMlt[K](X)$ be a given multiset of size
$K>0$.  Consider two distributions $\omega,\omega'\in\Dst(X)$ with the
same support as $\varphi$. Then:
\[ \begin{array}{rcl}
\omega\Jmodels\varphi \leq \omega'\Jmodels\varphi
& \Longleftrightarrow &
\multinomial[K]\big(\omega\big)(\varphi)
  \leq \multinomial[K]\big(\omega'\big)(\varphi)
\\[0.2em]
& \Longleftrightarrow &
\DKL\big(\flrn(\varphi), \, \omega\big)
  \geq \DKL\big(\flrn(\varphi), \, \omega'\big).
\end{array} \]

\noindent The first equivalence $\Longleftrightarrow$ is a direct
consequence of Lemma~\ref{JeffreyPointMultinomialLem}. The second
equivalence is the main statement.
\end{lemma}

This result makes perfect sense: the lower the (Jeffrey) validity of
point evidence $\varphi$ is in distribution $\omega$, the higher the
divergence is between the normalised distribution $\flrn(\varphi)$ and
$\omega$. Very briefly: low validity corresponds to high divergence.

\begin{proof}
We use, as first step below, that the logarithm $\ln$ preserves and
reflects the order and that it sends multiplications to sums.
\[ \hspace*{-0.5em}\begin{array}[b]{rcl}
\lefteqn{\multinomial[K]\big(\omega\big)(\varphi)
  \leq \multinomial[K]\big(\omega'\big)(\varphi)}
\\[+0.2em]
& \Longleftrightarrow &
\ln\Big(\multinomial[K]\big(\omega\big)(\varphi)\Big)
  \leq \ln\Big(\multinomial[K]\big(\omega'\big)(\varphi)\Big)
\\[+0.2em]
& \smash{\stackrel{\eqref{MultinomialEqn}}{\Longleftrightarrow}} &
\displaystyle
\ln\left(\coefm{\varphi} \cdot \prod_{x\in X}\, \omega(x)^{\varphi(x)}\right)
   \leq
\ln\left(\coefm{\varphi} \cdot \prod_{x\in X}\, \omega'(x)^{\varphi(x)}\right)
\\[+1.0em]
& \Longleftrightarrow &
\displaystyle
\ln\Big(\coefm{\varphi}\Big) \,+\,
   \sum_{x\in X}\, \varphi(x) \cdot \ln\Big(\omega(x)\Big)
   \leq
\ln\Big(\coefm{\varphi}\Big) \,+\,
   \sum_{x\in X}\, \varphi(x) \cdot \ln\Big(\omega'(x)\Big)
\\
& \Longleftrightarrow &
	\displaystyle
	- \sum_{x\in X}\, \frac{\varphi(x)}{\|\varphi\|} \cdot 
	\ln\big(\omega(x)\big)
	\geq 
	- \sum_{x\in X}\, \frac{\varphi(x)}{\|\varphi\|} \cdot 
	\ln\big(\omega'(x)\big)
\\[+1.2em]
& \Longleftrightarrow &
\displaystyle
\sum_{x\in X}\, \flrn(\varphi)(x) \cdot \ln\Big(\flrn(\varphi)(x)\Big) 
   - \sum_{x\in X}\, \flrn(\varphi)(x) \cdot \ln\big(\omega(x)\big)
\\[+1.2em]
& & \hspace*{2em} \displaystyle \geq\,
\sum_{x\in X}\, \flrn(\varphi)(x) \cdot \ln\Big(\flrn(\varphi)(x)\Big)
  - \sum_{x\in X}\, \flrn(\varphi)(x) \cdot \ln\big(\omega'(x)\big)
\\[+0.8em]
& \Longleftrightarrow &
\displaystyle \sum_{x\in X}\, \flrn(\varphi)(x) \cdot 
\ln\left(\frac{\flrn(\varphi)(x)}{\omega(x)}\right)
\geq 
\sum_{x\in X}\, \flrn(\varphi)(x) \cdot 
   \ln\left(\frac{\flrn(\varphi)(x)}{\omega'(x)}\right)
\\[+1.0em]
& \smash{\stackrel{\eqref{KLdivEqn}}{\Longleftrightarrow}} &
\DKL\big(\flrn(\varphi), \, \omega\big)
   \geq \DKL\big(\flrn(\varphi), \, \omega'\big).
\end{array} \eqno{\QEDbox} \]
\end{proof}

\section{Updating distributions}\label{UpdateSec}

This section first recalls how to update a distribution with a single
predicate, basically as in~\eqref{UpdateKetEqn}. Then it introduces
two update mechanisms, for Jeffrey and for Pearl, about updating a
distribution with evidence, in the form of a multiset of predicates.

\begin{definition}
\label{PredicateUpdateDef}
Let $\omega\in\Dst(X)$ be distribution and $p\colon X \rightarrow
\nnR$ be a factor with non-zero validity $\omega\models p$. We write
$\omega\update{p} \in \Dst(X)$ for the new, Bayesian update of
$\omega$ obtained as:
\begin{equation}
\label{UpdateEqn}
\begin{array}{rcl}
\omega\update{p}
& \coloneqq &
\displaystyle\sum_{x\in X}\, \frac{\omega(x)\cdot p(x)}{\omega\models p}
   \bigket{x}.
\end{array}
\end{equation}

\noindent Whenever we write $\omega\update{p}$ we implicitly assume
$\omega\models p \neq 0$.
\end{definition}

We have seen an example of how this updating works
in~\eqref{PostestUpdateEqn}, by incorporating the predicate into the
distribution.

\begin{example}
\label{PhysicsEx}
We can now describe the physical model from
Subsection~\ref{PhysicsSubsec} in a mathematically precise
manner. There we started with a distribution $\omega =
\frac{1}{2}\ket{L} + \frac{1}{3}\ket{M} + \frac{1}{6}\ket{R}$
corresponding to the division of water flows over three pipes, labeled
as left ($L$), middle ($M$) and right ($R$).

During the first update, the middle pipe was blocked. This can be
captured by the sharp predicate $\big(\indic{M}\big)^{\bot} =
\indic{\{L,R\}}$ corresponding to the subset $\{L,R\} \subseteq
\{L,M,R\}$ of open pipes. This sharp predicate has validity:
\[ \begin{array}{rcccccccl}
\omega \models \indic{\{L,R\}}
& = &
\displaystyle\sum_{x\in\{L,M,R\}} \omega(x) \cdot \indic{\{L,R\}}(x)
& = &
\omega(L) + \omega(R)
& = &
\frac{1}{2} + \frac{1}{6}
& = &
\frac{2}{3}.
\end{array} \]

\noindent Updating with this predicate yields the distribution written
as $\omega'$ in Subsection~\ref{PhysicsSubsec}. Explicitly:
\[ \begin{array}{rcl}
\omega'
\hspace*{\arraycolsep}\coloneqq\hspace*{\arraycolsep}
\omega\update{\indic{\{L,R\}}}
& \smash{\stackrel{\eqref{UpdateEqn}}{=}} &
\displaystyle\sum_{x\in\{L,M,R\}} 
   \frac{\omega(x) \cdot \indic{\{L,R\}}(x)}{\omega \models \indic{\{L,R\}}}
   \ket{x}
\\[+1.4em]
& = &
\displaystyle
\frac{\nicefrac{1}{2}\cdot 1}{\nicefrac{2}{3}}\ket{L} + 
   \frac{\nicefrac{1}{3}\cdot 0}{\nicefrac{2}{3}}\ket{M} + 
   \frac{\nicefrac{1}{6}\cdot 1}{\nicefrac{2}{3}}\ket{R}
\hspace*{\arraycolsep}=\hspace*{\arraycolsep}
\textstyle\frac{3}{4}\ket{L} + \frac{1}{4}\ket{R}.
\end{array} \]

\noindent The second distribution $\omega''$ in
Subsection~\ref{PhysicsSubsec} is obtained by updating $\omega$ with
the fuzzy predicate $p = \frac{2}{3}\cdot\indic{L} +
\frac{1}{3}\cdot\indic{M} + \frac{1}{2}\cdot\indic{R}$. It captures
the openness of the three taps. This leads to a validity
$\omega\models p = \frac{19}{36}$. It is called a normalisation factor
in Subsection~\ref{PhysicsSubsec}.  The calculation of $\omega''$
given there fits the above definition~\eqref{UpdateEqn} of
the Bayesian update.
\end{example}

Bayesian updating satisfies several basic, standard properties that
are easy to prove.

\begin{lemma}
\label{PredicateUpdateLem}
Let $\omega\in\Dst(X)$ and $p,q\in\Fact(X)$ be factors.
\begin{enumerate}
\item \label{PredicateUpdateLemConj} $\omega\update{\one} = \omega$ and
  $\omega\update{p\andthen q} = \omega\update{p}\update{q}$; as a result, the order
  of multiple updates is irrelevant;

\item \label{PredicateUpdateLemPoint} $\omega\update{\indic{x}} =
  1\ket{x}$, for each $x\in\supp(\omega)$;

\item \label{PredicateUpdateLemScal} $\omega\update{s\cdot p} =
  \omega\update{p}$ for a scalar $s > 0$;

\item \label{PredicateUpdateLemProd} The product rule holds:
\[ \begin{array}{rcl}
\omega\update{p} \models q
& \,=\, &
\displaystyle\frac{\omega \models p\andthen q}{\omega\models p}.
\end{array} \]

\item \label{PredicateUpdateLemBayes} Bayes' rule holds:
\[ \begin{array}{rcl}
\omega\update{p} \models q
& \,=\, &
\displaystyle\frac{(\omega\update{q} \models p) \cdot (\omega \models q)}
   {\omega\models p}.
\end{array} \]

\item \label{PredicateUpdateLemIter} For a sequence of factors $p_{1},
  \ldots, p_{n} \in \Fact(X)$ one has:
\[ \begin{array}{rcl}
\lefteqn{\omega \otimes \big(\omega\update{p_1}\big) \otimes 
   \big(\omega\update{p_1}\update{p_2}\big) \otimes \cdots \otimes
   \big(\omega\update{p_1}\cdots\update{p_{n-1}}\big) \models 
   p_{1}\otimes \cdots \otimes p_{n}}
\\[+0.2em]
& = &
\omega \models p_{1} \andthen \cdots \andthen p_{n}. \hspace*{18em}
\\[+0.2em]
& = &
\frac{1}{n!}\cdot\Big(\omega \Pmodels\, \sum_{i} 1\ket{p_i}\Big).
\end{array} \]
\end{enumerate}
\end{lemma}

The last item generalises what we have seen in
Remark~\ref{PearlMedicalUpdateRem}. Recall
from~\eqref{JeffreyValidity} that the Jeffrey validity of the factors
in the last item takes the form $\omega \otimes \cdots \otimes \omega
\models p_{1} \otimes \cdots \otimes p_{n} = \prod_{i} \omega\models
p_{i}$. It thus involves the `independent' validity of each of the
factors $p_{i}$ in the same distribution $\omega$. The above last
item~\eqref{PredicateUpdateLemIter} shows the difference with Pearl
validity, which involves `dependent' validity of $p_{i+1}$ in the
distribution that is updated with the previous factors $p_{1}, \ldots,
p_{i}$. Thus, at state $i+1$ one takes into account what has already
been learned at previous stages $1,\ldots,i$.

\begin{proof}
We only prove the last item~\eqref{PredicateUpdateLemIter} since the
other ones are standard (see e.g.~\cite[Lemma~5]{JacobsZ21}) and easy
to show. We use induction on $n$, where the base case for $n=1$ is
obvious. Next, by Lemma~\ref{ValidityLem}~\eqref{ValidityLemTensor} and
Lemma~\ref{PredicateUpdateLem}~\eqref{PredicateUpdateLemConj},
\eqref{PredicateUpdateLemProd}:
\[ \hspace*{-2.5em}\begin{array}[b]{rcl}
\lefteqn{\omega \otimes \big(\omega\update{p_1}\big) \otimes \cdots \otimes
   \big(\omega\update{p_1}\cdots\update{p_{n-1}}\big) \otimes
   \big(\omega\update{p_1}\cdots\update{p_{n}}\big) \models 
   p_{1}\otimes \cdots \otimes p_{n} \otimes p_{n+1} }
\\[+0.2em]
& = &
\Big(\omega \otimes \big(\omega\update{p_1}\big) \otimes 
   \cdots \otimes
   \big(\omega\update{p_1}\cdots\update{p_{n-1}}\big) \models 
   p_{1}\otimes \cdots \otimes p_{n}\Big) \cdot
   \Big(\omega\update{p_1}\cdots\update{p_{n}}\models \rlap{$p_{n+1}\Big)$}
\\[+0.2em]
& \smash{\stackrel{\text{(IH)}}{=}} &
\Big(\omega \models p_{1} \andthen \cdots \andthen p_{n}\Big) \cdot
   \Big(\omega\update{p_{1} \andthen \cdots \andthen p_{n}} \models 
   p_{n+1}\Big)
\\[+0.2em]
& = &
\omega \models p_{1} \andthen \cdots \andthen p_{n} \andthen p_{n+1}.
\end{array} \eqno{\QEDbox} \]
\end{proof}

We now formulate two general results about the gains of learning.
They will be applied later for the rules of Jeffrey and Pearl. The
formulation below occurs in~\cite{Jacobs23c} and is proven there. For
completeness we reproduce the proof in the appendix.

\begin{theorem}
\label{UpdateGainThm}
Let $\omega\in\Dst(X)$ be a distribution.
\begin{enumerate}
\item \label{UpdateGainThmSingle} For a factor $p$ on $X$ with
  non-zero validity $\omega\models p$ one has:
\[ \begin{array}{rcl}
\omega\update{p} \models p
& \,\geq\, &
\omega\models p.
\end{array} \]

\item \label{UpdateGainThmSeq} Let $p_{1}, \ldots, p_{n}$ be factors
  on $X$, all with non-zero validity $\omega\models p_{i}$. The
  uniform convex sum of updates $\omega' \coloneqq \sum_{i}
  \frac{1}{n}\cdot\omega\update{p_i}$ satisfies:
\[ \begin{array}{rcl}
{\displaystyle\prod}_{i}\, (\omega' \models p_{i})
& \,\geq\, &
{\displaystyle\prod}_{i}\, (\omega \models p_{i}).
\end{array} \eqno{\QEDbox} \]
\end{enumerate}
\end{theorem}

The first item expresses that in a Bayesian update of a distribution
with evidence / factor $p$ this $p$ becomes `more true' in the updated
posterior distribution, than in the original prior distribution. This
is a benefit that one expects from learning.  The second item is more
subtle and expresses an increase of multiplied validities.

We now give general formulations for the Jeffrey and Pearl update
mechanisms, with evidence. We have used the vertical bar $|$ for
Bayesian updating in the notation $\omega\update{p}$. For the Jeffrey,
and Pearl versions we make subtle variations: we change the vertical
bar $|$ into $\smash{\JupdateSign}$ for the Jeffrey update and into
$\smash{\PupdateSign}$ for the Pearl update, referring respectively to
the letters `J' and `P'.

\begin{definition}
\label{EvicenceUpdateDef}
Let $\omega\in\Dst(X)$ be a distribution with evidence
$\psi\in\natMlt\big(\Fact(X)\big)$.
\begin{enumerate}
\item \label{EvicenceUpdateDefJeffrey} The Jeffrey update $\omega
  \Jupdate{\psi}\in \Dst(X)$ of the distribution $\omega$ with the
  evidence $\psi$ is defined as a convex sum of Bayesian updates with
  separate predicates:
\begin{equation}
\label{EvicenceUpdateDefJeffreyEqn}
\begin{array}{rcl}
\omega \Jupdate{\psi}
& \coloneqq &
\displaystyle\sum_{p\in\supp(\psi)} \, \flrn(\psi)(p) \cdot \omega\update{p}.
\end{array}
\end{equation}

\item \label{EvicenceUpdateDefPearl} The Pearl update $\omega
  \Pupdate{\psi}\in \Dst(X)$ is the Bayesian update with a single
  factor, namely with the conjunction $\evidand\psi =
  \bigandthen_{p\in\supp(\psi)} p^{\psi(p)}$, as in:
\begin{equation}
\label{EvicenceUpdateDefPearlEqn}
\begin{array}{rcccl}
\omega \Pupdate{\psi}
& \coloneqq &
\omega\update{\evidand\psi}
& = &
\omega\update{\andthen_{p}\, p^{\psi(p)}}.
\end{array}
\end{equation}
\end{enumerate}
\end{definition}

For singleton evidence $1\ket{p}$ one has $\omega\Jupdate{1\ket{p}} =
\omega\Pupdate{1\ket{p}} = \omega\update{p}$, so Jeffrey and Pearl
updates coincide in this trivial case with a Bayesian update. The
differences between the update modes arise for multiple pieces of
evidence.

\begin{figure}
\begin{center}
\includegraphics[scale=0.5]{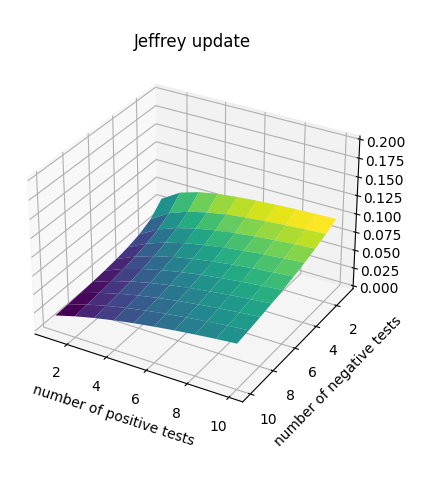}
\hspace*{2em}
\includegraphics[scale=0.5]{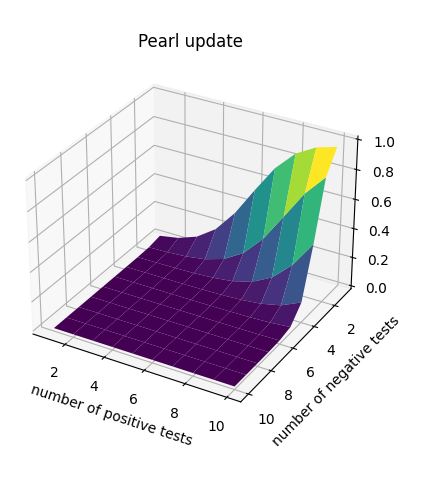}
\end{center}
\caption{Posterior disease probabilities in the medical test scenario
  from the introduction, for the two update mechanisms of Jeffrey and
  Pearl, with 100 different versions of evidence $i\bigket{\postest} +
  j\ket{\negtest}$, where $i,j\in\{1,\ldots,10\}$.}
\label{UpdateFig}
\end{figure}

Equations~\eqref{JeffreyMedUpdate} and~\eqref{PearlMedUpdate} provide
illustrations of these two update mechanisms.  Figure~\ref{UpdateFig}
contains plots of the updated (posterior) disease probabilities, for
Jeffrey and Pearl, using evidence of the form $i\bigket{\postest} +
j\ket{\negtest}$, with $i$ positive and $j$ negative tests, for all
values $i,j\in\{1,\ldots,10\}$. We see that in comparison with the
validities in Figure~\ref{ValidityFig} there are now considerable
differences between Jeffrey and Pearl.  In the Jeffrey case the plot
is rather flat, since the posterior disease probability is a convex
combination of the form:
\[ \begin{array}{rcl}
\lefteqn{\textstyle \frac{i}{i+j}\cdot \omega\update{\postest}(d) +
   \frac{j}{i+j}\cdot \omega\update{\negtest}(d),
   \qquad \mbox{for $1 \leq i,j \leq 10$, 
   see~\eqref{EvicenceUpdateDefJeffreyEqn}}}
\\[+0.3em]
& = &
\frac{i}{i+j}\cdot \frac{9}{85} + \frac{j}{i+j}\cdot \frac{1}{115}
   \hspace*{4em}\mbox{see~\eqref{PostestUpdateEqn} and \eqref{NegtestUpdateEqn}.}
\end{array} \]

In the Pearl case, on the right in Figure~\ref{UpdateFig}, we see more
variation. The relevant formula is not as simple as in the Jeffrey
case. The Pearl posterior disease probability is highest when the
evidence contains a high number $i$ of positive tests and a low number
$j$ of negative tests. This makes sense.

We proceed with some key validity increase --- and other ---
properties of these updates with evidence.  We begin with the Jeffrey
case.

\begin{theorem}
\label{JeffreyEvicenceUpdateThm}
Let $\omega\in\Dst(X)$ be a distribution with evidence
$\psi\in\natMlt\big(\Fact(X)\big)$.
\begin{enumerate}
\item \label{JeffreyEvicenceUpdateThmInc} Jeffrey updating increases
  the Jeffrey validity:
\[ \begin{array}{rcl}
\omega\Jupdate{\psi} \Jmodels \psi
& \,\geq\, &
\omega \Jmodels \psi.
\end{array} \]

\item \label{JeffreyEvicenceUpdateThmPoint} In the special case that
  $\psi$ is point evidence, inhabiting $\natMlt(X)$, Jeffrey updating
  gives a reduction of KL-divergence:
\[ \begin{array}{rcl}
\DKL\big(\flrn(\psi), \, \omega\Jupdate{\psi}\big)
& \,\leq\, &
\DKL\big(\flrn(\psi), \, \omega\big).
\end{array} \]

\item \label{JeffreyEvicenceUpdateThmScal} More of the same evidence
  $\psi\in\natMlt\big(\Fact(X)\big)$ does not change Jeffrey updating:
  for each $n\geq 1$,
\[ \begin{array}{rclcrcl}
\omega\Jupdate{n\cdot\psi}
& = &
\omega \Jupdate{\psi}
& \qquad\mbox{where}\qquad &
n\cdot\psi
& = &
\psi + \cdots + \psi.
\end{array} \]

\item \label{JeffreyEvicenceUpdateThmOrder} Jeffrey updating is
  sensitive to the order of updating: for other evidence
  $\chi\in\natMlt\big(\Fact(X)\big)$, in general,
\[ \begin{array}{rcl}
\omega\Jupdate{\psi}\,\Jupdate{\chi}
& \,\neq\, &
\omega\Jupdate{\chi}\,\Jupdate{\psi}.
\end{array} \]

\item \label{JeffreyEvicenceUpdateThmNothing} Let $\omega$ be
  `fractional', that is, the probablities in $\omega$ are fractions,
  and let $N\in\NNO$ be such that $N\cdot\omega$ is a
  multiset. Consider this $N\cdot\omega \in \natMlt(X)$ as point
  evidence. Then, updating with this point evidence has no effect:
\[ \begin{array}{rcl}
\omega\Jupdate{N\cdot\omega}
& = &
\omega.
\end{array} \]

\item \label{JeffreyEvicenceUpdateThmConv} A fractional convex sum of
  Jeffrey updates can be turned into a single Jeffrey update: for
  numbers $n_{i}\in\NNO$ with $\sum_{i} n_{i} = n$ and for evidence
  multisets $\psi_{i}\in\natMlt\big(\Fact(X)\big)$,
\[ \begin{array}{rcl}
\frac{n_1}{n}\cdot \omega\Jupdate{\psi_1} + \cdots +
   \frac{n_L}{n}\cdot \omega\Jupdate{\psi_L}
& = &
\omega\Jupdate{\psi}
\qquad\mbox{for }\; \psi = n_{1}\cdot\psi_{1} + \cdots + n_{L}\cdot\psi_{L}.
\end{array} \]
\end{enumerate}
\end{theorem}

Item~\eqref{JeffreyEvicenceUpdateThmNothing} says that you learn
nothing if you use the information from the distribution $\omega$ as
evidence. More informally, you learn nothing in Jeffrey style from
what you already know. The fact that $\omega$ is required to be
fractional is not essential. It is needed because we have defined
multisets of evidence to have only natural numbers as frequencies. One
can drop this limitation and allow non-negative real numbers as
frequencies. Then one can use the distribution $\omega\in\Dst(X)$ as a
multiset of point predicates itself --- without multiplication with
the number $N$ --- and write $\omega\Jupdate{\omega} = \omega$. The
same thing applies in item~\eqref{JeffreyEvicenceUpdateThmConv}.

\begin{proof}
Let the evidence multinomial $\psi\in\natMlt\big(\Fact(X)\big)$ have size $K$.
\begin{enumerate}
\item We write the multiset $\psi$ as accumulation $\psi = \acc(p_{1},
  \ldots, p_{K})$ of a sequence of factors, including multiple
  occurrences, with each $p\in\supp(\psi)$ occurring $\psi(p)$ many
  times in this sequence $p_{1}, \ldots, p_{K}$.
  Theorem~\ref{UpdateGainThm}~\eqref{UpdateGainThmSeq} tells us that:
\[ \begin{array}{rcl}
\omega' \Jmodels \psi
\hspace*{\arraycolsep}\smash{\stackrel{\eqref{JeffreyValidity}}{=}}\hspace*{\arraycolsep}
\displaystyle\coefm{\psi}\cdot\prod_{p\in\supp(\psi)} (\omega'\models p)^{\psi(p)}
& = &
\coefm{\psi}\cdot{\displaystyle\prod}_{i} \, (\omega'\models p_{i})
\\[-0.5em]
& \geq &
\coefm{\psi}\cdot{\displaystyle\prod}_{i} \, (\omega\models p_{i})
\\[+0.5em]
& = &
\displaystyle\coefm{\psi}\cdot\prod_{p\in\supp(\psi)} (\omega\models p)^{\psi(p)}
\hspace*{\arraycolsep}\smash{\stackrel{\eqref{JeffreyValidity}}{=}}\hspace*{\arraycolsep}
\omega \Jmodels \psi,
\end{array} \]

\noindent where, as defined in
Theorem~\ref{UpdateGainThm}~\eqref{UpdateGainThmSeq}:
\[ \begin{array}{rcccccccl}
\omega'
& \coloneqq &
{\displaystyle\sum}_{i} \, \frac{1}{K}\cdot \omega\update{p_i}
& = &
\displaystyle\prod_{p\in\supp(\psi)} \textstyle\frac{\psi(p)}{K}\cdot\omega\update{p}
& = &
\displaystyle\prod_{p\in\supp(\psi)} \flrn(\psi)(p)\cdot\omega\update{p}
& = &
\omega\Jupdate{\psi}.
\end{array} \]

\item This is obtained by combining the validity increase of the
  previous point with Lemma~\ref{JeffreyPointMultinomialLem}
  and~\ref{KLdivMulnomOrderLem}:
\[ \begin{array}{rcccccl}
\multinomial[K]\big(\omega\Jupdate{\psi}\big)(\psi)
& = &
\omega\Jupdate{\psi}\Jmodels\psi 
& \,\geq\, &
\omega\Jmodels\psi
& = &
\multinomial[K]\big(\omega\big)(\psi)
\end{array} \]

\noindent It gives the required divergence decrease:
\[ \begin{array}{rcl}
\DKL\big(\flrn(\psi), \, \omega\Jupdate{\psi}\big)
& \,\leq\, & 
\DKL\big(\flrn(\psi), \, \omega\big).
\end{array} \]

\item Since $\flrn\big(n\cdot\psi\big) = \flrn(\psi)$.

\item We give an illustration showing that the two orders of updating
  are not equal. We re-use the situation described in the
  introduction, with prior $\omega = \frac{1}{20}\ket{d} +
  \frac{19}{20}\ket{\no{d}}$ and positive and negative test predicates
  $\postest$ and $\negtest$. We take the original evidence $\psi =
  2\ket{\postest} + 1\ket{\negtest}$ with new evidence $\chi =
  1\ket{\postest} + 2\ket{\negtest}$.  Then:
\[ \begin{array}{rclcrcl}
\omega\Jupdate{\psi}\,\Jupdate{\chi}
& = &
0.059\ket{d} + 0.941\ket{\no{d}}
& \quad\mbox{and}\quad &
\omega\Jupdate{\chi}\,\Jupdate{\psi}
& = &
0.061\ket{d} + 0.939\ket{\no{d}}.
\end{array} \]

\ignore{

evidence1 = DState([2,1], Space(postest, negtest))
evidence2 = DState([1,2], Space(postest, negtest))
print("\nPearl val: ", prior.Pearl_validity(evidence1) )
print("\nJeffrey val: ", prior.Jeffrey_validity(evidence1) )
print("\nPearl upd: ", prior.Pearl_update(evidence1) )
print("\nJeffrey upd: ", prior.Jeffrey_update(evidence1) )
print("\nJeffrey two updates:")
print( prior.Jeffrey_update(evidence1).Jeffrey_update(evidence2) )
print( prior.Jeffrey_update(evidence2).Jeffrey_update(evidence1) )

# Jeffrey two updates:
# 0.0592|d> + 0.941|~d>
# 0.061|d> + 0.939|~d>

}

\item Suppose we can write $\omega = \sum_{i} \frac{n_{i}}{N}\ket{x_i}$, with
$n_{i}\in\NNO$ and $N = \sum_{i} n_{i}$. Then $N\cdot \omega$ is the
multiset $\sum_{i} n_{i}\ket{x_i}$, which we identify with the point evidence
$\sum_{i} n_{i}\bigket{\indic{x_i}}$. Hence, by Lemma~\ref{PredicateUpdateLem}~\eqref{PredicateUpdateLemPoint},
\[ \begin{array}{rcccccl}
\omega\Jupdate{N\cdot\omega}
& = &
{\displaystyle\sum}_{i}\, \flrn\big(N\cdot\omega\big)(x_{i}) \cdot
   \omega\update{\indic{x_i}}
& = &
{\displaystyle\sum}_{i}\, \frac{n_i}{N}\cdot 1\ket{x_i}
& = &
\omega.
\end{array} \eqno{\QEDbox} \]
\end{enumerate}
\end{proof}

We turn to Pearl's style of updating.

\begin{theorem}
\label{PearlEvicenceUpdateThm}
Let $\omega\in\Dst(X)$ be a distribution with evidence
$\psi,\chi\in\natMlt\big(\Fact(X)\big)$.
\begin{enumerate}
\item \label{PearlEvicenceUpdateThmInc} Pearl updating increases
  the Pearl validity:
\[ \begin{array}{rcl}
\omega\Pupdate{\psi} \Pmodels \psi
& \,\geq\, &
\omega \Pmodels \psi.
\end{array} \]

\item \label{PearlEvicenceUpdateThmProd} The `product' and `Bayes'
  analogues of
  Proposition~\ref{PredicateUpdateLem}~\eqref{PredicateUpdateLemProd},
  \eqref{PredicateUpdateLemBayes} hold for Pearl updating:
\[ \begin{array}{rcccl}
\omega\Pupdate{\psi} \Pmodels \chi
& = &
\displaystyle \frac{\omega \Pmodels \psi + \chi}{\omega \Pmodels \psi}
& = &
\displaystyle \frac{(\omega\Pupdate{\chi} \Pmodels \psi)\cdot
   (\omega \Pmodels \chi)}{\omega \Pmodels \psi},
\end{array} \]

\noindent where $\psi+\chi$ is the pointwise sum of multisets.

\item \label{PearlEvicenceUpdateThmComp} Pearl updates compose, as in:
\[ \begin{array}{rcl}
\omega\Pupdate{\psi}\Pupdate{\chi}
& \,=\, &
\omega\Pupdate{\psi+\chi}.
\end{array} \]

\item \label{PearlEvicenceUpdateThmOrder} Pearl updating is
  insensitive to the order of updating:
\[ \begin{array}{rcl}
\omega\Pupdate{\psi}\Pupdate{\chi}
& \,=\, &
\omega\Pupdate{\chi}\Pupdate{\psi}.
\end{array} \]

\item \label{PearlEvicenceUpdateThmNothing} One learns nothing via Pearl
  updates with uniform factors, that is, with scalars $s\cdot\one$ of
  the truth predicate, for $s > 0$. Explicitly,
\[ \begin{array}{rcl}
\omega\Pupdate{\psi}
& \,=\, &
\omega \qquad\mbox{for $\;\psi = {\displaystyle\sum}_{i}\, 
   \psi(i)\bigket{s_{i}\cdot \one}$ with $s_{i}>0$.}
\end{array} \]
\end{enumerate}
\end{theorem}

\begin{proof}
\begin{enumerate}
\item Directly by Theorem~\ref{UpdateGainThm}~\eqref{UpdateGainThmSingle}:
\[ \begin{array}{rcccccl}
\omega\Pupdate{\psi} \Pmodels \psi
& \,=\, &
\omega\update{\evidand\psi} \models \evidand\psi
& \,\geq\, &
\omega \models \evidand\psi
& \,=\, &
\omega \Pmodels \psi.
\end{array} \]

\item The first thing to note is that $\evidand\big(\psi + \chi\big) =
  \evidand(\psi) \andthen \evidand(\chi)$. This uses $p^{n+m} = p^{n}
  \andthen p^{m}$. Then, following
  Definition~\ref{EvidenceConjunctionDef}~\eqref{EvidenceConjunctionDefEv}
  one has:
\[ \begin{array}{rcl}
\evidand\big(\psi + \chi\big)
\hspace*{\arraycolsep}=\hspace*{\arraycolsep}
\displaystyle\bigandthen_{p\in\Fact(X)}\, p^{(\psi+\chi)(p)}
& = &
\displaystyle\bigandthen_{p\in\Fact(X)}\, p^{\psi(p)+\chi(p)}
\\
& = &
\displaystyle\bigandthen_{p\in\Fact(X)}\, p^{\psi(p)} \andthen p^{\chi(p)}
\\
& = &
\displaystyle\bigandthen_{p\in\Fact(X)}\, p^{\psi(p)} \andthen 
   \displaystyle\bigandthen_{p\in\Fact(X)}\, p^{\chi(p)}
\\
& = &
\evidand(\psi) \andthen \evidand(\chi).
\end{array} \]

\noindent Now, using Proposition~\ref{PredicateUpdateLem}~\eqref{PredicateUpdateLemProd}:
\[ \begin{array}{rcl}
\omega\Pupdate{\psi} \Pmodels \chi
\hspace*{\arraycolsep}=\hspace*{\arraycolsep}
\omega\update{\evidand(\psi)} \models \evidand(\chi)
& = &
\displaystyle \frac{\omega\models\evidand(\psi) \andthen \evidand(\chi)}
   {\omega \models \evidand(\psi)}
\\[+0.8em]
& = &
\displaystyle \frac{\omega\models\evidand(\psi+\chi)}
   {\omega \models \evidand(\psi)}
\hspace*{\arraycolsep}=\hspace*{\arraycolsep}
\displaystyle \frac{\omega \Pmodels \psi + \chi}{\omega \Pmodels \psi}.
\end{array} \]

\noindent Similarly, using
Proposition~\ref{PredicateUpdateLem}~\eqref{PredicateUpdateLemBayes},
one obtains:
\[ \begin{array}{rcl}
\omega\Pupdate{\psi} \Pmodels \chi
\hspace*{\arraycolsep}=\hspace*{\arraycolsep}
\omega\update{\evidand(\psi)} \models \evidand(\chi)
& = &
\displaystyle \frac{\omega\models\evidand(\psi) \andthen \evidand(\chi)}
   {\omega \models \evidand(\psi)}
\\[+0.8em]
& = &
\displaystyle \frac{(\omega\update{\evidand(\chi)}\models\evidand(\psi)
   \cdot (\omega\models\evidand(\chi))}
   {\omega \models \evidand(\psi)}
\\[+0.8em]
& = &
\displaystyle \frac{(\omega\Pupdate{\chi} \Pmodels \psi)\cdot
   (\omega \Pmodels \chi)}{\omega \Pmodels \psi}.
\end{array} \]

\item We again use the equation $\evidand\big(\psi + \chi\big) =
  \evidand(\psi) \andthen \evidand(\chi)$, now together with
Proposition~\ref{PredicateUpdateLem}~\eqref{PredicateUpdateLemConj}, in:
\[ \begin{array}{rcccccccl}
\omega\Pupdate{\psi}\Pupdate{\chi}
& = &
\omega\update{\evidand(\psi)}\update{\evidand(\chi)}
& = &
\omega\update{\evidand(\psi) \andthen \evidand(\chi)}
& = &
\omega\update{\evidand(\psi + \chi)}
& = &
\omega\Pupdate{\psi+\chi}.
\end{array} \]

\item We use the previous point and the commutativity of addition $+$
of multisets, in:
\[ \begin{array}{rcccccl}
\omega\Pupdate{\psi}\Pupdate{\chi}
& = &
\omega\Pupdate{\psi+\chi}
& = &
\omega\Pupdate{\chi+\psi}
& = &
\omega\Pupdate{\chi}\Pupdate{\psi}.
\end{array} \]

\item Let the evidence be of the form $\psi = \sum_{i}\,
  \psi(i)\bigket{s_{i}\cdot \one}$ with $s_{i}>0$. Then $\evidand\psi
  = \mathop{\andthen_{i}} (s_{i}\cdot\one)^{\psi(i)} = s\cdot\one$, where $s =
  \prod_{i} s_{i}^{\psi(i)} > 0$. Then, by
  Proposition~\ref{PredicateUpdateLem}~\eqref{PredicateUpdateLemScal},
  \eqref{PredicateUpdateLemConj}, in:
\[ \begin{array}{rcccccccl}
\omega\Jupdate{\psi}
& = &
\omega\update{\evidand\psi}
& = &
\omega\update{s\cdot\one}
& = &
\omega\update{\one}
& = &
\omega.
\end{array} \eqno{\QEDbox} \]
\end{enumerate}
\end{proof}

We thus see that X-updating increases X-validity for $X \in
\big\{\mbox{Jeffrey}, \mbox{Pearl}\big\}$, see
Theorem~\ref{JeffreyEvicenceUpdateThm}~\eqref{JeffreyEvicenceUpdateThmInc}
and
Theorem~\ref{PearlEvicenceUpdateThm}~\eqref{PearlEvicenceUpdateThmInc}.
Such learning makes us appropriately wiser.  At the end of
Section~\ref{MedicalSolutionSec} we have seen that mixing validity and
update of different kinds may lead to a decrease of validity ---
making us less wise. We can now formulate this more precisely: for the
prior disease distribution $\omega = \frac{1}{20}\ket{d} +
\frac{19}{20}\ket{\no{d}}$ and the three-test evidence $\psi =
2\ket{\postest} + 1\ket{\negtest}$ one has validity decreases after
(the wrong kind of) updates:
\[ \begin{array}{rcccccl}
\omega\Pupdate{\psi}\Jmodels\psi
& = &
0.3081
& < &
0.3116
& = &
\omega\Jmodels\psi
\\[+0.2em]
\omega\Jupdate{\psi}\Pmodels\psi
& = &
0.2847
& < &
0.2858
& = &
\omega\Pmodels\psi.
\end{array} \]

\ignore{

print("\nDecreases via crossed use")
print( prior.Pearl_update(evidence1).Jeffrey_validity(evidence1),
       " < ", prior.Jeffrey_validity(evidence1) )
print( prior.Jeffrey_update(evidence1).Pearl_validity(evidence1),
       " < ", prior.Pearl_validity(evidence1) )

# Decreases via crossed use
# 0.3081094331480003  <  0.31157812499999993
# 0.28469309462915604  <  0.28575000000000006

}

A final question that we consider in this section is: what happens if
one tries to update a distribution $\omega\in\Dst(X)$ with
inconsistent evidence? We take this inconsistent evidence to be of the
form $\psi = 1\bigket{\indic{U}} + 1\bigket{\indic{\neg U}}$, with two
sharp predicates $\indic{U}$ and $\indic{\neg U}$, for a non-empty
subset $U\subseteq \supp(\omega)$ with non-empty complement $\neg U =
\supp(\omega) \setminus U$. The Jeffrey update with such inconsistent
evidence yields a convex combination of the two updates:
\begin{equation}
\label{InconsistentJeffreyEqn}
\begin{array}{rcl}
\omega\Jupdate{\psi}
& = &
\frac{1}{2}\cdot \omega\update{\indic{U}} +
   \frac{1}{2}\cdot \omega\update{\indic{\neg U}}.
\end{array}
\end{equation}

\noindent This can be interpreted as a `superposition' of the two
updated distributions. To be explicit, the
update $\omega\update{\indic{U}}$ is the normalised restriction of
$\omega$ to the subset $U$, as in:
\[ \begin{array}{rcccl}
\omega\update{\indic{U}}
& \smash{\stackrel{\eqref{UpdateEqn}}{=}} &
\displaystyle \sum_{x \in U} \frac{\omega(x)}{\omega\models\indic{U}}\,\bigket{x}
& = &
\displaystyle \sum_{x \in U} \frac{\omega(x)}{\sum_{y\in U}\omega(y)}\,\bigket{x}.
\end{array} \]

\noindent Thus, the two distributions $\omega\update{\indic{U}}$ and
$\omega\update{\indic{\neg U}}$ have disjoint supports, but happily
live together in the Jeffrey update~\eqref{InconsistentJeffreyEqn}.
This reinforces the idea that Jeffrey's approach involves a convex
combination of independent updates.

The Pearl update would be of the form $\omega\update{\indic{U}
  \andthen \indic{\neg U}} =
\omega\update{\indic{U}}\update{\indic{\neg U}}$.  We notice that the
conjunction $\andthen\psi = \indic{U} \andthen \indic{\neg U} =
\indic{U\cap \neg U} = \indic{\emptyset} = \zero$ yields
falsity. Updating with the predicate $\zero$ is impossible, since its
validity is $0$.

\section{Channels}\label{ChannelSec}

At this stage we have seen the basics of validity and updating
according to Jeffrey and Pearl, and how these updates make us
appropriately wiser. In the set up so far we have used a distribution
$\omega\in\Dst(X)$ on a set $X$, and evidence
$\psi\in\natMlt\big(\Fact(X)\big)$ on this same set $X$. We now
generalise this situation to allow evidence on a different set, say
$Y$. This is a common situation in which information on $Y$ is
observable, and information on $X$ is hidden, often called latent. The
connection between the sets $X$ and $Y$ happens via what is called a
channel, forming a generative model. This section outlines the role
that these channels play in probabilistic updating.

Channels formalise the idea of conditional probabilities. They carry a
rich mathematical structure that can be used in compositional
reasoning, with both sequential and parallel composition and with
reversal. The concept of a channel has emerged in various forms,
namely as conditional probability, stochastic matrix, probabilistic
classifier, Markov kernel, statistical model, conditional probability
table (in Bayesian network), probabilistic function / computation, and
finally as Kleisli map (in category theory).

We begin with the definition of a channel and with the associated
forward and backward transformation mechanisms.

\begin{definition}
\label{ChannelDef}
Let $X,Y$ be arbitrary sets.
\begin{enumerate}
\item A channel $c$ from $X$ to $Y$ is a function of the form $c
  \colon X \rightarrow \Dst(Y)$. It assigns a distribution $c(x) \in
  \Dst(Y)$ to an arbitrary element $x\in X$, and thus corresponds to a
  conditional probability $\Prob(y|x)$. We write such a channel as
  $c\colon X \chanto Y$, with a circle on the shaft of the arrow.

\item For a channel $c\colon X \chanto Y$ and a distribution
  $\omega\in\Dst(X)$ we can form the pushforward distribution $c \push
  \omega$ on $Y$, as:
\begin{equation}
\label{ChannelForwardEqn}
\begin{array}{rcl}
c \push \omega
& \coloneqq &
\displaystyle \sum_{y\in Y} \left(\sum_{x\in X}\,\omega(x)\cdot c(x)(y)\right)
   \bigket{y} \,\in\,\Dst(Y).
\end{array}
\end{equation}

\noindent This pushforward $c \push \omega$ is also called the
prediction, see Example~\ref{MedicalChannelEx} below.

\item For a channel $c\colon X \chanto Y$ and an observation $q\colon Y
  \rightarrow \R$ we can form the pullback observation $c \pull q$ on
  $X$, via:
\begin{equation}
\label{ChannelBackwardEqn}
\begin{array}{rcl}
\big(c \pull q\big)(x)
& \coloneqq &
\displaystyle \sum_{y\in Y} c(x)(y)\cdot q(y).
\end{array}
\end{equation}

\noindent This pullback operation $c \pull (-)$ restricts to factors
and to fuzzy predicates, but not to sharp predicates.

Using point predicates we can equivalently write (when the
set $X$ is finite):
\begin{equation}
\label{ChannelBackwardPointEqn}
\begin{array}{rcl}
c \pull q
& \coloneqq &
\displaystyle \sum_{x\in X} \left(\sum_{y\in Y} c(x)(y)\cdot q(y)\right)
   \cdot\indic{x}.
\end{array}
\end{equation}
\end{enumerate}
\end{definition}

The following easy result shows how forward and backward transformation
are closely related via validity. The proof follows from unravelling
the relevant definitions.

\begin{lemma}
\label{TransformationValidityLem}
For a channel $c\colon X \chanto Y$ with a distribution
$\omega\in\Dst(X)$ on its domain $X$ and an observation $q\in\Obs(Y)$
on its codomain one has the following equality of validities.
\begin{equation}
\label{TransformationValidityEqn}
\begin{array}{rcl}
\big(c \push \omega\big) \models q
& \,=\, &
\omega \models \big(c \pull q).
\end{array} 
\end{equation}
\end{lemma}

Hidden in the medical test description in the introduction there is a
channel. We make it explicit below. In this situation the
distributions have only two elements in their support, so the
situation may look rather trivial.  But channels may involve much more
complicated distributions, where the approach of the above definition
provides a clear general methodology.

\begin{example}
\label{MedicalChannelEx}
We will redescribe the medical test from the introduction as a channel
$c\colon D \chanto T$, where $D = \{d,\no{d}\}$ is the disease set and
$T = \{p,n\}$ is the set of test outcomes. This test channel $c$ is
determined by the sensitivity of $\frac{9}{10}$ and specificity of
$\frac{3}{5}$, namely as:
\[ \begin{array}{rclcrcl}
c(d)
& = &
\frac{9}{10}\ket{p} + \frac{1}{10}\ket{n}
& \qquad\mbox{and}\qquad &
c(\no{d})
& = &
\frac{2}{5}\ket{p} + \frac{3}{5}\ket{n}.
\end{array} \]

\noindent One sees that in the first case, in presence of the disease
$d$, the channel outcome $c(d)$ gives a $\frac{9}{10}$ chance of a
positive test. Similarly, in absence of the disease $\no{d}$, the
distribution $c(\no{d})$ gives a $\frac{3}{5}$ chance of a negative
test.

We assumed a prevalence of $5\%$, corresponding to the prior disease
distribution $\omega = \frac{1}{20}\ket{d} +
\frac{19}{20}\ket{\no{d}}$. The predicted test outcome distribution is
the pushforward $c \push \omega$ of the prior along the channel.  It
is computed as:
\[ \begin{array}{rcl}
c \push \omega
& \smash{\stackrel{\eqref{ChannelForwardEqn}}{=}} &
\displaystyle\left(\sum_{x\in D} \omega(x)\cdot c(x)(p)\right)\ket{p}
   + \left(\sum_{x\in D} \omega(x)\cdot c(x)(n)\right)\ket{n}
\\[+1.2em]
& = &
\Big(\frac{1}{20}\cdot\frac{9}{10} + \frac{19}{20}\cdot\frac{2}{5}\Big)\ket{p}
+
\Big(\frac{1}{20}\cdot\frac{1}{10} + \frac{19}{20}\cdot\frac{3}{5}\Big)\ket{n}
\\[+0.6em]
& = &
\frac{17}{40}\ket{p} + \frac{23}{40}\ket{n}.
\end{array} \]

Backward transformation of (point) predicates is also relevant in this
example, since in this way we rediscover the positive / negative test
predicates $\postest$, $\negtest$ on $D$ that we used before:
\[ \begin{array}{rcl}
c \pull \indic{p}
& \smash{\stackrel{\eqref{ChannelBackwardPointEqn}}{=}} &
\displaystyle \left(\sum_{y\in T} c(d)(y)\cdot \indic{p}(y)\right)\cdot\indic{d}
   + \left(\sum_{y\in T} c(\no{d})(y)\cdot \indic{p}(y)\right)\cdot\indic{\no{d}}
\\[+1.4em]
& = &
c(d)(p)\cdot\indic{d} + c(\no{d})(p)\cdot\indic{\no{d}}
\hspace*{\arraycolsep}=\hspace*{\arraycolsep}
\frac{9}{10}\cdot\indic{d} + \frac{2}{5}\cdot\indic{\no{d}}
\hspace*{\arraycolsep}=\hspace*{\arraycolsep}
\postest
\\[+0.4em]
c \pull \indic{n}
& \smash{\stackrel{\eqref{ChannelBackwardPointEqn}}{=}} &
\displaystyle \left(\sum_{y\in T} c(d)(y)\cdot \indic{n}(y)\right)\cdot\indic{d}
   + \left(\sum_{y\in T} c(\no{d})(y)\cdot \indic{n}(y)\right)\cdot\indic{\no{d}}
\\[+1.4em]
& = &
c(d)(n)\cdot\indic{d} + c(\no{d})(n)\cdot\indic{\no{d}}
\hspace*{\arraycolsep}=\hspace*{\arraycolsep}
\frac{1}{10}\cdot\indic{d} + \frac{3}{5}\cdot\indic{\no{d}}
\hspace*{\arraycolsep}=\hspace*{\arraycolsep}
\negtest.
\end{array} \]
\end{example}

In order to be able to speak of validity and updating along a channel
we transform evidence along this channel. Then we can easily define
channel-based validity and updating.

\begin{definition}
\label{ChannelEvidenceDef}
Let $c\colon X \chanto Y$ be a channel, with a distribution
$\omega\in\Dst(X)$ on its domain $X$ and with evidence
$\psi\in\natMlt\big(\Fact(Y)\big)$ on its codomain $Y$.
\begin{enumerate}
\item \label{ChannelEvidenceDefTriplePull} We write $c \triplepull
  \psi$ for the evidence on $X$ obtained by transforming factor-wise
  along $c$. Thus:
\[ \begin{array}{rcl}
c \triplepull \psi
& \coloneqq &
\displaystyle\sum_{q\in\Fact(Y)} \, \psi(q)\bigket{c \pull q}
   \,\in\,\natMlt\big(\Fact(X)\big).
\end{array} \]

\noindent This can also be described via the functoriality of
$\natMlt$, applied to the function $c\pull(-) \colon \Fact(Y)
\rightarrow \Fact(X)$. 

\item \label{ChannelEvidenceDefVal} The Jeffrey validity of
  $\psi$ in $\omega$ along the channel $c$ is now defined as:
\begin{equation}
\label{JeffreyValidityAlong}
\begin{array}{rcl}
\omega \Jmodels c \triplepull \psi
& \,=\, &
\displaystyle\coefm{\psi}\cdot\prod_{q\in\Fact(Y)} \, 
   \big(\omega\models c \pull q\big)^{\psi(q)}
\\[+1em]
& \smash{\stackrel{\eqref{TransformationValidityEqn}}{=}} &
\displaystyle\coefm{\psi}\cdot\prod_{q\in\Fact(Y)} \, 
   \big(c \push \omega\models q\big)^{\psi(q)}
\\[+1em]
& = &
c \push \omega \Jmodels \psi.
\end{array}
\end{equation}

\noindent By construction, this Jeffrey validity satisfies an analogue
of Lemma~\ref{TransformationValidityLem}.

Similarly, the Pearl validity along the channel is defined as:
\begin{equation}
\label{PearlValidityAlong}
\begin{array}{rcl}
\omega \Pmodels c \triplepull \psi
& \,=\, &
\displaystyle\coefm{\psi}\cdot \Big(\omega \models 
   \evidand\big(c \triplepull \psi\big)\Big)
\\[+0.5em]
& = &
\displaystyle\coefm{\psi}\cdot \Big(\omega \models
   \bigandthen_{q\in\Fact(Y)} \, (c \pull q)^{\psi(q)}\Big).
\end{array}
\end{equation}

\noindent (This is \emph{not} equal to $c \push \omega \Pmodels \psi$
since transformation $c \pull (-) \colon \Fact(Y) \rightarrow
\Fact(X)$ does not preserve conjunctions $\andthen$.)

\item \label{ChannelEvidenceDefUpdate} Along the same lines, the
  Jeffrey update along a channel with evidence $\psi$ is:
\[ \begin{array}{rcl}
\omega\Jupdate{c \triplepull \psi}
& = &
\displaystyle\sum_{q\in\Fact{Y}} \, \flrn(\psi)(q)\cdot \omega\update{c \pull q}.
\end{array} \]

\noindent And the Pearl update along a channel is:
\[ \begin{array}{rcl}
\omega\Pupdate{c \triplepull \psi}
& = &
\omega\update{\andthen_{q} (c \pull q)^{\psi(q)}}.
\end{array} \]
\end{enumerate}
\end{definition}

We can now check that the Jeffrey and Pearl updates
in~\eqref{JeffreyMedUpdate} and \eqref{PearlMedUpdate} are updates
along the test channel $c\colon D \chanto T$ from
Example~\ref{MedicalChannelEx}, using the point evidence
$2\ket{\indic{p}} + 1\ket{\indic{n}}$.

The special case of updating with point evidence along a channel is
worth a closer look. The reformulations in terms of multinomials come
from~\cite{JacobsS23b}. The formulation of Jeffrey's update rule via a
dagger channel occurs in~\cite{Jacobs19c} and the associated decrease
of KL-divergence in
item~\eqref{ChannelUpdatePointEvidencePropJeffreyUpdate} was first
identified in~\cite{Jacobs21c}. In~\cite{JacobsS23b} one can also find
a description of Jeffrey updating along a channel via variational
inference.

\begin{proposition}
\label{ChannelUpdatePointEvidenceProp}
Let $c\colon X \chanto Y$ be a channel, with a distribution
$\omega\in\Dst(X)$ and with point evidence $\psi\in\natMlt[K](Y)$, so that $c
\triplepull \psi$ is $\sum_{y} \psi(y)\bigket{c \pull \indic{y}}$. Then:
\begin{enumerate}
\item \label{ChannelUpdatePointEvidencePropJeffreyValidity} The
  Jeffrey validity of $\psi$ along $c$ can be expressed as multinomial
  probability:
\[ \begin{array}{rcl}
\omega\Jmodels c \triplepull \psi
& \,=\, &
\multinomial[K]\big(c \push \omega\big)(\psi).
\end{array} \]

\item \label{ChannelUpdatePointEvidencePropPearlValidity} The Pearl
  validity of $\psi$ along $c$ is:
\[ \begin{array}{rcl}
\omega\Jmodels c \triplepull \psi
& \,=\, &
\big(\multinomial[K](c) \push \omega\big)(\psi),
\end{array} \]

\noindent where $\multinomial[K](c) \coloneqq \multinomial[K] \after c
\colon X \rightarrow \Dst(Y) \rightarrow \Dst\big(\natMlt[K](Y)\big)$.

\item \label{ChannelUpdatePointEvidencePropJeffreyUpdate} The Jeffrey
  update of $\omega$ with point evidence $\psi$ along $c$ can be
  described as:
\[ \begin{array}{rcl}
\omega\Jupdate{c \triplepull \psi}
& = &
c^{\dag}_{\omega} \push \flrn(\psi),
\end{array} \]

\noindent where $c^{\dag}_{\omega} \colon Y \chanto X$ is the
reversed `dagger' channel defined by $c^{\dag}_{\omega}(y) =
\omega\update{c \pull \indic{y}}$. Abbreviating $\omega' \coloneqq
c^{\dag}_{\omega} \push \flrn(\psi)$, we have associated validity
increases and divergence decreases:
\[ \begin{array}{rcl}
\multinomial[K]\big(c \push \omega'\big)(\psi)
\hspace*{\arraycolsep}\,=\,\hspace*{\arraycolsep}
c\push\omega'\Jmodels \psi
& \,\geq\, &
c\push\omega\Jmodels \psi
\hspace*{\arraycolsep}\,=\,\hspace*{\arraycolsep}
\multinomial[K]\big(c \push \omega\big)(\psi)
\\[+0.2em]
\DKL\big(\flrn(\psi), \, c \push \omega'\big)
& \,\leq\, &
\DKL\big(\flrn(\psi), \, c \push \omega\big).
\end{array} \]

\item \label{ChannelUpdatePointEvidencePropPearlUpdate} The Pearl
  update of $\omega$ with $\psi$ along the channel $c$ satisfies:
\[ \begin{array}{rcccl}
\omega\Pupdate{c \triplepull \psi}
& = &
\omega\update{\multinomial[K](c) \pull \indic{\psi}}
& = &
\multinomial[K](c)^{\dag}_{\omega}(\psi).
\end{array} \]

\noindent Writing $\omega' = \omega\update{\multinomial[K](c) \pull
  \indic{\psi}}$ we get:
\[ \begin{array}{rcccccl}
\big(\multinomial[K](c) \push \omega'\big)(\psi)
& \,=\, &
\omega'\Jmodels c \triplepull \psi
& \,\geq\, &
\omega\Jmodels c \triplepull \psi
& \,=\, &
\big(\multinomial[K](c) \push \omega\big)(\psi).
\end{array} \]

\end{enumerate}
\end{proposition}

\begin{proof}
\begin{enumerate}
\item First, the Jeffrey validity can be rewritten as:
\[ \begin{array}{rcl}
\omega\Jmodels c \triplepull \psi
& \smash{\stackrel{\eqref{JeffreyValidityAlong}}{=}} &
\displaystyle \coefm{\psi} \cdot
   \prod_{y\in Y}\, \big(c \push \omega \models \indic{y}\big)^{\psi(y)}
\\[+1em]
& = &
\displaystyle \coefm{\psi} \cdot
   \prod_{y\in Y}\, \big((c \push \omega)(y)\big)^{\psi(y)}
\hspace*{\arraycolsep}\,\smash{\stackrel{\eqref{MultinomialEqn}}{=}}\,\hspace*{\arraycolsep}
\multinomial[K]\big(c \push \omega\big)(\psi).
\end{array} \]

\item For the Pearl validity we get:
\[ \begin{array}{rcl}
\omega\Pmodels c \triplepull \psi
& \smash{\stackrel{\eqref{PearlValidityAlong}}{=}} &
\displaystyle \coefm{\psi} \cdot \Big(\omega \models 
   \bigandthen_{y\in Y}\, (c \push \indic{y})^{\psi(y)}\Big)
\\[+1em]
& = &
\displaystyle \coefm{\psi} \cdot \sum_{x\in X}\,
   \omega(x) \cdot \prod_{y\in Y}\, (c \pull \indic{y})(x)^{\psi(y)}
\\[+1em]
& = &
\displaystyle \sum_{x\in X}\,
   \omega(x) \cdot \coefm{\psi} \cdot \prod_{y\in Y}\, c(x)(y)^{\psi(y)}
\\[+1em]
& \smash{\stackrel{\eqref{MultinomialEqn}}{=}} &
\displaystyle \sum_{x\in X}\,
   \omega(x) \cdot \multinomial[K]\big(c(x)\big)(\psi)
\\[+0.4em]
& = &
\big(\multinomial[K](c) \push \omega\big)(\psi),
   \qquad \mbox{where }\multinomial[K](c) = \multinomial[K] \after c.
\end{array} \]

\item We move to the Jeffrey update and compute:
\[ \begin{array}{rcl}
\omega\Jupdate{c \triplepull \psi}
& \smash{\stackrel{\eqref{EvicenceUpdateDefJeffreyEqn}}{=}} &
\displaystyle\sum_{y\in Y} \, \flrn(\psi)(y) \cdot 
   \omega\update{c \pull \indic{y}}
\\[+1em]
& = &
\displaystyle\sum_{y\in Y} \, \flrn(\psi)(y) \cdot c^{\dag}_{\omega}(y)
\hspace*{\arraycolsep}=\hspace*{\arraycolsep}
c^{\dag}_{\omega} \push \flrn(\psi).
\end{array} \]

\noindent The validity increases and divergence decreases mentioned
above follow directly from
Theorem~\ref{JeffreyEvicenceUpdateThm}~\eqref{JeffreyEvicenceUpdateThmInc},
\eqref{JeffreyEvicenceUpdateThmPoint}.

\item For Pearl update with point evidence we have:
\[ \begin{array}{rcl}
\omega\update{\multinomial[K](c) \pull \indic{\psi}}
& \smash{\stackrel{\eqref{UpdateEqn}}{=}} &
\displaystyle\sum_{x\in X}\, \frac{\omega(x)\cdot 
   (\multinomial[K](c) \pull \indic{\psi})(x)}
   {\omega \models \multinomial[K](c) \pull \indic{\psi}}\,\bigket{x}
\\[+1em]
& = &
\displaystyle\sum_{x\in X}\, \frac{\omega(x)\cdot \multinomial[K](c(x))(\psi)}
   {(\multinomial[K](c) \push \omega)(\psi)}\,\bigket{x}
\\[+1em]
& = &
\displaystyle\sum_{x\in X}\, \frac{\omega(x)\cdot \coefm{\psi} \cdot
   \prod_{y} c(x)(y)^{\psi(y)}}
   {\sum_{x} \omega(x) \cdot \coefm{\psi} \cdot 
      \prod_{y} c(x)(y)^{\psi(y)}}\,\bigket{x}
\\[+1em]
& = &
\displaystyle\sum_{x\in X}\, \frac{\omega(x)\cdot 
   (\andthen_{y} (c \pull \indic{y})^{\psi(y)})(x)}
   {\sum_{x} \omega(x) \cdot (\andthen_{y} (c \pull \indic{y})^{\psi(y)})(x)}
   \,\bigket{x}
\\[+1em]
& \smash{\stackrel{\eqref{UpdateEqn}}{=}} &
\omega\update{\andthen_{y} (c \pull \indic{y})^{\psi(y)}}
\\[+0.2em]
& = &
\omega\Pupdate{c \triplepull \psi}.
\end{array} \]

\noindent The claimed increase of multinomial probability follows from
Theorem~\ref{PearlEvicenceUpdateThm}~\eqref{PearlEvicenceUpdateThmInc}. \QED

\end{enumerate}
\end{proof}

\section{A glance at updating in predictive coding}\label{PredictiveCodingSec}


The naive picture of human learning involves a teacher pouring
knowledge into a student's brain --- as visually expressed by what is
called a Nuremberg Funnel. One more modern approach in (computational)
cognitive science is called predictive coding (or processing), see
\textit{e.g.}~the books~\cite{Clark16,Hohwy13,ParrPF22} or
articles~\cite{AllenF18,FristonK09}.  Very briefly, the idea is that
the human mind projects, evaluates and updates. Predictive coding
describes the mind basically as a prediction engine that compares its
predictions to observations, leading to internal adaptations. To quote
Friston~\cite{Friston09}: ``The Bayesian brain hypothesis uses
Bayesian probability theory to formulate perception as a constructive
process based on internal or generative models. [\ldots] In this view,
the brain is an inference machine that actively predicts and explains
its sensations. Central to this hypothesis is a probabilistic model
that can generate predictions, against which sensory samples are
tested to update beliefs about their causes.'' This testing involves
error reduction: ``\ldots the core function of the brain is simply to
minimize prediction error, where the prediction errors signal
mismatches between predicted input and the input actually
received.''~\cite{MillidgeSB22}.  Mathematically, these prediction
errors can be expressed in terms of KL-divergence~\eqref{KLdivEqn},
see \textit{e.g.}~\cite{MillidgeSB22,Penny12}. The aimed decrease of
prediction errors / KL-divergence is a form of what we call getting
wiser. Here we concentrate on this error correction aspect of
predictive coding. There is however, a wider story, especially about
active inference~\cite{ParrPF22}.

When the human mind is understood as a Bayesian inference engine, the
question comes up: how does the mind handle multiple observations, via
the rules of Jeffrey or Pearl? Predictive coding emphasises error
correction, which, in terms of reducing KL-divergence, is achieved via
the update rule of Jeffrey (for point evidence), see
Theorem~\ref{JeffreyEvicenceUpdateThm}~\eqref{JeffreyEvicenceUpdateThmPoint}
and
Proposition~\ref{ChannelUpdatePointEvidenceProp}~\eqref{ChannelUpdatePointEvidencePropJeffreyUpdate}. This
suggests that the mind is a Jeffreyan update engine. A further
indication is that Jeffrey's updating is sensitive to the order of
updating, see
Theorem~\ref{JeffreyEvicenceUpdateThm}~\eqref{JeffreyEvicenceUpdateThmOrder}.
As is well-known, humans are very sensitive to the order in which they
process information (are `primed'). On the other hand, recall that in
the medical example in the beginning we associated Pearl's approach
with a clinical perspective, where tests are applied to the same
person. This fits a Pearlian perspective in predictive coding where
updates apply (dependently) to the same mind, of a single individual.
In predictive coding the updates of Jeffrey and Pearl are not
(explicitly) used, but an approximation of Jeffrey, called VFE update,
where VFE is an abbreviation of variational free energy. Below we
briefly put this VFE rule in the setting of predictive coding, based
on~\cite{Jacobs21c,TullKS23}.



Updating in predictive coding happens via variational inference, in
order to minimalise free energy~\cite{Friston09}. This is an
approximation technique, which in this case produces VFE updating as
an approximation of Jeffrey updating. As will be shown, VFE does not
inherit the crucial KL-divergence reduction from Jeffrey, but it does
increase Pearl validity --- although less than Pearl updating
does. One can conclude that both Jeffrey and Pearl updating outperform
VFE updating.

Later on, in Remark~\ref{TullRem} we go into the details of what free
energy amounts to in the setting of predictive coding. At this stage
we first provide the description of the VFE update, in
item~\eqref{EnergyUpdatePropExp} below. It uses a softmax description,
involving normalisation of $e$-powers. This is used more widely in
cognitive modeling, see \textit{e.g.}~\cite{StuhlmullerG14}. The
description that we use is based on~\cite[Eqn.~(44)]{TullKS23}. The
subsequent minimality characterisation in
item~\eqref{EnergyUpdatePropArg} occurs (in essence)
in~\cite{Penny12,TullKS23}.


\begin{definition}
\label{EnergyUpdateDef}
Let $\omega\in\Dst(X)$ be a distribution with evidence
$\psi\in\natMlt\big(\Fact(X)\big)$. The VFE update
$\omega\Fupdate{\psi}$ is defined as normalised `softmax' of the form:
\begin{equation}
\label{EnergyUpdateEqn}
\begin{array}{rcl}
\omega\Fupdate{\psi}
& \coloneqq &
\displaystyle\flrn\left(\sum_{x\in\supp(\omega)} \,
   e^{\flrn(\psi) \models \ln(\omega\update{(-)}(x))}\bigket{x}\right).
\end{array}
\end{equation}

\noindent The letter `F' in the sign $\smash{\FupdateSign}$ refers to
`free' in free energy.
\end{definition}

We immediately give an alternative description of VFE updating,
together with its characterisation in terms of minimal free energy.
In item~\eqref{EnergyUpdatePropForm} below we stretch earlier notation
since we use a predicate power $p^{f}$ for a fraction $f$ and not for
a natural number $n$, as orignally in
Definition~\ref{EvidenceConjunctionDef}~\eqref{EvidenceConjunctionDefFact}.
But the meaning is the same, given pointwise as $p^{f}(x) = p(x)^{f}$.
Similarly, we use the conjunction $\evidand$ not only for multisets
$\evidand(\psi)$ but also for distributions
$\evidand\big(\flrn(\psi)\big)$, where:
\[ \begin{array}{rcccl}
\evidand\big(\flrn(\psi)\big)(x)
& = &
\displaystyle\prod_{p\in\supp(\psi)} p^{\flrn(\psi)(p)}(x)
& = &
\displaystyle\prod_{p\in\supp(\psi)} p(x)^{\flrn(\psi)(p)}.
\end{array} \]

\begin{proposition}
\label{EnergyUpdateProp}
In the setting of Definition~\ref{EnergyUpdateDef}.
\begin{enumerate}
\item \label{EnergyUpdatePropForm} The VFE
  update~\eqref{EnergyUpdateEqn} can alternatively be described as:
\begin{equation}
\label{EnergyUpdateAltEqn}
\begin{array}{rcccl}
\omega\Fupdate{\psi}
& = &
\omega\update{\andthen_{p}\, p^{\flrn(\psi)(p)}}
& = &
\omega\update{\evidand(\flrn(\psi))}.
\end{array}
\end{equation}

\item \label{EnergyUpdatePropArg} The VFE update can be characterised
  as a minimum below. This is the minimum free energy, see
  Remark~\ref{TullRem} below.
\begin{equation}
\label{EnergyUpdatePropArgEqn}
\begin{array}{rcl}
\omega\Fupdate{\psi}
& \in &
\argmin\limits_{\rho\in\Dst(X)} \, \flrn(\psi) \models
   \DKL\big(\rho, \, \omega\update{(-)}\big)
\\[+0.4em]
& & \hspace*{\arraycolsep}=\hspace*{\arraycolsep}
\displaystyle\argmin\limits_{\rho\in\Dst(X)} \, \sum_{p\in\supp(\psi)} \,
   \flrn(\psi)(p) \cdot \DKL\big(\rho, \, \omega\update{p}\big).
\end{array}
\end{equation}

\noindent Here we consider $\argmin$ as giving a subset of outcomes,
since there may be multiple distributions with the same minimum value.
\end{enumerate}
\end{proposition}

Notice that the VFE update~\eqref{EnergyUpdateAltEqn} looks very much
like the Pearl update $\omega \Pupdate{\psi}$
in~\eqref{EvicenceUpdateDefPearlEqn}, which does not involve
fractional powers: $\omega \Pupdate{\psi} = \omega\update{\andthen_{p}
  p^{\psi(p)}} = \omega\update{\evidand\psi}$. The VFE and Pearl
update behave similarly, see Theorem~\ref{VFEUpdateThm} below.

\begin{proof}
\begin{enumerate}
\item Using the familiar equations $e^{a+b} = e^{a}\cdot e^{b}$ and $e^{a\cdot b} = 
\big(e^{a}\big)^{b}$ we observe that we can write the $e$-expression
in~\eqref{EnergyUpdateEqn} as:
\[ \begin{array}{rcl}
\lefteqn{e^{\flrn(\psi) \models \ln(\omega\update{(-)}(x))}}
\\[+0.2em]
& = &
e^{\sum_{p\in\supp(\psi)} \flrn(\psi)(p) \cdot \ln(\omega\update{p}(x))}
\\
& = &
\displaystyle\prod_{p\in\supp(\psi)} \, 
   e^{\flrn(\psi)(p) \cdot \ln(\omega\update{p}(x))}
\\
& = &
\displaystyle\prod_{p\in\supp(\psi)} \, 
   \Big(e^{\ln(\omega\update{p}(x))}\Big)^{\flrn(\psi)(p)}
\\
& = &
\displaystyle\prod_{p\in\supp(\psi)} \, 
   \Big(\omega\update{p}(x)\Big)^{\flrn(\psi)(p)}
\\
& \smash{\stackrel{\eqref{UpdateEqn}}{=}} &
\displaystyle\prod_{p\in\supp(\psi)} \, 
   \left(\frac{\omega(x)\cdot p(x)}{\omega\models p}\right)^{\flrn(\psi)(p)}
\\[+1.4em]
& = &
\displaystyle \omega(x)^{\sum_{p\in\supp(\psi)} \flrn(\psi)(p)} \cdot 
   \prod_{p\in\supp(\psi)} \, p(x)^{\flrn(\psi)(p)} \cdot
   \frac{1}{\prod_{p\in\supp(\psi)} (\omega\models p)^{\flrn(\psi)(p)}}
\\[+1.4em]
& \smash{\stackrel{\eqref{JeffreyValidity}}{=}} &
\displaystyle \omega(x) \cdot \big(\evidand \flrn(\psi)\big)(x) \cdot
   \left(\frac{\coefm{\psi}}{\omega\Jmodels \psi}\right)^{\nicefrac{1}{\|\psi\|}}.
\end{array} \]

Using this formulation we obtain Equation~\eqref{EnergyUpdateEqn}:
\[ \begin{array}{rcl}
\lefteqn{\flrn\left(\sum_{x\in\supp(\omega)} \,
   e^{\flrn(\psi) \models \ln(\omega\update{(-)}(x))}\bigket{x}\right)(z)}
\\[+1.4em]
& = &
\displaystyle \frac{e^{\flrn(\psi) \models \ln(\omega\update{(-)}(z))}}
   {\sum_{y}\, e^{\flrn(\psi) \models \ln(\omega\update{(-)}(y))}}
\\[+1.4em]
& = &
\displaystyle\frac{\omega(z) \cdot (\evidand \flrn(\psi))(z)}
    {\sum_{y}\, \omega(y) \cdot (\evidand \flrn(\psi))(y)}
   \qquad\mbox{ as just shown}
\\[+1.2em]
& \smash{\stackrel{\eqref{UpdateEqn}}{=}} &
\omega\update{\evidand\flrn(\psi)}(z)
\\[+0.4em]
& \smash{\stackrel{\eqref{EnergyUpdateAltEqn}}{=}} &
\omega\Fupdate{\psi}(z).
\end{array} \]

\item Let the normalisation factor in~\eqref{EnergyUpdateEqn} be $s$,
  so that $s = \sum_{x} e^{\flrn(\psi) \models
    \ln(\omega\update{(-)}(x))}$. For an arbitrary distribution
  $\rho\in\Dst(X)$ we compute the value of the expression
  in~\eqref{EnergyUpdatePropArgEqn}.
\[ \begin{array}{rcl}
\lefteqn{\flrn(\psi) \models \DKL\big(\rho, \omega\update{(-)}\big)}
\\
& = &
\displaystyle \sum_{p\in\supp(\psi)} \, \flrn(\psi)(p) \cdot \sum_{x\in X}\,
   \rho(x)\cdot\ln\left(\frac{\rho(x)}{\omega\update{p}(x)}\right)
\\
& = &
\displaystyle \sum_{x\in X} \rho(x)\cdot\left[\ln\big(\rho(x)\big) - 
   \sum_{p\in\supp(\psi)} \, \flrn(\psi)(p) \cdot 
   \ln\big(\omega\update{p}(x)\big)\right]
\\[+1.4em]
& = &
\displaystyle \sum_{x\in X} \rho(x)\cdot\left[\ln\big(\rho(x)\big) - 
   \flrn(\psi) \models \ln\big(\omega\update{(-)}(x)\big)\right]
\\[+1em]
& = &
\displaystyle \sum_{x\in X} \rho(x)\cdot\left[\ln\big(\rho(x)\big) - 
   \ln\big(e^{\flrn(\psi) \models \ln\big(\omega\update{(-)}(x)\big)}\big)\right]
\\[+1em]
& \smash{\stackrel{\eqref{EnergyUpdateEqn}}{=}} &
\displaystyle \sum_{x\in X} \rho(x)\cdot\Big[\ln\big(\rho(x)\big) - 
   \ln\big(s\cdot \omega\Fupdate{\psi}(x)\big)\Big]
\\[+0.2em]
& = &
\DKL\big(\rho, \, \omega\Fupdate{\psi}\big) - \ln(s).
\end{array} \]

\noindent The latter expression reaches its minimum when the
divergence $\DKL\big(\rho, \, \omega\Fupdate{\psi}\big) \geq 0$ is
actually $0$, that is, when the distributions $\rho$ and
$\omega\Fupdate{\psi}$ are equal. \QED
\end{enumerate}

\auxproof{
We first compute that the value for
  $\omega\Fupdate{\psi}$ in the expression
  in~\eqref{EnergyUpdatePropArgEqn} is $-\ln(s)$.
\[ \hspace*{-0.5em}\begin{array}{rcl}
\lefteqn{\flrn(\psi) \models \DKL\big(\omega\Fupdate{\psi}, \omega\update{(-)}\big)}
\\
& = &
\displaystyle \sum_{p\in\supp(\psi)} \, \flrn(\psi)(p) \cdot \sum_{x\in X}\,
   \omega\Fupdate{\psi}(x)\cdot
   \ln\left(\frac{\omega\Fupdate{\psi}(x)}{\omega\update{p}(x)}\right)
\\
& \smash{\stackrel{\eqref{EnergyUpdateEqn}}{=}} &
\displaystyle \sum_{p\in\supp(\psi), x\in X} \flrn(\psi)(p) \cdot 
   \omega\Fupdate{\psi}(x)\cdot\left[
   \ln\left(\frac{e^{\flrn(\psi) \models \ln(\omega\update{(-)}(x))}}{s}\right) - 
   \ln\big(\omega\update{p}(x)\big)\right]
\\[+1.6em]
& = &
\displaystyle \sum_{p\in\supp(\psi), x\in X} \flrn(\psi)(p) \cdot 
   \omega\Fupdate{\psi}(x)\cdot\Big(\big(\flrn(\psi) \models 
      \ln(\omega\update{(-)}(x))\big) - 
   \ln(s)\Big)
\\
& & \qquad \displaystyle -\;
   \sum_{p\in\supp(\psi), x\in X} \flrn(\psi)(p) \cdot 
      \omega\Fupdate{\psi}(x)\cdot \ln\big(\omega\update{p}(x)\big)
\\[+1.4em]
& = &
\displaystyle -\ln(s) + \left(\sum_{x\in X}\,\omega\Fupdate{\psi}(x) \cdot
   \big(\flrn(\psi) \models \ln(\omega\update{(-)}(x))\big)\right)
\\
& & \qquad \displaystyle -\;
   \left(\sum_{x\in X} \omega\Fupdate{\psi}(x)\cdot 
   \big(\flrn(\psi) \models \ln(\omega\update{(-)}(x))\big)\right)
\\
& = &
-\ln(s).
\end{array} \]
}
\end{proof}

Via the formulation~\eqref{EnergyUpdateAltEqn} we see that more of the
same evidence $\psi\in\natMlt\big(\Fact(X)\big)$ does not change VFE
updating: $\omega\Fupdate{n\cdot\psi} = \omega \Fupdate{\psi}$,for
each $n\geq 1$. Simlarly, VFE updating is insensitive to the order of
updating: $\omega\Pupdate{\psi}\Pupdate{\chi} \,=\,
\omega\Pupdate{\chi}\Pupdate{\psi}$.

\begin{remark}
\label{TullRem}
It may not be immediately clear that the VFE update defined
in~\eqref{EnergyUpdateEqn} corresponds to what is used in predictive
coding. Therefore we elaborate this relationship more explicitly,
using the setting of~\cite{TullKS23}. There, two sets $S$ and $O$ are
used for hidden elements and for observations in a generative
model. It is given by a joint distribution $M\in\Dst(S\times O)$,
together with a distribution $o\in\Dst(O)$ of observations. The latter
corresponds to point evidence in our situation, so that we can think
of it as $o = \flrn(\psi)$ for point evidence $\psi\in\natMlt(O)$.

The joint distribution $M\in\Dst(S\times O)$ may be
disintegrated~\cite{ChoJ19} into a form $M = \tuple{\idmap, c} \push
\omega$, for a channel $c \colon S \chanto O$ and a distribution
$\omega\in\Dst(S)$. We write $d$ for the reversal / dagger of the
channel $c$, of the form $d = c_{\omega}^{\dag} \colon O \chanto S$,
so that $d(y)(s) = \big(\omega\update{c \pull \indic{y}}\big)(s)$, see
Proposition~\ref{ChannelUpdatePointEvidenceProp}~\eqref{ChannelUpdatePointEvidencePropJeffreyUpdate}. In~\cite{TullKS23}
this dagger $d$ is written as $M(s|y)$.

The VFE-update rule considered in~\cite{TullKS23}, in Equations~(44)
and~(45), produces a posterior, via normalisation, written as:
\begin{equation}
\label{VFEoriginalEqn}
\displaystyle\NormOp\limits_{s} \, \left(e^{\expec_{y\sim o} \ln(M(s|y))}\right).
\end{equation}

\noindent In our setting we write frequentist learning $\flrn$ for the
normalisation and validity $\models$ for the expectation
$\expec_{x\sim o}$. This turns~\eqref{VFEoriginalEqn} into:
\[ \begin{array}{rcl}
\displaystyle\flrn\left(\sum_{s\in S} \, e^{o\models \ln(d(-)(s))}\,\bigket{s}\right)
& = &
\displaystyle\flrn\left(\sum_{s\in S} \, 
   e^{o\models \ln\big((\omega\update{c\pull \indic{(-)}})(s)\big)}\,\bigket{s}\right).
\end{array} \]

\noindent The VFE update~\eqref{EnergyUpdateEqn} is a generalisation
of the latter formula, in the following way.  The evidence that is
used for updating $\omega$ is not of the pullback form $c \pull
\indic{(-)}$, but is a general factor $p$.  This factor comes from the
evidence $o$, which is generalised from point evidence, to arbitrary
evidence $o = \flrn(\psi)$, for $\psi\in\natMlt\big(\Fact(O)\big)$.

Having explained the formula~\eqref{EnergyUpdateEqn}, we turn to the
associated free energy formulation. It can also be translated to the
current setting, again starting from~\cite{TullKS23}. There, the free
energy of a distribution $\rho\in\Dst(S)$, is defined with respect to
the generative model $M\in\Dst(S\times O)$ and the distribution of
observations $o\in\Dst(O)$. It involves the dagger $M(s|y)$ that we
mentioned before and takes the explicit form:
\[ \begin{array}{rcccccccccl}
\big(\omega|_{c \pull \indic{y}}\big)(s)
& = &
c_{\omega}^{\dag}(y)(s)
& = &
d(y)(s)
& = &
M(s|y)
& = &
\displaystyle\frac{M(s,y)}{M(y)}
& = &
\displaystyle\frac{M(s,y)}{(c \push \omega)(y)}.
\end{array} \]

\noindent What is then called free energy of $\rho$ is computed as:
\[ \begin{array}{rcl}
\lefteqn{\sum_{s\in S, \, y\in O} \, \rho(s)\cdot o(y) \cdot
   \ln\left(\frac{\rho(s)}{M(s,y)}\right)}
\\[+0.2em]
& = &
\displaystyle\sum_{s\in S, \, y\in O} \, \rho(s)\cdot o(y) \cdot
   \ln\left(\frac{\rho(s)\cdot (c \push \omega)(y)}
   {(\omega|_{c \pull \indic{y}})(s)}\right)
\\
& = &
\displaystyle \sum_{y\in O} o(y) \cdot \sum_{s\in S} \, \rho(s)\cdot 
   \ln\left(\frac{\rho(s)}{(\omega|_{c \pull \indic{y}})(s)}\right)
   \,-\, \sum_{y\in O} \, o(y) \cdot \ln\big((c \push \omega)(y)\big)
\\[+1.2em]
& = &
\displaystyle o \models \DKL\big(\rho, \, \omega|_{c \pull \indic{(-)}}\big)
   \,-\, o \models \ln\big((c \push \omega)(-)\big).
\end{array} \]

\noindent The latter term does not depend on $\rho$ and can be ignored
when we wish to minimise. The VFE-update from~\cite{TullKS23} is the
argmin of the first divergence-validity expression. Translated to the
current setting it occurs in
Proposition~\ref{EnergyUpdateProp}~\eqref{EnergyUpdatePropArg}, again
with the distribution $o$ replaced by $\flrn(\psi)$ and the factor $c
\pull \indic{(-)}$ by $p\in\supp(\psi)$.
\end{remark}

We like to put the expression after $\argmin$
in~\eqref{EnergyUpdatePropArgEqn} in a wider perspective.  The first
inequality below stems from~\cite[Eqn~(41)]{TullKS23}.

\begin{lemma}
\label{DKLchannelLem}
\begin{enumerate}
\item \label{DKLchannelLemGen} For a channel $c\colon Z \chanto X$ and
  two distributions $\sigma\in\Dst(Z)$ and $\rho\in\Dst(X)$ there is
  an inequality of the form:
\[ \begin{array}{rcl}
\sigma \models \DKL\big(\rho, \, c(-)\big)
& \,\geq\, &
\DKL\big(\rho, \, c\push\sigma\big).
\end{array} \]

\item \label{DKLchannelLemVFE} When we apply this inequality to
  the expression in~\eqref{EnergyUpdatePropArgEqn} for distributions
  $\omega,\rho\in\Dst(X)$ with evidence
  $\psi\in\natMlt\big(\Fact(X)\big)$, we see that the KL-divergence
  between $\rho$ and the Jeffrey update $\omega\Jupdate{\psi}$ serves
  as lower bound for the minimisation
  in~\eqref{EnergyUpdatePropArgEqn}:
\begin{equation}
\label{DKLchannelLemVFEIneq}
\begin{array}{rcl}
\flrn(\psi) \models \DKL\big(\rho, \, \omega\update{(-)}\big)
& \,\geq\, &
\DKL\big(\rho, \, \omega\Jupdate{\psi}\big).
\end{array}
\end{equation}
\end{enumerate}
\end{lemma}

\begin{proof}
\begin{enumerate}
\item We consider the minus on both sides and use Jensen's inequality in:
\[ \begin{array}{rcl}
-\Big(\sigma \models \DKL\big(\rho, c(-)\big)\Big)
& = &
-\displaystyle\sum_{z\in Z} \sigma(z)\cdot \sum_{x\in X} \rho(x)\cdot
   \ln\left(\frac{\rho(x)}{c(z)(x)}\right)
\\
& = &
\displaystyle\sum_{x\in X} \rho(x) \cdot \sum_{z\in Z} \sigma(z) \cdot
   \ln\left(\frac{c(z)(x)}{\rho(x)}\right)
\\
& \leq &
\displaystyle\sum_{x\in X} \rho(x) \cdot 
   \ln\left(\sum_{z\in Z} \sigma(z) \cdot \frac{c(z)(x)}{\rho(x)}\right)
\\[+1.4em]
& = &
\displaystyle\sum_{x\in X} \rho(x) \cdot 
   \ln\left(\frac{(c \push \sigma)(x)}{\rho(x)}\right)
\\
& = &
-\DKL\big(\rho, c\push\sigma\big).
\end{array} \]

\item By using the mapping $p \mapsto \omega\update{p}$ as function /
  channel $\Fact(X) \rightarrow \Dst(X)$, with distribution
  $\flrn(\psi) \in \Dst\big(\Fact(X)\big)$ we get:
\[ \begin{array}[b]{rcl}
\flrn(\psi) \models \DKL\big(\rho, \, \omega\update{(-)}\big)
& \,\geq\, &
\displaystyle\DKL\left(\rho, \, \sum_{p\in\supp(\psi)} \flrn(\psi)(p)\cdot
   \omega\update{p}\right)
\\[+1.4em]
& \smash{\stackrel{\eqref{EvicenceUpdateDefJeffreyEqn}}{=}} &
\DKL\big(\rho, \, \omega\Jupdate{\psi}\big).
\end{array} \eqno{\QEDbox} \]
\end{enumerate}
\end{proof}

Thus, we see that the VFE update, as minimiser of the left-hand-side
in~\eqref{DKLchannelLemVFEIneq} --- see
Proposition~\ref{EnergyUpdateProp}~\eqref{EnergyUpdatePropArg} --- is
used to approximate the Jeffrey update $\omega\Jupdate{\psi}$, on the
right-hand-side in~\eqref{DKLchannelLemVFEIneq}, via
KL-divergence. This VFE update may approximate the Jeffrey update, but
it does not do a crucial thing that Jeffrey update does, namely
correct errors --- see below --- in the form of decrease of
KL-divergence. This is remarkable, since predictive coding is all
about error correction.

\begin{remark}
\label{ErrorCorrectionRem}
In Figure~\ref{UpdateFig} we have considered 100 different updates, in
the running medical example, for point evidence of the form
$\psi_{i,j} = i\ket{p} + j\ket{n}$, for $i,j\in\{1,\ldots,10\}$. We
re-use the test channel $t \colon D \chanto T$ and the prior
$\omega\in\Dst(D)$ from Example~\ref{MedicalChannelEx}. We abbreviate
the VFE updates as:
\[ \begin{array}{rcccccl}
\omega_{i,j}
& \coloneqq &
\omega\Fupdate{t \triplepull \psi_{i,j}}
& = &
\omega\Fupdate{i\ket{t \pull \indic{p}} + j\ket{t \pull \indic{n}}}
& = &
\omega\Fupdate{i\ket{\postest} + j\ket{\negtest}} \;\in\; \Dst(D).
\end{array} \]

\noindent The the VFE posterior disease probabilities
$\omega_{i,j}(d)$ give the following plot --- like in
Figure~\ref{UpdateFig} for Jeffrey and Pearl.
\[ \includegraphics[scale=0.5]{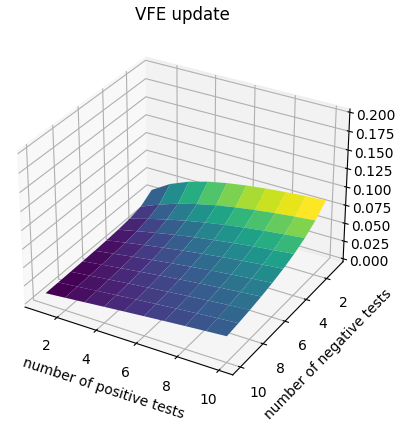} \]

\noindent This plot looks much like the one for Jeffrey in
Figure~\ref{UpdateFig}. That is not strange since VFE updating
approximates Jeffrey updating, as we have just seen.

We then check for the expected decrease of KL-divergence via
VFE update in:
\[ \begin{array}{rcl}
\DKL\big(\flrn(\psi_{i,j}), \, t \push \omega_{i,j}\big)
& \,\smash{\stackrel{?}{\leq}}\, &
\DKL\big(\flrn(\psi_{i,j}), \, t \push \omega\big).
\end{array} \]

\noindent This fails for 37 out of 100 possible cases, for
$i,j\in\{1,\ldots,10\}$. For instance, already for $i=j=1$ we have a
uniform state of the world $\flrn(\psi_{1,1}) = \frac{1}{2}\ket{p} +
\frac{1}{2}\ket{n}$ with a positive and negative test outcome that we
wish to get close to via learning. Stepwise, this works as follows.
\[ \begin{array}{lcrcl}
\mbox{prior}
& \quad &
\omega
& = &
0.05\ket{d} + 0.95\ket{\no{d}}
\\
\mbox{prior prediction}
& \quad &
t \push \omega
& = &
0.425\ket{p} + 0.575\ket{n}
\\
\mbox{VFE posterior}
& &
\omega_{1,1}
& = &
0.031\ket{d} + 0.969\ket{\no{d}}
\\
\mbox{posterior prediction}
& &
t \push \omega_{1,1}
& = &
0.416\ket{p} + 0.584\ket{n}
\\
\mbox{prior divergence}
& &
\DKL\big(\flrn(\psi_{1,1}), \, t \push\omega\big)
& = &
0.0164
\\
\mbox{posterior divergence}
& &
\DKL\big(\flrn(\psi_{1,1}), \, t \push\omega_{1,1}\big)
& = &
0.0208.
\end{array} \]

\noindent Thus, VFE updating takes us further away from the
goal distribution $\flrn(\psi_{1,1})$.

This is puzzling: the stated aim in predictive coding is error
correction in the form of KL-divergence reduction. In our medical
example this reduction fails for a bit more than one third of the
cases.  In contrast, the decrease of KL-divergence with Jeffrey
updating holds in all cases, see
Theorem~\ref{JeffreyEvicenceUpdateThm}~\eqref{JeffreyEvicenceUpdateThmPoint}.
It seems that predictive coding should not use VFE updating, but
use Jeffrey updating instead, at least if it really aims to achieve
error correction as goal of updating.
\end{remark}

In the end we ask ourselves in what sense does VFE updating make us
wiser? Does it increase Jeffrey and/or Pearl validity? It turns out
that it does not increase Jeffrey validity, in general. VFE updating
does increase Pearl validity, but not as much as Pearl updating.

\begin{theorem}
\label{VFEUpdateThm}
Let $\omega\in\Dst(X)$ be distribution with evidence
$\psi\in\natMlt\big(\Fact(X)\big)$.
\begin{enumerate}
\item \label{VFEUpdateThmJeffrey} VFE updating does not increase
Jeffrey validity, in general.

\item \label{VFEUpdateThmPearl} VFE updating does increase
Pearl validity, but not as much as Pearl updating itself:
\begin{equation}
\label{VFEUpdateThmPearlEqn}
\begin{array}{rcccl}
\omega \Pmodels \psi
& \,\leq\, &
\omega\Fupdate{\psi} \Pmodels \psi
& \,\leq\, &
\omega\Pupdate{\psi} \Pmodels \psi
\end{array}
\end{equation}
\end{enumerate}
\end{theorem}

\ignore{

N = 5
X = range_sp(N)
K = 6
for i in range(100):
    w = random_distribution(X)
    p1 = random_predicate(X)
    p2 = random_predicate(X)
    p3 = random_predicate(X)
    P = [p1,p2,p3]
    SP = Space(P)
    ev = random_multiset(K)(SP)
    wF = w.Friston_update(ev)
    wP = w.Pearl_update(ev)
    wFF = wF.Friston_validity(ev)
    wPF = wP.Friston_validity(ev)
    wPP = wP.Pearl_validity(ev)
    wFP = wF.Pearl_validity(ev)
    print( wPF >= wFF and wFF >= w.Friston_validity(ev),
           wPP >= wFP and wFP >= w.Pearl_validity(ev) )

N = 5
X = range_sp(N)
K = random.randint(2,10)
for i in range(100):
    w = random_distribution(X)
    q = random_predicate(X)
    v1 = w >= q
    v2 = w / q >= q
    v3 = w / q**K >= q
    v4 = w >= q**K
    v5 = w / q >= q**K
    v6 = w / q**K >= q**K
    print( v1 <= v2 and v2 <= v3, v4 <= v5 and v5 <= v6 )

}

\begin{proof}
\begin{enumerate}
\item We can reformulate the counterexample from
  Remark~\ref{ErrorCorrectionRem}. For evidence $\chi =
  1\bigket{\postest} + 1\bigket{\negtest}$ with one positive and one
  negative test, and prior $\omega = 0.05\ket{d} + 0.95\ket{\no{d}}$
  we get:
\[ \begin{array}{rcccccl}
\omega\Jmodels{\chi}
& \approx &
0.489
& > &
0.486
& \approx &
\omega\Fupdate{\chi}\Jmodels{\chi}.
\end{array} \]

\item We start with a number of auxiliary results (for $K\geq 1$).
\begin{enumerate}[label=(\roman*), ref=\roman*]
\item \label{VFEUpdateThmA} $\big(\evidand \flrn(\psi)\big)^{K} =
  \evidand\psi$, where $K = \|\psi\|$;

\item \label{VFEUpdateThmB} $\omega\update{q} \models q \,\leq\,
  \omega\update{q^K} \models q$;

\item \label{VFEUpdateThmC} $\omega \models q^{K} \,\leq\,
  \omega\update{q} \models q^{K}$;

\item \label{VFEUpdateThmD} $\omega\update{q} \models q^{K} \,\leq\,
  \omega\update{q^{K}} \models q^{K}$.
\end{enumerate}

\noindent The first inequality in~\eqref{VFEUpdateThmPearlEqn} is then
an instance of point~\eqref{VFEUpdateThmC}, for $q = \evidand
\flrn(\psi)$. The second inequality in~\eqref{VFEUpdateThmPearlEqn} is
an instance~\eqref{VFEUpdateThmD}.

\auxproof{
We use $q = \evidand \flrn(\psi)$, so that $q^{K} = \evidand\psi$.
\[ \begin{array}{rcl}
\omega \Pmodels \psi
& = &
\omega \models \evidand\psi
\\
& = &
\omega \models q^{K}
\\
& \leq &
\omega\update{q} \models q^{K}  \qquad \mbox{by~\eqref{VFEUpdateThmC}}
\\
& = &
\omega\update{\evidand\flrn(\psi)} \models \evidand{\psi}
\\
& = &
\omega\Fupdate{\psi} \Pmodels \psi
\end{array} \]

\[ \begin{array}{rcl}
\omega\Fupdate{\psi} \Pmodels \psi
& = &
\omega\update{\evidand\flrn(\psi)} \models \evidand{\psi}
\\
& = &
\omega\update{q} \models q^{K}
\\
& \leq &
\omega\update{q^K} \models q^{K}  \qquad \mbox{by~\eqref{VFEUpdateThmD}}
\\
& = &
\omega\update{\evidand\psi} \models \evidand\psi
\\
& = &
\omega\Pupdate{\psi} \Pmodels \psi
\end{array} \]
}

It thus suffices to prove these points~\eqref{VFEUpdateThmA} --
\eqref{VFEUpdateThmD}. For~\eqref{VFEUpdateThmA} we use:
\[ \begin{array}{rcl}
\big(\evidand \flrn(\psi)\big)^{K}(x)
\hspace*{\arraycolsep}=\hspace*{\arraycolsep}
\big(\evidand \flrn(\psi)\big)(x)^{K}
& = &
\displaystyle\left(\prod_{p\in\supp(\psi)} \, p(x)^{\flrn(\psi)(p)}\right)^{K}
\\[+1.4em]
& = &
\displaystyle\prod_{p\in\supp(\psi)} \, p(x)^{\flrn(\psi)(p)\cdot K}
\\[+1.2em]
& = &
\displaystyle\prod_{p\in\supp(\psi)} \, p(x)^{\psi(p)}
\hspace*{\arraycolsep}=\hspace*{\arraycolsep}
\big(\andthen \psi\big)(x).
\end{array} \]

The equation in~\eqref{VFEUpdateThmB} obviously holds for $K=1$. The
induction step works as follows.
\[ \begin{array}{rcll}
\omega\update{q^{K+1}} \models q 
\hspace*{\arraycolsep}=\hspace*{\arraycolsep}
\omega\update{q^{K}\,\andthen\, q} \models q 
& = &
\omega\update{q^{K}}\update{q} \models q \qquad
   & \mbox{by Lemma~\ref{PredicateUpdateLem}~\eqref{PredicateUpdateLemConj}}
\\[+0.2em]
& \geq &
\omega\update{q^{K}} \models q 
   & \mbox{by Theorem~\ref{UpdateGainThm}~\eqref{UpdateGainThmSingle}}
\\[+0.2em]
& \geq &
\omega\update{q} \models q
   & \mbox{by induction hypothesis.}
\end{array} \]

The case $K=1$ of~\eqref{VFEUpdateThmC} is covered by
Theorem~\ref{UpdateGainThm}~\eqref{UpdateGainThmSingle}. For the
induction step we use:
\[ \begin{array}{rcll}
\lefteqn{\omega|_{q} \models q^{K+1}}
\\[+0.2em]
& = &
\omega|_{q} \models q \andthen q^{K}
\\[+0.2em]
& = &
\big(\omega|_{q} \models q\big)\cdot
   \big(\omega|_{q}|_{q} \models q^{K}\big) \quad
   & \mbox{by Lemma~\ref{PredicateUpdateLem}~\eqref{PredicateUpdateLemProd}}
\\[+0.2em]
& \geq &
\big(\omega \models q\big)\cdot
   \big(\omega|_{q} \models q^{K}\big)
   & \mbox{by Theorem~\ref{UpdateGainThm}~\eqref{UpdateGainThmSingle} and 
   induction hypothesis}
\\[+0.2em]
& = &
\omega \models q^{K+1}
   & \mbox{by Lemma~\ref{PredicateUpdateLem}~\eqref{PredicateUpdateLemProd} again.}
\end{array} \]

The case $K=1$ trivially holds for~\eqref{VFEUpdateThmD}. Next:
\[ \begin{array}[b]{rcll}
\lefteqn{\omega\update{q^{K+1}}\models q^{K+1}}
\\[+0.2em]
& = &
\omega\update{q^{K+1}}\models q\andthen q^{K}
\\[+0.2em]
& = &
\big(\omega\update{q^{K+1}}\models q\big) \cdot
   \big(\omega\update{q^{K+1}}\update{q}\models q^{K}\big) \;\;
   & \mbox{by Lemma~\ref{PredicateUpdateLem}~\eqref{PredicateUpdateLemProd}}
\\[+0.2em]
& = &
\big(\omega\update{q}\update{q^{K}}\models q\big) \cdot
   \big(\omega\update{q}\update{q}\update{q^{K}}\models q^{K}\big)   
   & \mbox{by Lemma~\ref{PredicateUpdateLem}~\eqref{PredicateUpdateLemConj}}
\\[+0.2em]
& \geq &
\big(\omega\update{q}\update{q}\models q\big)\cdot
   \big(\omega\update{q}\update{q}\update{q}\models q^{K}\big)
   & \mbox{by~\eqref{VFEUpdateThmB} and induction hypothesis}
\\[+0.2em]
& \geq &
\big(\omega\update{q}\models q\big)\cdot
   \big(\omega\update{q}\update{q}\models q^{K}\big)
   & \mbox{by Theorem~\ref{UpdateGainThm}~\eqref{UpdateGainThmSingle} 
   and~\eqref{VFEUpdateThmC}}
\\[+0.2em]
& = &
\omega\update{q} \models q^{K+1}
   & \mbox{by Lemma~\ref{PredicateUpdateLem}~\eqref{PredicateUpdateLemProd} again.}
\end{array} \eqno{\QEDbox} \]

\auxproof{
\[ \begin{array}{rcl}
\omega\Pupdate{\psi} \Pmodels \psi
\hspace*{\arraycolsep}\smash{\stackrel{\eqref{EvicenceUpdateDefPearlEqn}}{=}}\hspace*{\arraycolsep}
\omega\update{\evidand\psi} \models \evidand{\psi}
& = &
\omega\update{(\evidand\flrn(\psi))^{K}} \models (\evidand\flrn(\psi))^{K}
\\
& \geq &
\omega\update{\evidand\flrn(\psi)} \models (\evidand\flrn(\psi))^{K}
\\
& = &
\omega\|_{\psi} \models (\evidand\flrn(\psi))^{K}
\\
& \geq &
\omega \models (\evidand\flrn(\psi))^{K}
\hspace*{\arraycolsep}=\hspace*{\arraycolsep}
\omega \Pmodels \psi.
\end{array} \]
}
\end{enumerate}
\end{proof}

In predictive coding one chooses to approximate distributions for
efficiency reasons --- especially in the context of continuous
probability. The problem lies in the computation of the normalisation
constants.  In Jeffrey's update
rule~\eqref{EvicenceUpdateDefJeffreyEqn} one has to compute a
normalisation constant $\omega\models p$ for each factor
$p\in\supp(\psi)$ in the evidence. In the Pearl and VFE update
rules~\eqref{EvicenceUpdateDefPearlEqn}, \eqref{EnergyUpdateAltEqn}
only one normalisation constant needs to be computed.

In practice, the Pearl and VFE approaches have the problem that the
conjunctions $\evidand(\psi) = \mathop{\andthen_{p}} p^{\psi(p)}$ and
$\evidand(\flrn(\psi)) = \mathop{\andthen_{p}} p^{\flrn(\psi)(p)}$ may
lead to many multiplications of small numbers, giving outcomes that
become too small to be reliable. One can do the updates successively,
using the equation $\omega\update{p\andthen q} =
\omega\update{p}\update{q}$ from
Lemma~\ref{PredicateUpdateLem}~\eqref{PredicateUpdateLemConj}, but
then one has to do multiple normalisations. In practice, Pearl's
updates are most conveniently done in a conjugate prior situation, so
that one can directly update the parameters, without computing any
distributions.

We may cautiously conclude that the mathematical clarifications about
validity and updating in this paper have not (yet) fully clarified
which update mechanism should be used in predictive coding and that a
more detailed explanation and analysis is needed.





\section{Concluding remarks}\label{ConclusionSec}

The topic of this paper is updating of probability distributions in
the face of multiple pieces of evidence. As shown, there are several
possible approaches, associated with Jeffrey, Pearl and VFE, with
different properties. Guarantees, like validity increase or divergence
decrease, only work for the right combination of validity and
updating. Hence, awareness of the different approaches is relevant in
practical situations. As the small survey with the medical tests
shows, such awareness is not common, not even among academic
professionals working in this area.

This article synthesises an earlier line of work of the author on
probabilistic
updating~\cite{Jacobs19c,Jacobs21c,Jacobs23c,Jacobs23b,JacobsS23b},
concentrating on the approaches of Jeffrey and Pearl. These two
approaches occur in the literature, but are not always identified as
such. For instance, in~\cite{JacobsS23b} it is shown that Jeffrey's
update rule is at the heart of the Expectation Maximisation and Latent
Dirichlet Allocation algorithms.  

In contrast to the setting of this earlier work, the current paper
puts the emphasis on multiple pieces of evidence as the key
distinguishing element between Jeffrey and Pearl. At first, notably
in~\cite{Jacobs19c,Jacobs21c}, the difference between Jeffrey and
Pearl involved treating evidence as distribution or as predicate. The
current approach combines these perspectives by seeing evidence as a
multiset of predicates (or factors), which can be turned into a
distribution via frequentist learning $\flrn$. The difference between
the validity / update rules of Jeffrey and Pearl lies in the
assumption of independence (or not). Also, in earlier work the
difference between Jeffrey and Pearl was described in the setting with
a generative model, given by both a channel and a distribution on its
domain. Here, in contrast, the different forms of updating are defined
in a simpler setting with only a prior distribution and with a
multiset of evidence. The more complex generative situation with a
channel is treated here as a special case.

A point on the horizon is a symbolic logic for probabilistic reasoning
that incorporates the various forms of probabilistic updating, with
appropriate derivation rules.

\subsection*{Acknowledgements}

The author wishes to thank Dario Stein, Sean Tull and Toby St Clere
Smithe for helpful discussions on the topic of this paper.


\bibliographystyle{plainurl}


\appendix

\section{Appendix, with the proof of Theorem~\ref{UpdateGainThm}}

We first describe an auxiliary result which we shall call the
sum-increase lemma. It is a special (discrete) case of a more general
result~\cite[Thm.~2.1]{BaumPSW70}. It describes how to find increases
for sum expressions in general.

\begin{lemma}
\label{SumIncreaseLem}
Let $X,Y$ be finite sets, with a function $F\colon X \times Y
\rightarrow \nnR$. For each element $x\in X$, write $F_{1}(x)
\coloneqq \sum_{y\in Y} F(x,y)$ for the sum that we wish to
increase. Assume that there is an $x'\in X$ with:
\[ \begin{array}{rclcrcl}
x'
& \in &
\underset{z}{\argmax}\; G(x,z)
\qquad\mbox{where}\qquad
G(x,z)
& \coloneqq &
\displaystyle\sum_{y\in Y}\, F(x,y)\cdot \ln\big(F(z,y)\big).
\end{array} \]

\noindent Then $F_{1}(x') \geq F_{1}(x)$. 
\end{lemma}


\begin{proof}
Let $x'$ be an element where $G(x,-) \colon Y \rightarrow \nnR$ takes
its maximum. This $x'$ satisfies $F_{1}(x') \geq F_{1}(x)$, since:
\[ \begin{array}[b]{rcl}
\displaystyle\ln\left(\frac{F_{1}(x')}{F_{1}(x)}\right)
& = &
\displaystyle\ln\left(\displaystyle\sum_{y\in Y}\,
   \frac{F(x',y)}{F_{1}(x)}\right)
\\[+1.4em]
& = &
\displaystyle\ln\left(\displaystyle\sum_{y\in Y}\, 
   \frac{F(x,y)}{F_{1}(x)} \cdot \frac{F(x',y)}{F(x,y)}\right)
\\[+1.4em]
& \geq &
\displaystyle\sum_{y\in Y}\,
   \frac{F(x,y)}{F_{1}(x)} \cdot \ln\left(\frac{F(x',y)}{F(x,y)}\right)
   \qquad \mbox{by Jensen's inequality}
\\[+1.2em]
& = &
\displaystyle\frac{1}{F_{1}(x)} \cdot \sum_{y\in Y}\,
   F(x,y) \cdot \Big(\ln\big(F(x',y)\big) - \ln\big(F(x,y)\big)\Big)
\\[+0.9em]
& = &
\displaystyle\frac{1}{F_{1}(x)} \cdot \Big(G(x,x') - G(x,x)\Big)
\hspace*{\arraycolsep}\geq\hspace*{\arraycolsep}
0.
\end{array} \eqno{\QEDbox} \]
\end{proof}

We use this result to prove the two items of Theorem~\ref{UpdateGainThm}
separately.

\begin{proof}[of Theorem~\ref{UpdateGainThm}~\eqref{UpdateGainThmSingle}]
We start from a validity $\omega\models p$, where $\omega\in\Dst(X)$
and $p\in\Fact(X)$, which we would like to increase, by changing the
distribution $\omega$, while keeping $p$ fixed. We apply
Lemma~\ref{SumIncreaseLem} with $F \colon \Dst(X) \times X \rightarrow
\nnR$ given by $F(\omega, x) = \omega(x) \cdot p(x)$. Then
$F_{1}(\omega) = \sum_{x} F(\omega, x) = \sum_{x} \omega(x) \cdot p(x)
= \omega \models p$.

We have $G(\omega, \omega') = \sum_{x} \omega(x) \cdot p(x) \cdot
\ln\big(\omega'(x) \cdot p(x)\big)$ and we wish to find a maximum of
$G(\omega, -) \colon \Dst(X) \rightarrow \nnR$. Let $\supp(\omega) =
\{x_{1}, \ldots, x_{n}\}$; we will use variables $v_i$ for the values
$\omega'(x_{i})$ that we wish to find, by seeing $G(\omega,-)$ as a
function $\R^{n} \rightarrow \nnR$ with constraints.

These constraints are handled via the Lagrange multiplier method for
finding the maximum (see \textit{e.g.}~\cite[\S2.2]{Bishop06}). We
keep $\omega$ fixed and consider a new function $H$, also known as the
Lagrangian, with an additional parameter $\kappa$.
\[ \begin{array}{rcl}
H(\vec{v},\kappa)
& \coloneqq &
G(\omega,\vec{v}) - \kappa\cdot \big((\sum_{i}v_{i}) - 1\big)
\\
& = &
\sum_{i} \omega(x_{i}) \cdot p(x_{i}) \cdot \ln\big(v_{i} \cdot p(x_{i})\big)
   - \kappa\cdot \big((\sum_{i}v_{i}) - 1\big).
\end{array} \]

\noindent The partial derivatives of $H$ are:
\[ \begin{array}{rclcrcl}
\displaystyle\frac{\partial H}{\partial v_{i}}(\vec{v},\kappa)
& = &
\displaystyle \frac{\omega(x_{i}) \cdot p(x_{i})}{v_{i} \cdot p(x_{i})}\cdot
   p(v_{i}) - \kappa
\hspace*{\arraycolsep}=\hspace*{\arraycolsep}
\displaystyle \frac{\omega(x_{i}) \cdot p(x_{i})}{v_{i}} - \kappa
\\[+0.9em]
\displaystyle\frac{\partial H}{\partial \kappa}(\vec{v},\kappa)
& = &
1 - \sum_{i}v_{i}.
\end{array} \]

\noindent Setting all these derivatives to zero yields:
\[ \begin{array}{rcccccl}
1
& = &
\sum_{i}v_{i}
& = &
\sum_{i} \displaystyle\frac{\omega(x_{i})\cdot p(x_{i})}{\kappa}
& = &
\displaystyle\frac{\omega\models p}{\kappa}.
\end{array} \]

\noindent Hence $\kappa \,=\, \omega\models p$ and thus:
\[ \begin{array}{rcccccl}
v_{i}
& = &
\displaystyle\frac{\omega(x_{i})\cdot p(x_{i})}{\kappa}
& = &
\displaystyle\frac{\omega(x_{i})\cdot p(x_{i})}{\omega\models p}
& = &
\omega\update{p}(x_{i}).
\end{array} \]

\noindent Hence, Lemma~\ref{SumIncreaseLem} gives as `better'
distribution the updated version $\omega\update{p}$ of $\omega$, as in
Theorem~\ref{UpdateGainThm}~\eqref{UpdateGainThmSingle}, satisfying
$\big(\omega\update{p} \models p\big) \geq \big(\omega \models
p\big)$. \QED
\end{proof}

\begin{proof}[of Theorem~\ref{UpdateGainThm}~\eqref{UpdateGainThmSeq}]
Let $\omega\in\Dst(X)$ be distribution on a set $X$ and let $p_{1},
\ldots, p_{n}$ be factors on $X$, all with non-zero validity
$\omega\models p_{i}$. The claim is that the distribution $\omega' =
\sum_{i} \frac{1}{n}\cdot\omega\update{p_i}$ satisfies:
\begin{equation}
\label{MultipliedValidity}
\begin{array}{rcl}
{\displaystyle\prod}_{i}\, (\omega' \models p_{i})
& \,\geq\, &
{\displaystyle\prod}_{i}\, (\omega \models p_{i}).
\end{array}
\end{equation}

\noindent We shall prove~\eqref{MultipliedValidity} for $n=2$. The
generalisation to arbitrary $n$ works similarly, but involves much
more book-keeping of additional variables.

We use Lemma~\ref{SumIncreaseLem} with function $F \colon \Dst(X)
\times X \times X \rightarrow \nnR$ given by:
\[ \begin{array}{rcl}
F(\omega, x, y)
& \coloneqq &
\omega(x) \cdot p_{1}(x) \cdot \omega(y) \cdot p_{2}(y).
\end{array} \]

\noindent Then by distributivity of multiplication over addition:
\[ \begin{array}{rcccl}
\sum_{x,y} F(\omega, x, y)
& = &
\big(\sum_{x}\omega(x) \cdot p_{1}(x)\big) \cdot 
   \big(\sum_{y}\omega(y) \cdot p_{2}(y)\big)
& = &
(\omega\models p_{1})\cdot(\omega\models p_{2}).
\end{array} \]

\noindent Let $\supp(\omega) = \{x_{1}, \ldots, x_{n}\}$ and let the
function $H$ be given by:
\[ \begin{array}{rcl}
H(\vec{v}, \lambda)
& \coloneqq &
\sum_{i,j}\, F(\omega, x_{i}, x_{j})\cdot 
   \ln\big(v_{i}\cdot p_{1}(x_{i}) \cdot v_{j} \cdot p_{2}(x_{j})\big) 
   \;-\; \lambda\cdot \big((\sum_{i}v_{i})-1\big).
\end{array} \]

\noindent Then:
\[ \begin{array}{rclcrcl}
\displaystyle\frac{\partial H}{\partial v_{k}}(\vec{v}, \lambda)
& = &
{\displaystyle\sum}_{i}\, \displaystyle
   \frac{F(\omega, x_{k}, x_{i}) + F(\omega, x_{i}, x_{k})}{v_k} \;-\; \lambda
& \quad &
\displaystyle\frac{\partial H}{\partial \lambda}(\vec{v}, \lambda)
& = &
1 - \sum_{i}v_{i}.
\end{array} \]

\auxproof{
The situation is a bit subtle and does not immediately generalise
beyond $n=2$, but I have checked the same argument for
$n=3$. Explicitly, in this case:
\[ \begin{array}{rcl}
\displaystyle\frac{\partial H}{\partial v_{k}}(\vec{v}, \lambda)
& = &
\displaystyle\sum_{i\neq k}\, 
   F(\omega, x_{k}, x_{i}) \cdot
   \frac{p_{1}(x_{k}) \cdot v_{i} \cdot p_{2}(x_{i})}
   {v_{k}\cdot p_{1}(x_{k}) \cdot v_{i} \cdot p_{2}(x_{i})} 
\\
& & \qquad +\; \displaystyle\sum_{i\neq k}\, 
   F(\omega, x_{i}, x_{k}) \cdot
   \frac{v_{i} \cdot p_{1}(x_{i}) \cdot v_{k} \cdot p_{2}(x_{k})}
   {v_{i}\cdot p_{1}(x_{i}) \cdot v_{k} \cdot p_{2}(x_{k})} 
\\
& & \qquad +\; \displaystyle
   F(\omega, x_{i}, x_{i}) \cdot
   \frac{2v_{i} \cdot p_{1}(x_{i}) \cdot p_{2}(x_{k})}
   {v_{i}^{2} \cdot p_{1}(x_{i}) \cdot p_{2}(x_{i})} \;-\; \lambda
\\[+1em]
& = &
\displaystyle\sum_{i\neq k}\, \frac{F(\omega, x_{k}, x_{i})}{v_k} +
   \sum_{i\neq k}\, \frac{F(\omega, x_{i}, x_{k})}{v_k} +
   \frac{2F(\omega, x_{i}, x_{i})}{v_i} \;-\; \lambda
\\[+1em]
& = &
{\displaystyle\sum}_{i}\, \displaystyle
   \frac{F(\omega, x_{k}, x_{i}) + F(\omega, x_{i}, x_{k})}{v_k} \;-\; \lambda.
\end{array} \]
}

\noindent Setting these to zero gives:
\[ \begin{array}{rcccccl}
1
& = &
\sum_{k}v_{k}
& = &
\displaystyle
   \frac{\sum_{k,i} F(\omega, x_{k}, x_{i}) + F(\omega, x_{i}, x_{k})}{\lambda}
& = &
\displaystyle
   \frac{2\cdot(\omega\models p_{1})\cdot(\omega\models p_{2})}{\lambda}.
\end{array} \]

\noindent Hence $\lambda = 2\cdot(\omega\models p_{1})\cdot(\omega\models
p_{2})$ so that:
\[ \begin{array}[b]{rcl}
v_{k}
& = &
\displaystyle
   \frac{\sum_{i} F(\omega, x_{k}, x_{i}) + F(\omega, x_{i}, x_{k})}{\lambda}
\\[+0.8em]
& = &
\frac{1}{2}\cdot \displaystyle 
   \frac{\omega(x_{k})\cdot p_{1}(x_{k}) \cdot (\omega\models p_{2})}
        {(\omega\models p_{1})\cdot(\omega\models p_{2})} +
\textstyle\frac{1}{2}\cdot \displaystyle 
   \frac{(\omega\models p_{1}) \cdot \omega(x_{k})\cdot p_{2}(x_{k})}
        {(\omega\models p_{1})\cdot(\omega\models p_{2})}
\\[+0.8em]
& = &
\frac{1}{2}\cdot \displaystyle 
   \frac{\omega(x_{k})\cdot p_{1}(x_{k})}{\omega\models p_{1}} +
\textstyle\frac{1}{2}\cdot \displaystyle 
   \frac{\omega(x_{k})\cdot p_{2}(x_{k})}{\omega\models p_{2}}
\\[+0.8em]
& = &
\frac{1}{2}\cdot \omega\update{p_1}(x_{k}) + \frac{1}{2}\cdot \omega\update{p_2}(x_{k}).
\end{array} \eqno{\QEDbox} \]
\end{proof}

\end{document}